\RequirePackage{fix-cm}
\documentclass[smallextended]{svjour3}       
\smartqed  
\usepackage{stmaryrd}
\usepackage{bussproofs}
\usepackage{turnstile}
\usepackage{amssymb}
\usepackage{latexsym}
\usepackage{graphicx}
\usepackage{url}
\usepackage{amsmath}
\usepackage{listings}
\usepackage{subfig}
\usepackage{pgf}
\usepackage{tikz}
\usetikzlibrary{arrows,shapes,snakes,automata,backgrounds,petri}
\usepackage[latin1]{inputenc}
\usepackage{float}
\usepackage{amssymb}
\usepackage{wrapfig}
\usepackage{lscape}
\usepackage[counterclockwise]{rotating}
\usepackage{hyperref}

\usepackage{ifthen}
\usepackage{amssymb}
\newboolean{showcomments}
\setboolean{showcomments}{true}
\ifthenelse{\boolean{showcomments}}
  {\newcommand{\mynote}[2]{
    \fbox{\bfseries\sffamily\scriptsize#1}
    {\small$\blacktriangleright$\textsf{\emph{#2}}$\blacktriangleleft$}
   }
  }
  {\newcommand{\mynote}[2]{}
  }

\newtheorem{convention}{Convention}
\renewenvironment{proof}{\paragraph{Proof:}}{\hfill\qed}

\def \Cathoristic {Cathoristic logic}
\def \cathoristic {cathoristic logic}
\def\FOL{First-order logic}
\def\fol{first-order logic}
\newcommand{\NOVSPACEPARAGRAPH}[1]{\NI\textbf{#1.}}
\newcommand{\PARAGRAPH}[1]{\vspace{2mm}\NOVSPACEPARAGRAPH{#1}}
\newcommand{\NI}{\noindent}

\newcommand{\MODELLEQ}{\preceq}
\newcommand{\MODELEQ}{\simeq}
\newcommand{\SIM}{\preceq_{sim}}
\newcommand{\BISIM}{\sim}
\newcommand{\BIGLUB}{\bigsqcup}
\newcommand{\BIGGLB}{\bigsqcap}
\newenvironment{FIGURE}{\begin{figure}[h]\rule{\linewidth}{0.5pt}
}{\rule{\linewidth}{0.5pt}\end{figure}}
\newenvironment{SIDEWAYSFIGURE}{\begin{sidewaysfigure}[ht]\rule{\linewidth}{0.5pt}
}{\rule{\linewidth}{0.5pt}\end{sidewaysfigure}}
\newenvironment{RULES}{\[\begin{array}{c}}{\end{array}\]}
\newenvironment{GRAMMAR}{\[\begin{array}{lcl}}{\end{array}\]}
\newcommand{\VERTICAL}{\  \mid\hspace{-3.0pt}\mid \ }
\newcommand{\infer}[2]{\frac{\displaystyle{ #1 }}{\displaystyle{ #2 }}}
\newcommand{\ZEROPREMISERULE}[1]{\infer{-}{#1}}
\newcommand{\ONEPREMISERULE}[2]{\infer{#1}{#2}}
\newcommand{\TWOPREMISERULE}[3]{\infer{#1 \quad #2}{#3}}

\newcommand{\RULENAME}[1]{\textsc{#1}}
\newcommand{\SMALLRULENAME}[1]{\textsc{\small #1}}
\newcommand{\ZEROPREMISERULENAMEDRIGHT}[2]{\ZEROPREMISERULE{#1}\,\SMALLRULENAME{#2}}
\newcommand{\ONEPREMISERULENAMEDRIGHT}[3]{\ONEPREMISERULE{#1}{#2}\,\SMALLRULENAME{#3}}
\newcommand{\TWOPREMISERULENAMEDRIGHT}[4]{\TWOPREMISERULE{#1}{#2}{#3}\,\SMALLRULENAME{#4}}

\newcommand{\FV}[1]{\mathsf{fv}(#1)}
\newcommand{\THEORY}[1]{\mathsf{Th}(#1)}
\newcommand{\FORGET}[1]{\mathsf{forget}(#1)}
\newcommand{\ACTIONS}[1]{\mathsf{actions}(#1)}
\newcommand{\MAX}[1]{\mathsf{max}(#1)}
\newcommand{\MIN}[1]{\mathsf{min}(#1)}
\newcommand{\SIMPL}[1]{\mathsf{simpl}(#1)}
\newcommand{\CHAR}[1]{\mathsf{char}(#1)}
\newcommand{\MAY}[2]{\langle #1 \rangle #2}

\newcommand{\AND}{\land}
\newcommand{\BIGAND}{\bigwedge}
\newcommand{\BIGOR}{\bigvee}
\newcommand{\OR}{\lor}

\newcommand{\CAL}[1]{\mathcal{#1}}

\newcommand{\SEMB}[1]{\lbrack\!\lbrack #1 \rbrack\!\rbrack}
\newcommand{\SEMBTWO}[1]{\langle\!\langle #1 \rangle\!\rangle}

\newcommand{\TRUE}{\top}
\newcommand{\FALSE}{\bot}

\newcommand{\MMM}{\frak{M}}
\newcommand{\NNN}{\frak{N}}
\newcommand{\PPP}{\frak{P}}
\newcommand{\SSS}{\CAL{S}}
\newcommand{\VVV}{\CAL{V}}
\newcommand{\IMPLIES}{\rightarrow}

\newcommand{\RESTRICT}[2]{\mathsf{Restrict}_{#1}(#2)}
\newcommand{\ALLOWED}[3]{\mathsf{Allow}^{#1}(#2, #3)}

\newcommand{\ARROW}[3]{\mathsf{Arrow}_{#1}(#2, #3)}
\newcommand{\ARROWTWO}[4]{\mathsf{Arrow}^{#1}(#2, #3, #4)}
\newcommand{\LLL}{\mathcal{L}}
\newcommand{\RRR}{\mathcal{R}}
\newcommand{\TRANS}[1]{\stackrel{#1}{\longrightarrow}}

\lstnewenvironment{code}
    {\lstset{}%
      \csname lst@SetFirstLabel\endcsname}
    {\csname lst@SaveFirstLabel\endcsname}
\newcommand{\CASE}[1]{\underline{Case #1.}}
\newcommand{\SUBCASE}[1]{\underline{Subcase #1.}}

\def\fBang {!}
\def\fOr {\ | \ }

\newcommand{\judge}{\vdash}

\EnableBpAbbreviations

\begin{document}

\title {\Cathoristic{}{}} 
\subtitle {A modal logic of incompatible  propositions}
\author{Richard Prideaux Evans \and  Martin Berger}
\institute {Richard Prideaux Evans,
Imperial College \email{richardprideauxevans@imperial.ac.uk} \and
 Martin Berger, University of
 Sussex. \email{M.F.Berger@sussex.ac.uk}.}
\date{Received: date / Accepted: date}

\bibliographystyle{abbrv} 
\maketitle


\begin{abstract}

\Cathoristic{} is a multi-modal logic where negation is replaced by a
novel operator allowing the expression of incompatible sentences.  
We present the syntax and semantics of the logic including complete proof rules, and
establish a number of results such as compactness, a semantic
characterisation of elementary equivalence, the existence of a
quadratic-time decision procedure, and Brandom's incompatibility
semantics property.  
We demonstrate the usefulness of the logic as a language for knowledge representation.
\end{abstract}

\setcounter{tocdepth}{2}
\tableofcontents

\section{Introduction}\label{introduction}

Natural language is full of incompatible alternatives.
If Pierre is the current king of France, then nobody else can simultaneously fill that role.
A traffic light can be green, amber or red - but it cannot be more than one colour at a time.
Mutual exclusion is a natural and ubiquitous concept.

\FOL{} can represent mutually exclusive alternatives, of course.
To say that Pierre is the only king of France, we can write, following Russell:
\[
king(france, pierre) \land \forall x . (king(france, x) \rightarrow x = pierre).
\]
To say that a particular traffic light, $tl$, is red - and red is its only colour - we could write:
\[
colour(tl, red) \land \forall x . colour(tl, x) \rightarrow x = red.
\]
In this approach, incompatibility is a \emph{derived} concept, reduced to 
a combination of universal quantification and identity.  
\FOL{}, in other words, uses relatively complex machinery to express a
simple concept:
\begin{itemize}

\item Quantification's complexity comes from the
  rules governing the distinction between free
  and bound variables\footnote{Efficient handling of free/bound variables
    is an active field of research, e.g.~nominal approaches to logic
    \cite{PittsAM:nomsetnasics}.
    The problem was put in focus in recent years with the rising
     interest in the computational cost of syntax manipulation in
     languages with binders.}.

\item Identity's complexity comes from the infinite collection of axioms required to formalise the
  indiscernibility of identicals.

\end{itemize}

\NI The costs of quantification and identity, such as a larger proof
search space, have to be borne every time one expresses a sentence that excludes others - even
though incompatibility does not, prima facie, appear to have anything to do
with the free/bound variable distinction, or require the full power of 
the identity relation.

This paper introduces an alternative approach, where
exclusion is expressed directly, as a first-class concept.
 \Cathoristic{}\footnote{``Cathoristic'' comes from the Greek
  $\kappa \alpha \theta o \rho \acute{\i} \zeta \epsilon i \nu$: to impose narrow
  boundaries. We are grateful to Tim Whitmarsh for suggesting this
  word.} is the simplest logic we could find in which incompatible
statements can be expressed.  
It is a multi-modal logic, a variant of Hennessy-Milner logic,
that replaces negation with a new logical primitive
\[
   !A
\]
pronounced \emph{tantum}\footnote{``Tantum'' is Latin for ``only''.}
$A$. Here $A$ is a finite set of alternatives, and $!A$ says that the
alternatives in $A$ exhaust all possibilities.  For example:
\begin{eqnarray*}
\fBang \{green, amber, red\}
\end{eqnarray*}
states that nothing but $green$, $amber$ or $red$ is possible.  Our
logic uses modalities to state facts, for example $\MAY{amber}{}$
expresses that $amber$ is currently the case.  The power of the logic
comes from the conjunction of modalities and tantum. For example
\[
   \MAY{amber}{}\ \AND\ !\{green, amber, red\} 
\]
expresses that $amber$ is currently the case and $red$ as well as
$green$ are the only two possible alternatives to $amber$.  Any
statement that exceeds what tantum $A$ allows, like
\[
   \MAY{blue} \ \AND\ !\{green, amber, red\},
\]
is necessarily false.  When the only options are green, amber, or red,
then blue is not permissible.  Now to say that Pierre is the only king
of France, we write:
\[
\MAY{king}\MAY{france}(\MAY{pierre} \land \fBang \{pierre\}).
\]
Crucially, \cathoristic{}'s representation involves no
universal quantifier and no identity relation.  It is a purely
propositional formulation.  To say that the traffic
light is currently red, and red is its only colour, we write:
\[
\MAY{tl} \MAY{colour} (\MAY{red} \land !\{red\}).
\]
This is simpler, both in terms of representation length and
computational complexity, than the formulation in \fol{} given on the
previous page.
Properties changing over time can be expressed by adding extra
modalities that can be understood as time-stamps.  To say that that
the traffic light was red at time $t_1$ and amber at time $t_2$, we
can write:
\[
   \MAY{tl} \MAY{colour} (\MAY{t_1} (\MAY{red} \land !\{red\}) \land \MAY{t_2} (\MAY{amber} \land !\{amber\}))
\]
Change over time can be expressed in first-order logic with bounded
quantification - but modalities are succinct and avoid introducing
bound variables.

Having claimed that incompatibility is a natural logical concept, not
easily expressed in first-order logic\footnote{We will precisify this
  claim in later sections; (1) first-order logic's representation of
  incompatibility is longer in terms of formula length than
  \cathoristic{}'s (see Section \ref{incompatiblepredicatesinfol});
  and (2) logic programs in \cathoristic{} can be optimised to run
  significantly faster than their equivalent in \fol{} (see Section
  \ref{optimizingpreconditions}).}, we will now argue the following:

\begin{itemize}

\item Incompatibility is conceptually prior to negation.

\item Negation arises as the weakest form of incompatibility.

\end{itemize}

\subsection{Material incompatibility and negation}

\NI Every English speaker knows that
\begin{quote}
``Jack is male'' is incompatible with ``Jack is female''
\end{quote}

\NI But \emph{why} are these sentences incompatible? The orthodox
position is that these sentences are incompatible because of the
following general law:
\begin{quote}
If someone is male, then it is not the case that they are female
\end{quote}
Recast in first-order logic:
\[
\forall x. ( male(x) \IMPLIES \neg female(x) ).
\]

\NI In other words, according to the orthodox position, the
incompatibility between the two particular sentences depends on a
general law involving universal quantification, implication and
negation.

Brandom \cite{brandom2} follows Sellars in proposing an alternative explanation: ``Jack
is male'' is incompatible with ``Jack is female'' because ``is male''
and ``is female'' are \emph{materially incompatible} predicates.  They
claim we can understand incompatible predicates even if we do
not understand universal quantification or negation.  
Material incompatibility is conceptually prior to logical negation.

Imagine, to make this vivid, a primitive people speaking a primordial
language of atomic sentences\footnote{In this paper, we define a sentence as \emph{atomic} if it does
  not contain another sentence as a syntactic constituent.}. These people can express sentences
that \emph{are} incompatible.  But they cannot express \emph{that}
they are incompatible.  They recognise when atomic sentences are
incompatible, and see that one sentence entails another - but their
behaviour outreaches their ability to articulate it.

Over time, these people \emph{may} advance to a more sophisticated
language where incompatibilities are made explicit, using a negation
operator - but this is a later (and optional) development:
\begin{quote}
[If negation is added to the language], it lets one say that two
claims are materially incompatible:``If a monochromatic patch is red,
then it is not blue.'' That is, negation lets one make explicit in the
form of claims - something that can be said and (so) thought - a
relation that otherwise remained implicit in what one practically did,
namely treat two claims as materially
incompatible\footnote{\cite{brandom} pp.47-48}.
\end{quote}

\NI But before making this optional
explicating step, our primitive people understand incompatibility
without understanding negation.  If this picture of our primordial
language is coherent, then material incompatibility is conceptually
independent of logical negation.

Now imagine a modification of our primitive linguistic practice in
which no sentences are ever treated as incompatible.  If one person
says ``Jack is male'' and another says ``Jack is female'', nobody
counts these claims as \emph{conflicting}.  The native speakers never
disagree, back down, retract their claims, or justify them. They just
say things.  Without an understanding of incompatibility, and the
variety of behaviour that it engenders, we submit (following Brandom)
that there is insufficient richness in the linguistic practice for
their sounds to count as assertions.  Without material
incompatibility, their sounds are just \emph{barks}.
\begin{quote}

  Suppose the reporter's differential responsive dispositions to call
  things red are matched by those of a parrot trained to utter the
  same noises under the same stimulation. What practical capacities of
  the human distinguish the reporter from the parrot? What, besides
  the exercise of regular differential responsive dispositions, must
  one be able to \emph{do}, in order to count as having or grasping
  \emph{concepts}? ... To grasp or understand a concept is, according
  to Sellars, to have practical mastery over the inferences it is
  involved in... The parrot does not treat ``That's red'' as
  incompatible with ``That's green''\footnote{\cite{brandom2}
    pp.88-89, our emphasis.}.
\end{quote}

\NI If this claim is also accepted, then material incompatibility is
not just conceptually \emph{independent} of logical negation, but
conceptually \emph{prior}.  

\subsection{Negation as the minimal incompatible}

In \cite{brandom2} and \cite{brandom}, Brandom describes 
logical negation as a limiting form of material incompatibility:
\begin{quote}
Incompatible sentences are Aristotelian \emph{contraries}. A sentence
and its negation are \emph{contradictories}. What is the relation
between these? Well, the contradictory is a contrary: any sentence is
incompatible with its negation. What distinguishes the contradictory
of a sentence from all the rest of its contraries? The contradictory
is the \emph{minimal} contrary: the one that is entailed by all the
rest. Thus every contrary of ``Plane figure $f$ is a circle'' - for
instance ``$f$ is a triangle'', ``$f$ is an octagon'', and so on -
entails ``$f$ is \emph{not} a circle''.
\end{quote}

\NI If someone asserts that it is not the case that Pierre is the (only) King of France,
we have said very little.  There are so many different ways in which
it could be true:
\begin{itemize}
\item
The King of France might be Jacques
\item
The King of France might be Louis
\item
...
\item
There may be no King of France at all
\item
There may be no country denoted by the word ``France''
\end{itemize}
Each of these concrete propositions is incompatible with Pierre being the King of France.
To say ``It is not the case that the King of France is Pierre'' is just to claim that one of these indefinitely many concrete possibilities is true.
Negation is just the logically weakest form of incompatibility.

In the rest of this paper, we assume - without further argument - that material incompatibility is conceptually prior to logical negation.
We develop a simple
 modal logic to articulate Brandom's intuition: a language, without negation, in which we can nevertheless make incompatible claims.

\subsection{Inferences between atomic sentences}
\label{intrasentential}
So far, we have justified the claim that incompatibility is a
fundamental logical concept by arguing that incompatibility is
conceptually prior to negation.  Now incompatibility is an inferential
relation between \emph{atomic sentences}.  In this subsection, we
shall describe \emph{other} inferential relations between atomic
sentences - inferential relations that \fol{} cannot articulate (or
can only do so awkwardly), but that \cathoristic{} handles naturally.

The \emph{atomic sentences} of a natural language can be
characterised as the sentences which do not contain any other
sentences as constituent parts\footnote{Compare Russell \cite{russell}
  p.117: ``A sentence is of atomic form when it contains no logical
  words and no subordinate sentence''. We use a broader notion of
  atomicity by focusing solely on whether or not it contains a
  subordinate sentence, allowing logical words such as ``and'' \emph{as long
  as they are conjoining noun-phrases} and not sentences.}.  According
to this criterion, the following are atomic:

\begin{itemize}

\item Jack is male
\item Jack loves Jill
\end{itemize}

\NI The following is not atomic:

\begin{quote}
  Jack is male and Jill is female
\end{quote}

\NI because it contains the complete sentence ``Jack is male'' as a
syntactic constituent.  Note that, according to this criterion, the
following \emph{is} atomic, despite using ``and'':

\begin{quote}
  Jack loves Jill and Joan
\end{quote}

\NI Here, ``Jack loves Jill'' is not a syntactic constituent\footnote{To see that ``Jack loves Jill'' is not a constituent of ``Jack loves Jill and Joan'', observe that ``and'' conjoins constituents of the \emph{same syntactic type}. But ``Jack loves Jill'' is a sentence, while ``Joan'' is a noun. Hence the correct parsing is ``Jack (loves (Jill and Joan))'', rather than ``(Jack loves Jill) and Joan''.}.

There are many types of inferential relations between atomic
sentences of a natural language.  For example:

\begin{itemize}

\item ``Jack is male'' is incompatible with ``Jack is female''
\item ``Jack loves Jill'' implies ``Jack loves''
\item ``Jack walks slowly'' implies ``Jack walks''
\item ``Jack loves Jill and Joan'' implies ``Jack loves Jill''
\item ``Jack is wet and cold'' implies ``Jack is cold''

\end{itemize}

\NI The first of these examples involves an incompatibility relation,
while the others involve entailment relations.  A key question this
paper seeks to answer is: what is the simplest logic that can capture
these inferential relations between atomic sentences?

\subsection{Wittgenstein's vision of a logic of elementary propositions}

\NI In the \emph{Tractatus} \cite{wittgenstein-tractatus}, Wittgenstein
claims that the world is a set of atomic sentences in an idealised
logical language.  Each atomic sentence was supposed to be
\emph{logically independent} of every other, so that they could be
combined together in every possible permutation, without worrying
about their mutual compatibility.
But already there were doubts and problem cases.  He was aware that
certain  statements seemed atomic, but did not seem logically
independent:

\begin{quote}
  For two colours, e.g., to be at one place in the visual field is
  impossible, and indeed logically impossible, for it is excluded by
  the logical structure of colour. (6.3751)
\end{quote}

\NI At the time of writing the \emph{Tractatus}, he hoped that further
analysis would reveal that these statements were not really atomic.

Later, in the \emph{Philosophical Remarks} \cite{wittgenstein-remarks}, he
renounced the thesis of the logical independence of atomic
propositions.  In \S 76, talking about incompatible colour predicates,
he writes:

\begin{quote}
  That makes it look as if \emph{a construction might be possible
    within the elementary proposition}. That is to say, as if there
  were a construction in logic which didn't work by means of truth
  functions.  What's more, it also seems that these constructions have
  an effect on one proposition's following logically from another.
  For, if different degrees exclude one another it follows from the
  presence of one that the other is not present.  In that case,
  \emph{two elementary propositions can contradict one another}.
\end{quote}

\NI Here, he is clearly imagining a logical language in which there
are incompatibilities between atomic propositions. In \S 82:

\begin{quote}
  This is how it is, what I said in the Tractatus doesn't exhaust the
  grammatical rules for 'and', 'not', 'or', etc; \emph{there are rules
    for the truth functions which also deal with the elementary part
    of the proposition}.  The fact that one measurement is right
  \emph{automatically} excludes all others.
\end{quote}

\NI Wittgenstein does not, unfortunately, show us what this
language would look like.  
In this paper, we present \cathoristic{} as one way of formalising inferences
between atomic sentences.

\subsection{Outline}

\NI The rest of this paper is organised as follows: The next section
briefly recapitulates the mathematical background of our work.
Section \ref{coreEL} introduces the syntax and semantics of
\cathoristic{} with examples. Section \ref{naturalLanguageInference}
discusses how \cathoristic{} can be used to model inferences between
atomic sentences.  Section \ref{kr} describes informally how our logic
is useful as a knowledge representation language.  Section
\ref{elAndBangCore} presents core results of the paper, in particular
a semantic characterisation of elementary equivalence and a decision
procedure with quadratic time-complexity. The decision procedure has
been implemented in Haskell and is available for public use
\cite{HaskellImplementation} under a liberal open-source
license. This section also
shows that Brandom's incompatibility semantics condition holds for
\cathoristic{}.  Section \ref{elAndBangMore} presents the proof rules
for \cathoristic{} and proves completeness. Section \ref{compactness}
provides two translations from \cathoristic{} into \fol{}, and proves
compactness using one of them.  Section \ref{ELAndNegation}
investigates a variant of \cathoristic{} with an additional negation
operator, and provides a decision procedure for this extension that
has an exponential time-complexity.  Section \ref{quantifiedEL}
extends \cathoristic{} with first-order quantifiers and sketches the
translation of first-order formulae into first-order
\cathoristic{}. The conclusion surveys related work and lists open
problems.  Appendix \ref{pureModels} outlines a different approach to
giving the semantics of \cathoristic{}, including a characterisation
of the corresponding elementary equivalence. The appendix also
discusses the question of non-deterministic models. The remaining
appendices present routine proof of facts used in the main section.

A reader not interested in mathematical detail is recommended to look
only at Chapters \ref{introduction}, \ref{coreEL},
\ref{naturalLanguageInference}, \ref{kr}, the beginning of Chapter
\ref{elAndBangMore}, and the Conclusion.

\section{Mathematical preliminaries}\label{preliminaries}

\NI This section briefly surveys the mathematical background of our
paper.  A fuller account of order-theory can be found in
\cite{DaveyBA:intlatao}.  Labelled transition systems are explored in
\cite{HennessyM:Algtheop,SassoneV:modcontac} and bisimulations in
\cite{SangiorgiD:intbisac}. Finally, \cite{EndertonHB:matinttl} is one
of many books on first-order logic.

\PARAGRAPH{Order-theory}
A \emph{preorder} is a pair $(S, \sqsubseteq)$ where $S$ is a set, and
$\sqsubseteq$ is a binary relation on $S$ that is reflexive and
transitive. Let $T \subseteq S$ and $x \in S$. We say $x$ is an
\emph{upper bound} of $T$ provided $t \sqsubseteq x$ for all $t \in
T$. If in addition $x \sqsubseteq y$ for all upper bounds $y$ of $T$,
we say that $x$ is the \emph{least} upper bound of $T$.  The set of
all least upper bounds of $T$ is denoted $\BIGLUB T$.  \emph{Lower
  bounds}, \emph{greatest lower bounds} and $\BIGGLB T$ are defined
mutatis mutandis.  A \emph{partial order} is a preorder $\sqsubseteq$
that is also anti-symmetric.  A partial order $(S, \sqsubseteq)$ is a
\emph{lattice} if every pair of elements in $S$ has a least upper and
a greatest lower bound.  A lattice is a \emph{bounded lattice} if it
has top and bottom elements $\top$ and $\bot$ such that for all $x \in
S$:
\[
x \sqcap \bot = \bot \nonumber \qquad
x \sqcup \bot = x \nonumber \qquad
x \sqcap \top = x \nonumber \qquad
x \sqcup \top = \top \nonumber.
\]

\NI If $(S, \sqsubseteq)$ is a preorder, we can turn it into a
partial-order by quotienting: let $a \simeq b$ iff $a \sqsubseteq b$
as well as $b \sqsubseteq a$. Clearly $\simeq$ is an equivalence. Let
$E$ be the set of all $\simeq$-equivalence classes of $S$. We get a
canonical partial order, denoted $\sqsubseteq_E$, on $E$ by setting:
$[a]_{\simeq} \sqsubseteq_E [b]_{\simeq}$ whenever $a \sqsubseteq
b$. If all relevant upper and lower bounds exist in $(S,
\sqsubseteq)$, then $(E, \sqsubseteq_E)$ becomes a bounded lattice by
setting
\[
   [x]_{\simeq} \sqcap [y]_{\simeq} = [ x \sqcap y ]_{\simeq}
       \quad
   [x]_{\simeq} \sqcup [y]_{\simeq} = [ x \sqcup y ]_{\simeq}
       \quad
   \bot_{E} = [\bot]_{\simeq}
       \quad
   \top_{E} = [\top]_{\simeq}.
\]

\PARAGRAPH{Transition systems}
Let $\Sigma$ be a set of \emph{actions}.  A \emph{labelled transition
  system over $\Sigma$} is a pair $(\mathcal{S}, \rightarrow)$ where
$\mathcal{S}$ is a set of \emph{states} and $\rightarrow \subseteq
\mathcal{S} \times \Sigma \times \mathcal{S}$ is the \emph{transition
  relation}.  We write $x \xrightarrow{a} y$ to abbreviate $(x,a,y)
\in \rightarrow$. We let $s, t, w, w', x, y, z, ...$ range over
states, $a, a', b, ...$ range over actions and $\LLL, \LLL', ...$
range over labelled transition systems. We usually speak of labelled
transition systems  when the set of actions is clear from the
context.  We say $\LLL$ is \emph{deterministic} if $x \TRANS{a} y$ and
$x \TRANS{a} z$ imply that $y = z$. Otherwise $\LLL$ is
\emph{non-deterministic}.  A labelled transition system is
\emph{finitely branching} if for each state $s$, the set $\{t\ |\ s
\TRANS{a} t\}$ is finite.

\PARAGRAPH{Simulations and bisimulations}
Given two labelled transition systems $\LLL_i = (S_i, \rightarrow_i)$
over $\Sigma$ for $i = 1, 2$, a \emph{simulation from $\LLL_1$ to
  $\LLL_2$} is a relation $\RRR \subseteq S_1 \times S_2$ such that
whenever $(s, s') \in \RRR$: if $s \TRANS{a}_1 s'$ then there exists a
transition $t \TRANS{a}_2 t'$ with $(t, t') \in \RRR$.  We write $s \SIM
t$ whenever $(s, t) \in \RRR$ for some simulation $\RRR$.  We say
$\RRR$ is a \emph{bisimulation between $\LLL_1$ and $\LLL_2$} if both,
$\RRR$ and $\RRR^{-1}$ are simulations. Here $\RRR^{-1} = \{(y,
x)\ |\ (x, y) \in \RRR\}$.  We say two states $s, s'$ are
\emph{bisimilar}, written $s \BISIM s'$ if there is a bisimulation
$\RRR$ with $(s, s') \in \RRR$.

\PARAGRAPH{First-order logic}
A \emph{many-sorted first-order signature} is specified by the
following data.  A non-empty set of \emph{sorts}, a set \emph{function
  symbols} with associated \emph{arities}, i.e.~non-empty list of
sorts $\#(f)$ for each function symbol $f$; a set of \emph{relation
  symbols} with associated \emph{arities}, i.e.~a list of sorts
$\#(R)$ for each relation symbol $R$; a set of \emph{constant symbols}
with associated \emph{arity}, i.e.~a sort $\#(c)$ for each constant
symbol $c$. We say a function symbol $f$ is \emph{$n$-ary} if $\#(f)$
has length $n+1$. Likewise, a relation symbol is \emph{$n$-ary} if
$\#(R)$ has length $n$.

Let $\SSS$ be a signature. An \emph{$\SSS$-model} $\CAL{M}$ is an
object with the following components.  For each sort $\sigma$ a set
$U_{\sigma}$ called \emph{universe} of sort $\sigma$.  The members of
$U_{\sigma}$ are called \emph{$\sigma$-elements} of $\CAL{M}$; an
element $c^\CAL{M}$ of $U_{\sigma}$ for each constant $c$ of sort
$\sigma$; a function $f^\CAL{M} : (U_{\sigma_1}\times \dots \times
U_{\sigma_n}) \rightarrow U_{\sigma}$ for each function symbol $f$ of
arity $(\sigma_1, ..., \sigma_n, \sigma)$; a relation $R^\CAL{M}
\subseteq U_{\sigma_1}\times \dots \times U_{\sigma_n}$ for each
relation symbol $R$ of arity $(\sigma_1, ..., \sigma_n)$.

Given an infinite set of variables for each sort $\sigma$, the
\emph{terms} and \emph{first-order formulae} for $\SSS$ are given by
the following grammar
\begin{GRAMMAR}
  t &\ ::=\ & x \VERTICAL c \VERTICAL f(t_1, ..., t_n) \\[1mm]
  \phi &::=& t = t' \VERTICAL R(t_1, ..., t_n) \VERTICAL \neg \phi \VERTICAL \phi \AND \psi \VERTICAL \forall x.A
\end{GRAMMAR}

\NI Here $x$ ranges over variables of all sorts, $c$ over constants,
$R$ over $n$-ary relational symbols and $f$ over $n$-ary function
symbols from $\SSS$.  Other logical constructs such as disjunction or
existential quantification are given by de Morgan duality, and truth
$\top$ is an abbreviation for $x = x$. If $\SSS$ has just a single
sort, we speak of \emph{single-sorted first-order logic} or just
\emph{first-order logic}. 

Given an $\SSS$-model $\CAL{M}$, an \emph{environment}, ranged over by
$\sigma$, is a partial function from variables to $\CAL{M}$'s
universes.  We write $x \mapsto u$ for the environment that maps $x$ to
$u$ and is undefined for all other variables. Moreover, if $\sigma, x
\mapsto u$ is the environment that is exactly like $\sigma$, except
that it also maps $x$ to $u$, assuming that $x$ is not in the domain
of $\sigma$.  The \emph{interpretation} $\SEMB{t}_{\CAL{M}, \sigma}$
of a term $t$ w.r.t. $\CAL{M}$ and $\sigma$ is given by the following
clauses, assuming that the domain of $\sigma$ contains all free
variables of $t$:
\begin{itemize}

\item $\SEMB{x}_{\CAL{M}, \sigma} = \sigma(x)$.
\item $\SEMB{c}_{\CAL{M}, \sigma} = c^{\CAL{M}}$.
\item $\SEMB{f(t_1, ..., t_n)}_{\CAL{M}, \sigma} =
  f^{\CAL{M}}(\SEMB{t_1}_{\CAL{M}, \sigma}, ..., \SEMB{t_n}_{\CAL{M},
    \sigma})$.

\end{itemize}

\NI The \emph{satisfaction relation} $\CAL{M} \models_{\sigma} \phi$
is given by the following clauses, this time assuming that the domain
of $\sigma$ contains all free variables of $\phi$:
\begin{itemize}

\item $\CAL{M} \models_{\sigma} t = t'$ iff $\SEMB{t}_{\CAL{M}, \sigma} = \SEMB{t'}_{\CAL{M}, \sigma}$.
\item $\CAL{M} \models_{\sigma} R(t_1, ..., t_n)$ iff
  $R^{\CAL{M}}(\SEMB{t_1}_{\CAL{M}, \sigma}, ..., \SEMB{t_n}_{\CAL{M},
  \sigma})$.
\item $\CAL{M} \models_{\sigma} \neg \phi$ iff $\CAL{M} \not\models_{\sigma} \phi$.
\item $\CAL{M} \models_{\sigma} \phi \AND \psi$ iff $\CAL{M} \models_{\sigma} \phi$ and $\CAL{M} \models_{\sigma} \psi$.
\item $\CAL{M} \models_{\sigma} \forall x.\phi$ iff for all $u$ in the
  universe of $\CAL{M}$ we have $\CAL{M} \models_{\sigma, x \mapsto v} \phi$.

\end{itemize}

\NI Note that if $\sigma$ and $\sigma'$ agree on the free variables of
$t$, then $\SEMB{t}_{\CAL{M}, \sigma} =\SEMB{t}_{\CAL{M},
  \sigma'}$. Likewise $\CAL{M} \models_{\sigma} \phi$ if and only iff
$\CAL{M} \models_{\sigma'} \phi$, provided $\sigma$ and $\sigma'$ agree
on the free variables of $\phi$.

The \emph{theory} of a model $\CAL{M}$, written $\THEORY{\CAL{M}}$, is
the set of all formulae made true by $\CAL{M}$, i.e.~$\THEORY{\CAL{M}}
= \{\phi\ |\ \CAL{M}\models \phi\}$. We say two models $\CAL{M}$ and
$\CAL{N}$ are \emph{elementary equivalent} if $\THEORY{\CAL{M}} =
\THEORY{\CAL{N}}$. In first-order logic $\THEORY{\CAL{M}} \subseteq
\THEORY{\CAL{N}}$ already implies that $\CAL{M}$ and $\CAL{N}$ are
elementary equivalent.

\section{\Cathoristic{}}\label{coreEL}

In this section we introduce the syntax and semantics of \cathoristic{}.

\subsection{Syntax}
\label{elsyntax}
\NI Syntactically, \cathoristic{} is a multi-modal logic with one new
operator.

\begin{definition} Let $\Sigma$ be a non-empty set of \emph{actions}.
Actions are ranged over by $a, a', a_1, b, ...$, and $A$ ranges over
finite subsets of $\Sigma$. The \emph{formulae} of \cathoristic{}, ranged over by $\phi,
\psi, \xi ...$, are given by the
following grammar.

\begin{GRAMMAR}
  \phi 
     &\quad ::= \quad & 
  \TRUE 
     \VERTICAL 
  \phi \AND \psi
     \VERTICAL 
  \MAY{a}{\phi}
     \VERTICAL 
  \fBang A 
\end{GRAMMAR}
\end{definition}

\NI The first three forms of $\phi$ are standard from Hennessy-Milner
logic \cite{HennessyM:alglawfndac}: $\TRUE$ is logical truth, $\AND$
is conjunction, and $\MAY{a}{\phi}$ means that the current state can
transition via action $a$ to a new state at which $\phi$ holds. Tantum
$A$, written $!A$, is the key novelty of \cathoristic{}.  Asserting
$!A$ means: in the current state at most the modalities $\MAY{a}{}$
that satisfy $a \in A$ are permissible.

We assume that $\MAY{a}{\phi}$ binds more tightly than conjunction, so
$\MAY{a}{\phi} \AND \psi$ is short for $(\MAY{a}{\phi}) \AND \psi$.
We often abbreviate $\MAY{a}{\TRUE}$ to $\MAY{a}{}$. We define falsity
$\FALSE$ as $!\emptyset \AND \MAY{a}{}$ where $a$ is an arbitrary
action in $\Sigma$. 
Hence, $\Sigma$ must be
non-empty. 
Note that, in the absence of negation, we cannot
readily define disjunction, implication, or $[a]$ modalities by de
Morgan duality. 

\begin{convention}
From now on we assume a fixed set $\Sigma$ of actions, except where
stated otherwise.
\end{convention}

\subsection{Semantics}

\NI The semantics of \cathoristic{} is close to Hennessy-Milner logic,
but uses deterministic transition systems augmented with labels on
states.



\begin{definition}\label{cathoristicTS}
A \emph{cathoristic transition system} is a triple $\LLL = (S,
\rightarrow, \lambda)$, where $(S, \rightarrow)$ is a deterministic
labelled transition system over $\Sigma$, and $\lambda$ is a function
from states to sets of actions (not necessarily finite), subject to
the following constraints:
\begin{itemize}

\item For all states $s \in S$ it is the case that $ \{a \fOr \exists
  t \; s \xrightarrow{a} t\} \subseteq \lambda(s)$. We call this
  condition \emph{admissibility}.

\item For all states $s \in S$, $\lambda (s)$ is either finite or
  $\Sigma$. We call this condition \emph{well-sizedness}.

\end{itemize}
\end{definition}

\NI The intended interpretation is that $\lambda(w)$ is the set of
allowed actions emanating from $w$.  The $\lambda$ function
is the semantic counterpart of the $!$ operator.  The admissibility
restriction is in place because transitions $s \TRANS{a} t$ where $a
\notin \lambda(s)$ would be saying that an $a$ action is possible at
$s$ but at the same time prohibited at $s$.
Well-sizedness is not a fundamental restriction but rather a
convenient trick. Cathoristic transition systems have two kinds
of states:

\begin{itemize}

\item States $s$ without restrictions on outgoing transitions. Those are
  labelled with $\lambda ( s) = \Sigma$.

\item States $s$ with restriction on outgoing transitions. Those are
  labelled by a finite set $\lambda ( s)$ of actions.

\end{itemize}

\NI Defining $\lambda$ on all states and not just on those with
restrictions makes some definitions and proofs slightly easier.

As with other modal logics, satisfaction of formulae is defined
relative to a particular state in the transition system, giving
rise to the following definition.

\begin{definition}
A \emph{cathoristic model}, ranged over by $\MMM, \MMM', ...$, is a
pair $(\LLL, s)$, where $\LLL$ is a cathoristic transition system $(S,
\rightarrow, \lambda)$, and $s$ is a state from $S$. We call $s$ the
\emph{start state} of the model.  An cathoristic model 
 is a \emph{tree} if the underlying transition system is a tree
whose root is the start state.
\end{definition}

\NI Satisfaction of a formula is defined relative to a cathoristic model.

\begin{definition}\label{ELsatisfaction}
The \emph{satisfaction relation} $\MMM \models \phi$ is defined
inductively by the following clauses, where we assume that $\MMM =
(\LLL, s)$ and $\LLL = (S, \rightarrow, \lambda)$.
\[
\begin{array}{lclcl}
  \MMM & \models & \top   \\
  \MMM & \models & \phi \AND \psi &\ \mbox{ iff } \ & \MMM  \models \phi \mbox { and } \MMM \models \psi  \\
  \MMM & \models & \langle a \rangle \phi & \mbox{ iff } & \text{there is transition } s \xrightarrow{a} t \mbox { such that } (\LLL, t) \models \phi  \\
  \MMM & \models & \fBang A &\mbox{ iff } & \lambda(s) \subseteq A
\end{array}
\]
\end{definition}

\NI The first three clauses are standard. The last clause enforces the
intended meaning of $!A$: the permissible modalities in the model are
\emph{at least as constrained} as required by $!A$. They may even be
more constrained if the inclusion $\lambda(s) \subseteq A$ is
proper. For infinite sets $\Sigma$ of actions, allowing $\lambda(s)$
to return arbitrary infinite sets in addition to $\sigma$ does not
make a difference because $A$ is finite by construction, so
$\lambda(s) \subseteq A$ can never hold anyway for infinite
$\lambda(s)$.

\begin{FIGURE}
\centering
\begin{tikzpicture}[node distance=1.3cm,>=stealth',bend angle=45,auto]
  \tikzstyle{place}=[circle,thick,draw=blue!75,fill=blue!20,minimum size=6mm]
  \tikzstyle{red place}=[place,draw=red!75,fill=red!20]
  \tikzstyle{transition}=[rectangle,thick,draw=black!75,
  			  fill=black!20,minimum size=4mm]
  \tikzstyle{every label}=[red]
  
  \begin{scope}
    \node [place] (w1) {$\Sigma$};
    \node [place] (e1) [below left of=w1] {$\{b,c\}$}
      edge [pre]  node[swap] {a}                 (w1);      
    \node [place] (e2) [below right of=w1] {$\emptyset$}
      edge [pre]  node[swap] {c}                 (w1);      
    \node [place] (e3) [below of=e1] {$\Sigma$}
      edge [pre]  node[swap] {b}                 (e1);      
  \end{scope}
    
\end{tikzpicture}
\caption{Example model.}\label{figure:elSmall}
\end{FIGURE}

We continue with concrete examples.  The model in Figure
\ref{figure:elSmall} satisfies all the following formulae, amongst
others.
\[
\begin{array}{lclclclcl}
\MAY{a} &\qquad&
\MAY{a} \MAY{b} &\qquad&
\MAY{a} \fBang \{b,c\} &\qquad&
\MAY{a} \fBang \{b,c,d\} &\qquad&
\MAY{c} \\[1mm]
\MAY{c} \fBang \emptyset &&
\MAY{c} \fBang \{a\} &&
\MAY{c} \fBang \{a,b\} &&
\MAY{a} \land \MAY{c} &&
\MAY{a} (\MAY{b} \land \fBang \{b,c\}
\end{array}
\]

\NI Here we assume, as we do with all subsequent figures, that the top
state is the start state.  The same model does not satisfy any of the
following formulae.
\[
\MAY{b} \qquad
\fBang \{a\} \qquad
\fBang \{a, c\} \qquad
\MAY{a} \fBang \{b\} \qquad
\MAY{a} \MAY{c} \qquad
\MAY{a} \MAY{b} \fBang \{c\} 
\]

\NI Figure \ref{threemodels} shows various models of $\MAY{a} \MAY{b}$
and Figure \ref{more models} shows one model that does, and one that
does not, satisfy the formula $\fBang \{a,b\}$.  Both models validate
$!\{a, b, c\}$.

\Cathoristic{} does not have the operators $\neg, \lor, $ or
$\IMPLIES$.  This has the following two significant consequences.
First, every satisfiable formula has a unique (up to isomorphism)
simplest model.  In Figure \ref{threemodels}, the left model is the
unique simplest model satisfying$\MAY{a} \MAY{b}$.  We will clarify
below that model simplicity is closely related to the process
theoretic concept of similarity, and use the existence of unique
simplest models in our quadratic-time decision procedure.

\begin{FIGURE}
\centering
\begin{tikzpicture}[node distance=1.3cm,>=stealth',bend angle=45,auto]
  \tikzstyle{place}=[circle,thick,draw=blue!75,fill=blue!20,minimum size=6mm]
  \tikzstyle{red place}=[place,draw=red!75,fill=red!20]
  \tikzstyle{transition}=[rectangle,thick,draw=black!75,
  			  fill=black!20,minimum size=4mm]
  \tikzstyle{every label}=[red]
  \begin{scope}[xshift=0cm]
    \node [place] (w1) {$\Sigma$};
    \node [place] (e1) [below of=w1] {$\Sigma$}
      edge [pre]  node[swap] {a}                 (w1);
    \node [place] (e2) [below of=e1] {$\Sigma$}
      edge [pre]  node[swap] {b}                 (e1);
  \end{scope}   
  \begin{scope}[xshift=4cm]
    \node [place] (w1) {$\Sigma$};
    \node [place] (e1) [below of=w1] {$\{a,b,c\}$}
      edge [pre]  node[swap] {a}                 (w1);
    \node [place] (e2) [below of=e1] {$\Sigma$}
      edge [pre]  node[swap] {b}                 (e1);
  \end{scope}   
  \begin{scope}[xshift=8cm]
    \node [place] (w1) {$\{a\}$};
    \node [place] (e1) [below of=w1] {$\{b\}$}
      edge [pre]  node[swap] {a}                 (w1);
    \node [place] (e2) [below of=e1] {$\emptyset$}
      edge [pre]  node[swap] {b}                 (e1);
  \end{scope}   
\end{tikzpicture}
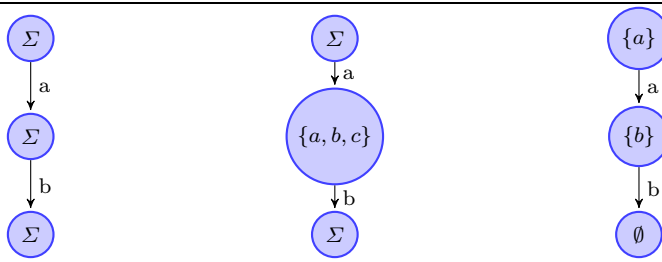
\caption{Three models of $\langle a \rangle \langle b \rangle \top$}\label{figure:elAndBang:models}
\label{threemodels}
\end{FIGURE}

\begin{FIGURE}
\centering
\begin{tikzpicture}[node distance=1.3cm,>=stealth',bend angle=45,auto]
  \tikzstyle{place}=[circle,thick,draw=blue!75,fill=blue!20,minimum size=6mm]
  \tikzstyle{red place}=[place,draw=red!75,fill=red!20]
  \tikzstyle{transition}=[rectangle,thick,draw=black!75,
  			  fill=black!20,minimum size=4mm]
  \tikzstyle{every label}=[red]
  \begin{scope}[xshift=0cm]
    \node [place] (w1) {$\{a\}$};
    \node [place] (e1) [below of=w1] {$\Sigma$}
      edge [pre]  node {a}                 (w1);
  \end{scope}   
  \begin{scope}[xshift=4cm]
    \node [place] (w1) {$\{a, b, c\}$};
    \node [place] (e1) [below of=w1] {$\{a\}$}
      edge [pre]  node[swap] {c}                 (w1);
  \end{scope}   
\end{tikzpicture}
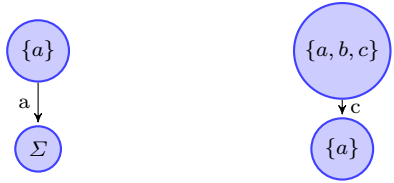
\caption{The model on the left validates $!\{a, b\}$
while the model on the right does not.}\label{figure:elAndBang:moreMdels}
\label{more models}
\end{FIGURE}

Secondly, \cathoristic\  is different from other logics in that there is an
asymmetry between tautologies and contradictories: logics with
conventional negation have an infinite number of non-trivial
tautologies, as well as an infinite number of contradictories.  In
contrast, because \cathoristic{} has no negation or disjunction
operator, it is expressively limited in the tautologies it can
express: $\top$ and conjunctions of $\top$ are its sole tautologies. On the
other hand, the tantum operator enables an infinite number of
contradictories to be expressed.  For example:
\[
   \MAY{a} \;\land\; \fBang \emptyset \qquad
   \MAY{a} \;\land\; \fBang \{b\} \qquad
   \MAY{a} \;\land\; \fBang \{b, c\} \qquad
   \MAY{b} \;\land\; \fBang \emptyset \qquad
\]


\NI Next, we present the semantic consequence relation.
\begin{definition} 
 We say the formula $\phi$ \emph{semantically implies} $\psi$, written $\phi
   \models \psi$, provided for all cathoristic models $\MMM$ if it is the
   case that $\MMM \models \phi$ then also $\MMM \models \phi$.
   We sometimes write $\models \phi$ as a shorthand for $\top \models \phi$.
\end{definition}

\NI \Cathoristic{} shares with other (multi)-modal logics the following
implications:
\begin{eqnarray*}
\MAY{a} \MAY{b} \models \MAY{a} 
 \qquad\qquad
\MAY{a} (\MAY{b} \land \MAY{c}) \models \MAY{a} \MAY{b}
\end{eqnarray*}
As \cathoristic{} is restricted to deterministic models, it also
validates the following formula:
\begin{eqnarray*}
\MAY{a} \MAY{b} \land \MAY{a} \MAY{b}  \models \MAY{a} (\MAY{b} \land \MAY{c})
\end{eqnarray*}
\Cathoristic{} also validates all implications in which the set of constraints is relaxed from left to right. For example:
\begin{eqnarray*}
\fBang \{c\} \models\ \fBang \{a, b, c\} 
 \qquad\qquad
\fBang \emptyset \models\ \fBang \{a, b\} 
\end{eqnarray*}

\section{Inferences between atomic sentences}\label{naturalLanguageInference}

\NI \Cathoristic{} arose in part as an attempt to answer the
question: what is the simplest logic that can capture inferences
between atomic sentences of natural language?  In this section, we
give examples of such inferences, and then show how \cathoristic{}
handles them.  We also compare our approach with attempts at
expressing the inferences in first-order logic.

\subsection{Intra-atomic inferences in \cathoristic{}}

\NI Natural language admits many types of inference between atomic
sentences.  First, exclusion:
\begin{quote}
``Jack is male'' is incompatible with ``Jack is female''.
\end{quote}
Second, entailment inferences from dyadic to monadic predicates:
\begin{quote}
``Jack loves Jill'' implies ``Jack loves''.
\end{quote}
Third, adverbial inferences:
\begin{quote}
``Jack walks quickly'' implies ``Jack walks''.
\end{quote}
Fourth, inferences from conjunctions of sentences to conjunctions of noun-phrases (and vice-versa):
\begin{quote}
``Jack loves Jill'' and ``Jack loves Joan'' together imply that ``Jack loves Jill and Joan''.
\end{quote}
Fifth, inferences from conjunctions of sentences to conjunction of
predicates\footnote{See \cite{sommers} p.282 for a spirited defence of
  predicate conjunction against Fregean regimentation.} (and
vice-versa):
\begin{quote}
``Jack is bruised'' and ``Jack is humiliated'' together imply that ``Jack is bruised and humiliated''.
\end{quote}

\NI They all can be handled directly and naturally in \cathoristic{}, as we
shall now show.


Incompatibility, such as that between ``Jack is male'' and ``Jack is
female'', is translated into \cathoristic{} as the pair of incompatible
sentences:
\begin{eqnarray*}
\MAY{jack} \MAY{sex} (\MAY{male} \land \fBang \{male\}) 
   \qquad\qquad
\MAY{jack} \MAY{sex} (\MAY{female} \land \fBang \{female\}).
\end{eqnarray*}

\NI \Cathoristic{} handles entailments from dyadic to monadic predicates\footnote{Although
  natural languages are full of examples of inferences from dyadic to
  monadic predicates, there are certain supposed counterexamples to
  the general rule that a dyadic predicate always implies a monadic
  one. For example, ``Jack explodes the device'' does not, on its most
  natural reading, imply that ``Jack explodes''. Our response to cases
  like this is to distinguish between two distinct monadic predicates
  $explodes_1$ and $explodes_2$:
 \begin{itemize}
 \item
 $X explodes_1$ iff $X$ is an object that undergoes an explosion
 \item
 $X explodes_2$ iff $X$ is an agent that initiates an explosion
 \end{itemize}
 Now ``Jack explodes the device'' does imply that ``Jack $explodes_2$'' but does not imply that ``Jack $explodes_1$''. 
There is no deep problem here - just another case where natural language overloads the same word in different situation to have different meanings.}.
``Jack loves Jill'' is translated into \cathoristic{} as:
\begin{eqnarray*}
   \MAY{jack} \MAY{loves} \MAY{jill}.
\end{eqnarray*}
The semantics of modalities ensures that this directly entails:
\begin{eqnarray*}
   \MAY{jack} \MAY{loves}.
\end{eqnarray*}

\NI Similarly, \cathoristic{} supports inferences from triadic to dyadic
predicates:
\begin{quote}
  ``Jack passed the biscuit to Mary'' implies ``Jack passed the biscuit''.
\end{quote}

\NI This can be expressed directly in \cathoristic{} as:
\[
   \MAY{jack} \MAY{passed} \MAY{biscuit} \MAY{to} (\MAY{mary} \land !\{mary\}) \models \MAY{jack} \MAY{passed} \MAY{biscuit}.
\]

\NI Adverbial inferences is captured in \cathoristic{} as follows.
\begin{eqnarray*}
  \MAY{jack} \MAY{walks} \MAY{quickly}
\end{eqnarray*}
entails:
\begin{eqnarray*}
  \MAY{jack} \MAY{walks}.
\end{eqnarray*}

\NI \Cathoristic{} directly supports inferences from conjunctions of
sentences to conjunctions of noun-phrases.  As our models are
deterministic, we have the general rule that $ \MAY{a} \MAY{b} \land
\MAY{a} \MAY{c} \models \MAY{a} (\MAY{b} \land \MAY{c})$ from which
it follows that
\begin{eqnarray*}
   \MAY{jack} \MAY{loves} \MAY{jill}
      \qquad\text{and}\qquad
   \MAY{jack} \MAY{loves} \MAY{joan}
\end{eqnarray*}
together imply
\begin{eqnarray*}
\MAY{jack} \MAY{loves} (\MAY{jill} \land \MAY{joan}).
\end{eqnarray*}

\NI Using the same rule, we can infer that
\begin{eqnarray*}
   \MAY{jack} \MAY{bruised} \land \MAY{jack} \MAY{humiliated}
\end{eqnarray*}

\NI together imply
\begin{eqnarray*}
\MAY{jack} (\MAY{bruised} \land \MAY{humiliated}).
\end{eqnarray*}
 
\subsection{Intra-atomic inferences in \fol{}}
Next, we look at how these inferences are handled in \fol{}.

\subsubsection{Incompatible predicates in \fol{}}\label{incompatiblepredicatesinfol}

\NI How are incompatible predicates represented in \fol{}?  Brachman
and Levesque \cite{brachman} introduce the topic by remarking:
\begin{quote}
   We would consider it quite ``obvious'' in this domain that if it
   were asserted that $John$ were a $Man$, then we should answer
   ``no'' to the query $Woman(John)$.
\end{quote}

\NI They propose adding an extra axiom to express the incompatibility:
\[
   \forall x. ( Man(x) \IMPLIES \neg Woman(x) )
\]  
 
\NI This proposal imposes a burden on the knowledge-representer: an
extra axiom must be added for every pair of incompatible predicates.
This is burdensome for large sets of incompatible predicates.  For
example, suppose there are 50 football teams, and a person can only
support one team at a time.  We would need to add ${50 \choose
  2}$ axioms, which is unwieldy.
\[
\begin{array}{l}
  \forall x.  \neg (SupportsArsenal(x) \land SupportsLiverpool(x))  \\
  \forall x.  \neg (SupportsArsenal(x) \land SupportsManUtd(x))  \\
  \forall x.  \neg (SupportsLiverpool(x) \land SupportsManUtd(x))  \\
  \qquad \qquad \qquad \vdots
\end{array}
\]

\NI Or, if we treat the football-teams as objects, and have a
two-place $Supports$ relation between people and teams, we could have:
\[
   \forall x y z . (Supports(x,y) \land y \neq z \IMPLIES \neg Supports(x,z)).
\]   

\NI If we also assume that each football team is distinct from all 
others, this certainly captures the desired uniqueness condition.  But
it does so by using relatively complex logical machinery.

\subsubsection{Inferences from dyadic to monadic predicates in \fol{}}
If we want to capture the inference from ``Jack loves Jill`` to ``Jack
loves'' in \fol{}, we can use a non-logical axiom:
\[
   \forall x. y .( Loves_2(x,y) \IMPLIES Loves_1(x))
\]

\NI We would have to add an extra  axiom like this for every
$n$-place predicate.  This is cumbersome at best.  In \cathoristic{}, by
contrast, we do not need to introduce any non-logical machinery 
to capture these inferences because they all follow from the general
rule that $\MAY{a} \MAY{b} \models \MAY{a}$.

\subsubsection{Adverbial inferences in \fol{}}

\NI How can we represent verbs in traditional \fol{} so as to
support adverbial inference?  Davidson \cite{davidson2} proposes that
every $n$-place action verb be analysed as an $n$+1-place predicate,
with an additional slot representing an event.  For example, he
analyses ``I flew my spaceship to the Morning Star'' as
\[
\exists x. ( Flew(I, MySpaceship, x) \land To(x, TheMorningStar))
\]
This implies 
\[
\exists x.  Flew(I, MySpaceship, x)
\]
This captures the inference from ``I flew my spaceship to the Morning
Star'' to ``I flew my spaceship''.

First-order logic cannot support logical inferences between atomic
sentences.  If it is going to support inferences from adverbial
sentences, it \emph{cannot} treat them as atomic and must instead
\emph{reinterpret} them as logically complex propositions.  The cost
of Davidson's proposal is that a seemingly simple sentence - such as
``Jones walks'' - turns out, on closer inspection, not to be atomic at
all - but to involve existential quantification:
\[
\exists x.  Walks(Jones, x)
\]

\NI First-order logic \emph{can} handle such inferences - but only by
reinterpreting the sentences as logically-complex compound
propositions.

\section{\Cathoristic{} as a language for knowledge representation}\label{kr}

\Cathoristic{} has been used as the representation language for a
large, complex, dynamic multi-agent simulation \cite{evans-and-short}.
This is an industrial-sized application, involving tens of thousands
of rules and facts\footnote{The application has thousands of paying
  users, and available for download on the App Store for the iPad
  \cite{Versu}.}.  In this simulation, the entire world state is stored
as a cathoristic model.
	
We found that \cathoristic\ has two distinct advantages as a language for knowledge representation. First, it is ergonomic: ubiquitous concepts (such as uniqueness) can be expressed directly.
Second, it is efficient: the tantum operator allows certain sorts of optimisation that would not otherwise be available.
We shall consider these in turn.

\subsection{Representing facts  in \cathoristic{}}

A sentence involving a one-place predicate of the form $p(a)$ is
expressed in \cathoristic{} as
\begin{eqnarray*}
   \MAY{a} \MAY{p}
\end{eqnarray*}

\NI A sentence involving a many-to-many two-place relation of the form
$r(a,b)$ is expressed in \cathoristic{} as
\begin{eqnarray*}
  \MAY{a} \MAY{r} \MAY{b}
\end{eqnarray*}

\NI But a sentence involving a many-to-one two-place relation of the
form $r(a,b)$ is expressed as:
\begin{eqnarray*}
  \MAY{a} \MAY{r} (\MAY{b} \land \fBang \{b\})
\end{eqnarray*}


\NI So, for example, to say that ``Jack likes Jill'' (where ``likes'' is,
of course, a many-many relation), we would write:
\begin{eqnarray*}
  \MAY{jack} \MAY{likes} \MAY{jill}
\end{eqnarray*}

\NI But to say that ``Jack is married to Joan''
(where``is-married-to'' is a many-one relation), we would write:
\begin{eqnarray*}
  \MAY{jack} \MAY{married} (\MAY{joan} \land \fBang \{joan\})
\end{eqnarray*}

\NI Colloquially, we might say that ``Jack is married to Joan - and
only Joan''.  Note that the relations are placed in infix position, so
that the facts about an object are ``contained'' within the object.
One reason for this particular way of structuring the data will be
explained below.
 
Consider the following facts about a gentleman named Brown:

\[
   \MAY{brown} 
   \left(
   \begin{array}{l}
     \MAY{sex} (\MAY{male} \land \fBang \{male\}) \\
        \qquad \AND \\
     \MAY{friends} (\MAY{lucy} \land \MAY{elizabeth}) 
   \end{array}
   \right)
\]

\NI All facts  starting with the prefix $\MAY{brown}$ form a
sub-tree of the entire database.  And all  facts which start with
the prefix $\MAY{brown} \MAY{friends}$ form a sub-tree of that tree.
A sub-tree can be treated as an individual via its prefix.  
A sub-tree of formulae is the \cathoristic{} equivalent of an
\emph{object} in an object-oriented programming language.

To model change over time, we assert and retract statements from the
database, using a non-monotonic update mechanism.  If a fact is
inserted into the database that involves a state-labelling restricting
the permissible  transitions emanating from that state, then all
transitions out of that state that are incompatible with the
restriction are removed.  So, for example, if the database currently
contains the fact that the traffic light is amber, and then we update
the database to assert the traffic light is red:
\[
\MAY{tl} \MAY{colour} (\MAY{red} \land !\{red\})
\]
Now the restriction on the state (that red is the only transition)
means that the previous transition from that state (the transition
labelled with amber) is automatically removed.

The tree-structure of formulae allows us to express the \emph{life-time of data} in a natural way. 
If we wish a piece of data $d$ to exist for just the duration of a proposition $t$, then we make $t$ be a sub-expression of $d$. 
For example, if we want the friendships of an agent to exist just as long as the agent, then we place the relationships inside the agent: 
\[
\MAY{brown} \MAY{friends}
\]
Now, when we remove $\MAY{brown}$ all the sub-trees, including the data about who he is friends with, will be automatically deleted as well.

Another advantage of our representation is that we get a form of \emph{automatic currying} which simplifies queries.
So if, for example, Brown is married to Elizabeth, then the database would contain 
\begin{eqnarray*}
\MAY{brown} \MAY{married} (\MAY{elizabeth} \land \fBang \{elizabeth\})
\end{eqnarray*}
In \cathoristic{}, if we want to find out whether Brown is married, we can query the sub-formula directly -  we just ask if 
\begin{eqnarray*}
\MAY{brown} \MAY{married}
\end{eqnarray*}
In \fol, if $married$ is a two-place predicate, then we need to fill in the extra argument place with a free variable - we would need to find out if there exists an $x$ such that $married(brown, x)$ - this is more cumbersome to type and slower to compute. 

\subsection{Simpler postconditions}

In this section, we contrast the representation in action languages based on \fol{}\footnote{E.g. STRIPS \cite{strips}}, with our \cathoristic{}-based representation.
Action definitions are rendered in typewriter font.

When expressing the pre- and postconditions of an action, planners
based on \fol{} have to explicitly describe the propositions that
are removed when an action is performed:
\begin{verbatim}
   action move(A, X, Y)
       preconditions
           at(A, X)
       postconditions
           add: at(A, Y) 
           remove: at(A, X)
\end{verbatim}
Here, we need to explicitly state that when $A$ moves from $X$ to $Y$, $A$ is no longer at $X$. It might seem obvious to us that if $A$ is now at $Y$, he is no longer at $X$ - but we need to explicitly tell the system this. This is unnecessary, cumbersome and error-prone. In \cathoristic{}, by contrast, the exclusion operator means we do not need to specify the facts that are no longer true:
\begin{verbatim}
   action move (A, X, Y)
       preconditions
           <A><at>(<X> /\ !{X})
       postconditions
           add: <A><at>(<Y> /\ !{Y})
\end{verbatim}
The tantum operator $!$ makes it clear that something can only be at one
place at a time, and the non-monotonic update rule described above
\emph{automatically} removes the old invalid location data.

\subsection{Using tantum $!$  to optimise preconditions}
\label{optimizingpreconditions}
Suppose, for example, we want to find all married couples who are both Welsh.
In Prolog, we might write something like:
\begin{verbatim}
   welsh_married_couple(X, Y) :-
       welsh(X),
       welsh(Y),
       spouse(X,Y).
\end{verbatim}	
Rules like this create a large search-space because we need to find all instances of $welsh(X)$ and all instances of  $welsh(Y)$ and take the cross-product \cite{smith-and-genesereth}. If there are $n$ Welsh people, then we will be searching $n^2$ instances of $(X,Y)$ substitutions.

If we express the rule in \cathoristic{}, the compiler is able to use the extra information expressed in the $!$ operator to reorder the literals to find the result significantly faster.
Assuming someone can only have a single spouse at any moment, the rule is expressed in \cathoristic{} as:
\begin{verbatim}
   welsh_married_couple(X, Y) :-
       <welsh> <X>,
       <welsh> <Y>,
       <spouse> <X> (<Y> /\ !{Y}).
\end{verbatim}	
Now the compiler is able to reorder these literals to minimise the search-space. 
It can see that, once $X$ is instantiated, the following literal can be instantiated without increasing the search-space:
\begin{verbatim}
   <spouse> <X> (<Y> /\ !{Y})
\end{verbatim}
The \emph{tantum} operator can be used by the compiler to see that there is at most one $Y$ who is the spouse of $X$.
So the compiler reorders the clauses to produce:
\begin{verbatim}
   welsh_married_couple (X, Y) :-
       <welsh> <X>,
       <spouse> <X> (<Y> /\ !{Y}),
       <welsh> <Y>.
\end{verbatim}	
Now it is able to find all results by just searching $n$ instances - a significant optimisation.
In our application, this optimisation has made a significant difference to the run-time cost of query evaluation.

\section{Semantics and Decision Procedure}\label{elAndBangCore}

In this section we provide our key semantic results.  We define a
partial ordering $\MODELLEQ$ on models, and show how the partial
ordering can be extended into a bounded lattice.  We use the bounded
lattice to construct a quadratic-time decision procedure.

\subsection{Semantic characterisation of elementary equivalence}\label{elementaryEquivalence}

Elementary equivalence induces a notion of model equivalence: two
models are elementarily equivalent exactly when they make the same
formulae true. Elementary equivalence as a concept thus relies on
\cathoristic{} even for its definition. We now present an alternative
characterisation that is purely semantic, using the concept of
(mutual) simulation from process theory. Apart from its intrinsic
interest, this characterisation will also be crucial for proving
completeness of the proof rules.

We first define a pre-order $\MODELLEQ$ on models by extending the
notion of simulation on labelled transition systems to cathoristic
models. Then we prove an alternative characterisation of $\MODELLEQ$
in terms of set-inclusion of the theories induced by models. We then
show that two models are elementarily equivalent exactly when they are
related by $\MODELLEQ$ and by $\MODELLEQ^{-1}$.

\begin{definition}
Let $\LLL_i = (S_i, \rightarrow_i, \lambda_i)$ be cathoristic transition
systems for $i = 1, 2$.  A relation $\RRR \subseteq S_1 \times S_2$ is
a \emph{simulation from $\LLL_1$ to $\LLL_2$}, provided:
\begin{itemize} 

\item $\RRR$ is a simulation on the underlying transition systems. 

\item Whenever $(x, y) \in \RRR$ then also $\lambda_1(x) \supseteq
  \lambda_2(y)$.

\end{itemize}

\NI If $\MMM_i = (\LLL_i, x_i)$ are models, we say $\RRR$ is a
\emph{simulation from $\MMM_1$ to $\MMM_2$}, provided the following hold.

\begin{itemize}

\item $\RRR$ is a simulation from $\LLL_1$ to $\LLL_2$ as cathoristic transition systems.

\item  $(x_1, x_2) \in \RRR$. 

\end{itemize}

\end{definition}

\NI Note that the only difference from the usual definition of
simulation is the additional requirement on the state labelling
functions $\lambda_1$ and $\lambda_2$.

\begin{definition}
The largest simulation from $\MMM_1$ to $\MMM_2$ is denoted $\MMM_1
\SIM \MMM_2$.  It is easy to see that $\SIM$ is itself a
simulation from $\MMM_1$ to $\MMM_2$, and the union of all such
simulations.  If $\MMM_1 \SIM \MMM_2$ we say $\MMM_2$
\emph{simulates} $\MMM_1$.

We write $\MODELEQ$ for $\SIM \cap \SIM^{-1}$. We call $\MODELEQ$ the
\emph{mutual simulation} relation.
\end{definition}

\NI We briefly discuss the relationship of $\MODELEQ$ with
bisimilarity, a notion of equality well-known from process theory and
modal logic. For non-deterministic transition systems $\MODELEQ$ is a
strictly coarser relation than bisimilarity.

\begin{definition}
We say $\RRR$ is a \emph{bisimulation} if $\RRR$ is a simulation from
$\MMM_1$ to $\MMM_2$ and $\RRR^{-1}$ is a simulation from $\MMM_2$ to
$\MMM_1$. By $\BISIM$ we denote the largest bisimulation, and we say
that $\MMM_1$ and $\MMM_2$ are \emph{bisimilar} whenever $\MMM_1
\BISIM \MMM_2$.
\end{definition}

\begin{lemma}
On cathoristic models, $\BISIM$ and $\MODELEQ$ coincide.
\end{lemma}
\begin{proof}
Straightforward from the definitions.
\end{proof}

\begin{definition}
Let $\THEORY{\MMM}$ be the \emph{theory} of $\MMM$, i.e.~the formulae
made true by $\MMM$, i.e.~$\THEORY{\MMM} = \{\phi\ |\ \MMM \models
\phi \}$.
\end{definition}

\NI We give an alternative characterisation on $\SIM^{-1}$ using
theories. In what follows, we will mostly be interested in
$\SIM^{-1}$, so we give it its own symbol.

\begin{definition}
Let $\MODELLEQ$  be short for $\SIM^{-1}$. 
\end{definition}

\NI Figure \ref{figure:leq} gives some examples of models and how they
are related by $\MODELLEQ$.

\begin{FIGURE}
\centering
\begin{tikzpicture}[node distance=1.3cm,>=stealth',bend angle=45,auto]
  \tikzstyle{place}=[circle,thick,draw=blue!75,fill=blue!20,minimum size=6mm]
  \tikzstyle{red place}=[place,draw=red!75,fill=red!20]
  \tikzstyle{transition}=[rectangle,thick,draw=black!75,
  			  fill=black!20,minimum size=4mm]
  \tikzstyle{every label}=[red]
  \begin{scope}[xshift=0cm]
    \node [place] (w1) {$\Sigma$};
    \node [place] (e1) [below left of=w1] {$\{b\}$}
      edge [pre]  node[swap] {a}                 (w1);      
    \node [place] (c) [below of=e1] {$\Sigma $}
      edge [pre]  node[swap] {b}                 (e1);      
    \node [place] (e2) [below right of=w1] {$\Sigma $}
      edge [pre]  node[swap] {c}                 (w1);      
  \end{scope}  
  
  \begin{scope}[xshift=3cm]
    \node [place] (w1) {$\Sigma $};
    \node [place] (e1) [below of=w1] {$\{b,c\}$}
      edge [pre]  node[swap] {a}                 (w1);      
    \node [place] (e2) [below of=e1] {$\Sigma $}
      edge [pre]  node[swap] {b}                 (e1);      
  \end{scope}  
  
  \begin{scope}[xshift=6cm]
    \node [place] (w1) {$\Sigma $};
    \node [place] (e1) [below of=w1] {$\Sigma $}
      edge [pre]  node[swap] {a}                 (w1);      
    \node [place] (e2) [below of=e1] {$\Sigma $}
      edge [pre]  node[swap] {b}                 (e1);      
  \end{scope}  
  
  \begin{scope}[xshift=9cm]
    \node [place] (w1) {$\Sigma $};
    \node [place] (e1) [below of=w1] {$\Sigma $}
      edge [pre]  node[swap] {a}                 (w1);      
  \end{scope}  
  
  \draw (2,0) node {$\MODELLEQ $};
  \draw (4.5,0) node {$\MODELLEQ $};
  \draw (7.5,0) node {$\MODELLEQ $};
  
\end{tikzpicture}
\caption{Examples of $\MODELLEQ $}\label{figure:leq}
\end{FIGURE}

\begin{theorem}[Characterisation of elementary equivalence]\label{theorem:completeLattice}
\begin{enumerate}

\item\label{theorem:completeLattice:1} $\MMM' \MODELLEQ \MMM$ if and
  only if $\THEORY{\MMM} \subseteq \THEORY{\MMM'}$.

\item\label{theorem:completeLattice:2} $\MMM' \MODELEQ \MMM$ if and
  only if $\THEORY{\MMM} = \THEORY{\MMM'}$.

\end{enumerate}
\end{theorem}

\begin{proof}
For (\ref{theorem:completeLattice:1}) assume $\MMM' \MODELLEQ \MMM$
and $\MMM \models \phi$.  We must show $\MMM' \models \phi$.  Let
$\MMM = (\LLL, w)$ and $\MMM' = (\LLL', w')$.  The proof proceeds by
induction on $\phi$.  The cases for $\top$ and $\land$ are trivial.
Assume $\phi = \MAY{a}\psi$ and assume $(\LLL, w) \models
\MAY{a}\psi$.  Then $w \xrightarrow{a} x$ and $(\LLL, x) \models
\psi$.  As $\MMM'$ simulates $\MMM$, there is an $x'$ such that
$(x,x') \in R$ and $w' \xrightarrow{a} x'$.  By the induction
hypothesis, $(\LLL', x') \models \psi$.  Therefore, by the semantic
clause for $\MAY{}$, $(\LLL', w') \models \MAY{a}\psi$.  Assume now
that $\phi = \; ! \; A$, for some finite $A \subseteq \Sigma$, and
that $(\LLL, w) \models \; ! \; A$.  By the semantic clause for $!$,
$\lambda(w) \subseteq A$.  Since $(\LLL', w') \MODELLEQ (\LLL, w)$, by
the definition of simulation of cathoristic transition systems, $\lambda(w)
\supseteq \lambda'(w')$.  Therefore, $\lambda'(w') \subseteq
\lambda(w) \subseteq A$.  Therefore, by the semantic clause for $!$,
$(\LLL', w') \models \; ! \; A$.

For the other direction, let $\MMM = (\LLL, w)$ and $\MMM' = (\LLL',
w')$.  Assume $\THEORY{\MMM} \subseteq \THEORY{\MMM'} $. We need to
show that $\MMM'$ simulates $\MMM$.  In other words, we need to
produce a relation $R \subseteq S \times S'$ where $S$ is the state
set of $\LLL$, $S'$ is the state set for $\LLL'$ and $(w,w') \in R$
and $R$ is a simulation from $(\LLL, w)$ to $ (\LLL', w')$.  Define $R
= \{(x,x') \; | \; \THEORY{ (\LLL, x)} \subseteq \THEORY{ (\LLL',
  x')}\}$.  Clearly, $(w,w') \in R$, as $\THEORY{(\LLL, w)} \subseteq
\THEORY{(\LLL', w')} $.  To show that $R$ is a simulation, assume $x
\xrightarrow{a} y$ in $\LLL$ and $(x,x') \in R$. 
We need to provide a
$y'$ such that $x' \xrightarrow{a} y'$ in $\LLL'$ and $(y,y') \in R$.  
Consider the formula $\MAY{a}\CHAR{(\LLL, y)}$. 
Now $x \models \MAY{a}\CHAR{(\LLL, y)}$, and since $(x,x') \in R$, $x' \models \MAY{a}\CHAR{(\LLL, y)}$.
By the semantic clause for $\MAY{a}$, if $x' \models \MAY{a}\CHAR{(\LLL, y)}$ then there is a $y'$ such that 
$y' \models \CHAR{(\LLL, y)}$.
We need to show $(y,y') \in R$, i.e. that $y \models \phi$ implies $y' \models \phi$ for all $\phi$.
Assume $y \models \phi$. 
Then by the definition of $\CHAR$, $\CHAR{(\LLL, y)} \models \phi$.
Since $y' \models \CHAR{(\LLL, y)}$, $y' \models \phi$. 
So $(y,y') \in R$, as required.

Finally,we need to show that whenever $(x,x') \in R$, then $\lambda(x)
\supseteq \lambda'(x')$.  Assume, first, that $\lambda(x)$ is finite.
Then $(\LLL, x) \models \; ! \; \lambda(x)$.  But as $(x,x') \in R$,
$\THEORY{(\LLL, x)} \subseteq \THEORY{(\LLL', x')} $, so $(\LLL', x')
\models \; ! \; \lambda(x)$.  But, by the semantic clause for $!$,
$(\LLL', x') \models \; ! \; \lambda(x)$ iff $\lambda'(x') \subseteq
\lambda(x)$.  Therefore $\lambda(x) \supseteq \lambda'(x')$.  If, on
the other hand, $\lambda(x)$ is infinite, then $\lambda(x) = \Sigma$
(because the only infinite state labelling that we allow is
$\Sigma$). Every state labelling is a subset of $\Sigma$, so here too,
$\lambda(x) = \Sigma \supseteq \lambda'(x')$.  

This establishes (\ref{theorem:completeLattice:1}), and
(\ref{theorem:completeLattice:2}) is immediate from the definitions.

\end{proof}

\NI Theorem
\ref{theorem:completeLattice}.\ref{theorem:completeLattice:1}
captures one way in which the model theory of
classical and \cathoristic{} differ. In classical logic the theory of
each model is complete, and $\THEORY{\CAL{M}} \subseteq
\THEORY{\CAL{N}}$ already implies that $\THEORY{\CAL{M}} =
\THEORY{\CAL{N}}$, i.e.~$\CAL{M}$ and $\CAL{N}$ are elementarily
equivalent. \Cathoristic{}'s lack of negation changes this drastically, and
gives $\MODELLEQ$ the structure of a non-trivial bounded lattice as we
shall demonstrate below.

Theorem \ref{theorem:completeLattice} has various
consequences.

\begin{corollary}
\begin{enumerate}

\item If $\phi$ has a model then it has a model whose underlying
  transition system is a tree, i.e.~all states except for the start state
  have exactly one predecessor, and the start state has no predecessors.

\item If $\phi$ has a model then it has a model where every state is
  reachable from the start state.

\end{enumerate}
\end{corollary}
\begin{proof}
Both are straightforward because $\MODELEQ$ is closed under
tree-unfolding as well as under removal of states not reachable from
the start state.
\end{proof}

\subsection{Quotienting models}

\NI The relation $\MODELLEQ$ is not a partial order, only a
pre-order. For example 
\[
   \MMM_1 = ( (\{w\}, \emptyset, \{w \mapsto \Sigma\}), w)
      \qquad\qquad
   \MMM_2 = ( (\{v\}, \emptyset, \{v \mapsto \Sigma\}), v)
\]

\NI are two distinct models with $\MMM_1 \MODELLEQ \MMM_2$ and $\MMM_2
\MODELLEQ \MMM_1$. The difference between the two models, the name of
the unique state, is trivial and not relevant for the formulae they
make true: $\THEORY{\MMM_1} = \THEORY{\MMM_2}$.  As briefly mentioned
in the mathematical preliminaries (Section \ref{preliminaries}), we
obtain a proper partial-order by simply quotienting models:

\[
   \MMM \MODELEQ \MMM'
      \qquad\text{iff}\qquad
   \MMM \MODELLEQ \MMM' \ \text{and}\ \MMM' \MODELLEQ \MMM
\]

\NI and then ordering the $\MODELEQ$-equivalence classes as follows:
\[
    [\MMM]_{\MODELEQ} \MODELLEQ [\MMM']_{\MODELEQ}
      \qquad\text{iff}\qquad
    \MMM \MODELLEQ \MMM'.
\]

\NI Greatest lower and least upper bounds can also be computed on
representatives:
\[
   \BIGLUB \{[\MMM]_{\MODELEQ} \ |\ \MMM \in S\ \} \quad=\quad [\BIGLUB S]_{\MODELEQ}
\]
whenever $\BIGLUB S$ exists, and likewise for the greatest lower bound.
We also define 
\[
   [\MMM]_{\MODELEQ} \models \phi 
      \qquad\text{iff}\qquad
   \MMM \models \phi.
\]

\NI It is easy to see that these definitions are independent of the
chosen representatives.

In the rest of this text we will usually be sloppy and work with
concrete models instead of $\MODELEQ$-equivalence classes of models
because the quotienting process is straightforward and not especially
interesting. We can do this because all relevant subsequent
constructions are also representation independent.  

\subsection{The bounded lattice of models}
\label{boundedlattice}
It turns out that $\MODELLEQ $ on ($\MODELEQ$-equivalence classes of)
models is not just a partial order, but a bounded lattice, except
that a bottom element is missing.

\begin{definition}
We extend the collection of models with a single \emph{bottom} element
$\bot$, where $\bot \models \phi$ for all $\phi$. We also write $\bot$
for $[\bot]_{\MODELEQ}$.  We extend the relation $\MODELLEQ $ and
stipulate that $\bot \MODELLEQ \MMM$ for all models $\MMM$.
\end{definition}

\begin{theorem}
The collection of (equivalence classes of) models together with
$\bot$, and ordered by $\MODELLEQ$ is a bounded lattice.
\end{theorem}
\begin{proof}
The topmost element in the lattice is the model $( (\{w\}, \emptyset,
\{w \mapsto \Sigma\}), w)$ (for some state $w$): this is the model
with no transitions and no transition restrictions.  The bottom
element is $\bot$.  Below, we shall define two functions $\mathsf{glb}$ and $\mathsf{lub}$, and show that they satisfy the required properties of $\sqcap$ and $\sqcup$ respectively.
\end{proof}

\Cathoristic{} is unusual in that every set of models has a unique (up to
isomorphism) least upper bound. Logics with disjunction, negation or implication do
not have this property.  

Consider propositional logic, for example.
Define a model of propositional logic as a set of atomic formulae that are set to true.
Then we have a natural ordering  on propositional logic models:
\[
\MMM \leq \MMM' \quad\text{ iff }\quad \MMM \supseteq \MMM'
\]
Consider all the possible models that satisfy $\phi \lor \psi$:
\[
\{\phi\} \qquad
\{\psi\} \qquad
\{\phi, \psi\} \qquad
\{\phi, \psi, \xi\}\qquad
\cdots
\]
This set of satisfying models has no least upper bound, since $\{\phi\} \nleq \{\psi\}$ and $\{\psi\} \nleq \{\phi\}$.
Similarly,  the set of models satisfying $\neg (\neg \phi \land \neg \psi)$ has no least upper bound.

The fact that \cathoristic{} models have unique least upper bounds is used in proving
completeness of our inference rules, and implementing the quadratic-time
decision procedure.

\subsection{Computing the least upper bound of the models that satisfy a formula}
\label{simpl}

In our decision procedure, we will see if $\phi \models \psi$ by constructing the least upper bound of the models satisfying $\phi$, and checking whether it satisfies $\psi$.

In this section, we define a function $\SIMPL{\phi}$ that satisfies the following condition:
\[
\SIMPL{\phi} = \bigsqcup \{ \MMM | \MMM \models \phi \}
\]
Define $\SIMPL{\phi}$ as:
\begin{eqnarray*}
  \SIMPL{\top} &\ = \ & ( (\{v\}, \emptyset, \{v \mapsto \Sigma\}), v)  \\
  \SIMPL{\fBang A} & = & ( (\{v\}, \emptyset, \{v \mapsto A\}), v)  \\
  \SIMPL{\phi_1 \AND \phi_2} & = & \mathsf{glb} (\SIMPL{\phi_1}, \SIMPL{\phi_2})  \\
  \SIMPL{\langle a \rangle \phi} 
     & = & ( (S \cup \{w'\}, \rightarrow \cup (w' \xrightarrow{a} w), \lambda \cup \{w' \mapsto \Sigma\}]), w')  \\
		& & \mbox{where }\SIMPL{\phi} = ( (S, \rightarrow, \lambda), w) \mbox{and } w' \mbox{ is a new state} \\
                &&  \mbox{not appearing in }S 
\end{eqnarray*}

\begin{figure}[H]
\centering
\begin{tikzpicture}[node distance=1.3cm,>=stealth',bend angle=45,auto]
  \tikzstyle{place}=[circle,thick,draw=blue!75,fill=blue!20,minimum size=6mm]
  \tikzstyle{red place}=[place,draw=red!75,fill=red!20]
  \tikzstyle{transition}=[rectangle,thick,draw=black!75,
  			  fill=black!20,minimum size=4mm]
  \tikzstyle{every label}=[red]
  \begin{scope}
    \node [place] (w1) {$\Sigma$};
    \node [place] (e1) [below of=w1] {$\Sigma$}
      edge [pre]  node[swap] {a}                 (w1);      
  \end{scope}
  \begin{scope}[xshift=4cm]
    \node [place] (w1) {$\Sigma$};
    \node [place] (e1) [below of=w1] {$\Sigma$}
      edge [pre]  node[swap] {b}                 (w1);      
  \end{scope} 
  \begin{scope}[xshift=8cm]
    \node [place] (w1) {$\Sigma$};
    \node [place] (e1) [below left of=w1] {$\Sigma$}
      edge [pre]  node[swap] {a}                 (w1);      
    \node [place] (e1) [below right of=w1] {$\Sigma$}
      edge [pre]  node[swap] {b}                 (w1);      
  \end{scope}
  \draw (2,0) node {$\sqcap$};
  \draw (6,0) node {$=$};
\end{tikzpicture}
\caption{Example of $\sqcap$.}
\end{figure}

\begin{figure}[H]
\centering
\begin{tikzpicture}[node distance=1.3cm,>=stealth',bend angle=45,auto]
  \tikzstyle{place}=[circle,thick,draw=blue!75,fill=blue!20,minimum size=6mm]
  \tikzstyle{red place}=[place,draw=red!75,fill=red!20]
  \tikzstyle{transition}=[rectangle,thick,draw=black!75,
  			  fill=black!20,minimum size=4mm]
  \tikzstyle{every label}=[red]
  \begin{scope}
    \node [place] (w1) {$\Sigma$};
    \node [place] (e1) [below of=w1] {$\{b\}$}
      edge [pre]  node[swap] {a}                 (w1);      
    \node [place] (e2) [below of=e1] {$\Sigma$}
      edge [pre]  node[swap] {b}                 (e1);      
  \end{scope}
  \begin{scope}[xshift=4cm]
    \node [place] (w1) {$\Sigma$};
    \node [place] (e1) [below of=w1] {$\Sigma$}
      edge [pre]  node[swap] {a}                 (w1);      
    \node [place] (e2) [below right of=e1] {$\Sigma$}
      edge [pre]  node[swap] {c}                 (e1);      
    \node [place] (e3) [below left of=e1] {$\Sigma$}
      edge [pre]  node[swap] {b}                 (e1);      
  \end{scope} 
  \begin{scope}[xshift=8cm]
    \node (w1) {$\bot$};
  \end{scope}
  \draw (2,0) node {$\sqcap$};
  \draw (6,0) node {$=$};
\end{tikzpicture}
\caption{Example of $\sqcap$.}
\end{figure}

\begin{figure}[H]
\centering
\begin{tikzpicture}[node distance=1.3cm,>=stealth',bend angle=45,auto]
  \tikzstyle{place}=[circle,thick,draw=blue!75,fill=blue!20,minimum size=6mm]
  \tikzstyle{red place}=[place,draw=red!75,fill=red!20]
  \tikzstyle{transition}=[rectangle,thick,draw=black!75,
  			  fill=black!20,minimum size=4mm]
  \tikzstyle{every label}=[red]
  \begin{scope}
    \node [place] (w1) {$\{a,b\}$};
    \node [place] (e1) [below of=w1] {$\Sigma$}
      edge [pre]  node[swap] {a}                 (w1);      
    \node [place] (e2) [below of=e1] {$\Sigma $}
      edge [pre]  node[swap] {b}                 (e1);      
  \end{scope}
  \begin{scope}[xshift=4cm]
    \node [place] (w1) {$\{a,c\}$};
    \node [place] (e1) [below of=w1] {$\{b,c\}$}
      edge [pre]  node[swap] {a}                 (w1);      
    \node [place] (e2) [below of=e1] {$\Sigma $}
      edge [pre]  node[swap] {c}                 (e1);      
    \node [place] (e3) [below of=e2] {$\Sigma $}
      edge [pre]  node[swap] {d}                 (e2);      
  \end{scope} 
  \begin{scope}[xshift=8cm]
    \node [place] (w1) {$\{a\}$};
    \node [place] (e1) [below of=w1] {$\{b, c\} $}
      edge [pre]  node[swap] {a}                 (w1);      
    \node [place] (e2) [below left of=e1] {$\Sigma $}
      edge [pre]  node[swap] {b}                 (e1);      
    \node [place] (e3) [below right of=e1] {$\Sigma $}
      edge [pre]  node[swap] {c}                 (e1);      
    \node [place] (e4) [below of=e3] {$\Sigma $}
      edge [pre]  node[swap] {d}                 (e3);      
  \end{scope}
  \draw (2,0) node {$\sqcap$};
  \draw (6,0) node {$=$};
\end{tikzpicture}
\caption{Example of $\sqcap$. }
\end{figure}

\begin{figure}[H]
\centering
\begin{tikzpicture}[node distance=1.3cm,>=stealth',bend angle=45,auto]
  \tikzstyle{place}=[circle,thick,draw=blue!75,fill=blue!20,minimum size=6mm]
  \tikzstyle{red place}=[place,draw=red!75,fill=red!20]
  \tikzstyle{transition}=[rectangle,thick,draw=black!75,
  			  fill=black!20,minimum size=4mm]
  \tikzstyle{every label}=[red]
  
  \begin{scope}
    \node [place] (w1) {$\Sigma$};
    \node [place] (e1) [below left of=w1] {$\{c\}$}
      edge [pre]  node[swap] {a}                 (w1);      
    \node [place] (e2) [below right of=w1] {$\Sigma$}
      edge [pre]  node[swap] {b}                 (w1);      
    \node [place] (e3) [below of=e2] {$\Sigma$}
      edge [pre]  node[swap] {d}                 (e2);      
  \end{scope}
  
  \begin{scope}[xshift=4cm]
    \node [place] (w1) {$\Sigma$};
    \node [place] (e1) [below left of=w1] {$\Sigma$}
      edge [pre]  node[swap] {a}                 (w1);      
    \node [place] (e2) [below right of=w1] {$\{d\}$}
      edge [pre]  node[swap] {b}                 (w1);      
    \node [place] (e3) [below of=e1] {$\Sigma$}
      edge [pre]  node[swap] {c}                 (e1);      
  \end{scope}

  \begin{scope}[xshift=8cm]
    \node [place] (w1) {$\Sigma$};
    \node [place] (e1) [below left of=w1] {$\{c\}$}
      edge [pre]  node[swap] {a}                 (w1);      
    \node [place] (e2) [below right of=w1] {$\{d\}$}
      edge [pre]  node[swap] {b}                 (w1);      
    \node [place] (e3) [below of=e1] {$\Sigma$}
      edge [pre]  node[swap] {c}                 (e1);      
    \node [place] (e4) [below of=e2] {$\Sigma$}
      edge [pre]  node[swap] {d}                 (e2);      
  \end{scope}
  
  \draw (2,0) node {$\sqcap$};
  \draw (6,0) node {$=$};
\end{tikzpicture}
\caption{Example of $\sqcap$.}
\end{figure}

\NI Note that, by our conventions, $\SIMPL{\phi}$ really returns a
$\MODELEQ$-equivalence class of models.

The only complex case is the clause for $\SIMPL{\phi_1 \AND \phi_2}$,
which uses the $\mathsf{glb}$ function, defined as follows, where we assume
that the sets of states in the two models are disjoint and are trees.
It is easy to see that $\SIMPL{\cdot}$ always returns tree models.

\begin{eqnarray*}
  \mathsf{glb}(\bot, \MMM)  &\ =\ &  \bot  \\
  \mathsf{glb}(\MMM, \bot)      & = &  \bot  
     \\
  \mathsf{glb}(\MMM, \MMM')
     & = & 
     \mathsf{merge}(\mathcal{L}, \mathcal{L}', \{(w,w')\}) 
     \\
     & & \text{where } \MMM = (\mathcal{L}, w) \text{ and } \MMM' = (\mathcal{L'}, w')
\end{eqnarray*}

\NI The $\mathsf{merge}$ function returns $\bot$ if either of its
arguments are $\bot$. Otherwise, it merges the two transition systems
together, given a set of state-identification pairs (a set of pairs of
states from the two transition systems that need to be identified).  The
state-identification pairs are used to make sure that the resulting
model is deterministic.

\begin{eqnarray*}
  \mathsf{merge}(\mathcal{L}, \mathcal{L}', ids) 
     & = & 
  \begin{cases}
    \bot & \text{if } \mathsf{inconsistent}(\mathcal{L}, \mathcal{L}', ids)  \\
    \mathsf{join}(\mathcal{L}, \mathcal{L}') & \text{if } ids = \emptyset  \\
    \mathsf{merge}(\mathcal{L}, \mathcal{L}'', ids')  & \text{else, where }
          \mathcal{L}'' = \mathsf{applyIds}(ids, \mathcal{L}') \\
          & \text{and } ids' = \mathsf{getIds}(\mathcal{L}, \mathcal{L}', ids)
  \end{cases}
\end{eqnarray*}

\NI The $\mathsf{inconsistent}$ predicate is true if there is pair of
states in the state-identification set such that the out-transitions
of one state is incompatible with the state-labelling on the other
state:
\begin{eqnarray*}
  \lefteqn{\mathsf{inconsistent}(\mathcal{L}, \mathcal{L}', ids)}\qquad
     \\
     &\text{ iff }& \exists (w,w') \in ids \; \text{with} \; \mathsf{out}(\mathcal{L},w) \nsubseteq \lambda'(w') \; \text{or} \; \mathsf{out}(\mathcal{L}',w') \nsubseteq \lambda(w).
\end{eqnarray*}
 
\NI Here the $\mathsf{out}$ function returns all the actions
immediately available from the given state $w$.
\[
  \mathsf{out}(((S,\rightarrow,\lambda),w)) 
     \ =\  \{ a \fOr \exists w' . w \xrightarrow{a} w'\} 
\]

\NI The $\mathsf{join}$ function takes the union of the two transition systems.
\[
   \mathsf{join}((S, \rightarrow,\lambda), (S', \rightarrow', \lambda')) 
      \quad=\quad
   (S \cup S', \rightarrow \cup \rightarrow', \lambda'')
\]

\NI Here $\lambda''$ takes the constraints arising from both, $\lambda$ and
$\lambda'$ into account: 
\[
   \lambda''(s) 
      \quad = \quad
   \begin{array}{l}
      \{\lambda(s) \cap \lambda'(s) \; | \; s \in S \cup S'\} \\ \cup \ 
      \{\lambda(s)\; |\; s\in S \setminus S' \} \\ \cup \ 
      \{\lambda(s)\; |\; s\in S' \setminus S \}. 
   \end{array}
\]

\NI The $\mathsf{applyIds}$ function applies all the
state-identification pairs as substitutions to the Labelled Transition
System:
\[
   \mathsf{applyIds}(ids, (S, \rightarrow, \lambda)) 
      \quad=\quad 
   (S', \rightarrow', \lambda')
\]

\NI where
\begin{eqnarray*}
  S' &\quad =\quad & S \; [ w / w' \; | \; (w,w') \in ids] \\
  \rightarrow' & = & \rightarrow \; [ w / w' \; | \; (w,w') \in ids] \\
  \lambda' & = & \lambda \; [ w / w' \; | \; (w,w') \in ids]
\end{eqnarray*}

\NI Here $[ w / w' \; | \; (w,w') \in ids]$ means the simultaneous
substitution of $w$ for $w'$ for all pairs $(w, w')$ in $ids$.  The
$\mathsf{getIds}$ function returns the set of extra
state-identification pairs that need to be added to respect
determinism:
\[
   \mathsf{getIds}(\mathcal{L}, \mathcal{L}', ids) 
      \quad=\quad 
   \{(x,x') \; | \; (w,w') \in ids, \exists a \; . \; w \xrightarrow{a} x, w' \xrightarrow{a} x'\}
\]

\NI The function $\SIMPL{\cdot}$ has the expected properties, as the next
lemma shows.  

\begin{lemma}
$\SIMPL{\phi} \models \phi.$
\end{lemma}
\begin{proof}
By induction on $\phi$.
\end{proof}

\begin{lemma}
$\mathsf{glb}$ as defined is the greatest lower bound
\end{lemma}
We will show that:
\begin{itemize}
\item
$\mathsf{glb}(\MMM, \MMM') \MODELLEQ \MMM$ and $\mathsf{glb}(\MMM, \MMM') \MODELLEQ \MMM'$
\item
If $\NNN \MODELLEQ \MMM$ and $\NNN \MODELLEQ \MMM'$, then $\NNN \MODELLEQ \mathsf{glb}(\MMM, \MMM')$
\end{itemize}
If $\MMM$, $ \MMM'$ or $\mathsf{glb}(\MMM, \MMM')$ are equal to $\bot$, then we just apply the rule that $\bot \MODELLEQ m$ for all models $m$. 
So let us assume that $\mathsf{consistent}(\MMM, \MMM')$ and that $\mathsf{glb}(\MMM, \MMM')  \neq \bot$.

\begin{proof}
To show $\mathsf{glb}(\MMM, \MMM') \MODELLEQ \MMM$, we need to provide a simulation $\mathcal{R}$ from $\MMM$ to  $\mathsf{glb}(\MMM, \MMM')$.
If $\MMM = ((S,\rightarrow,\lambda),w)$, then define $\mathcal{R}$ as the identity relation on the states of $S$:
\[
\mathcal{R} = \{(x,x) \; | \; x \in S\}
\]
It is straightforward to show that $\mathcal{R}$ as defined is a simulation from $\MMM$ to  $\mathsf{glb}(\MMM, \MMM')$.
If there is a transition $x \xrightarrow{a} y$ in $\MMM$, then by the construction of $\mathsf{merge}$, there is also a transition $x \xrightarrow{a} y$ in $\mathsf{glb}(\MMM, \MMM')$.
We also need to show that $\lambda_{\MMM}(x) \supseteq \lambda_{\mathsf{glb}(\MMM, \MMM')}(x)$ for all states $x$ in $\MMM$. This is immediate from the construction of $\mathsf{merge}$.

\end{proof}

\begin{proof}
To show that $\NNN \MODELLEQ \MMM$ and $\NNN \MODELLEQ \MMM'$ imply $\NNN \MODELLEQ \mathsf{glb}(\MMM, \MMM')$, assume there is a simulation $\mathcal{R}$ from $\MMM$ to $\NNN$ and there is a simulation $\mathcal{R}'$ from $\MMM'$ to $\NNN$.
We need to provide a simulation $\mathcal{R}*$ from $\mathsf{glb}(\MMM, \MMM')$ to $\NNN$.

Assume the states of $\MMM$ and $\MMM'$ are disjoint.
Define:
\[
\mathcal{R}* = \mathcal{R} \cup \mathcal{R}'
\]
We need to show that $\mathcal{R}*$ as defined is a simulation from $\mathsf{glb}(\MMM, \MMM')$ to $\NNN$.

Suppose $x \xrightarrow{a} y$ in $\mathsf{glb}(\MMM, \MMM')$ and that $(x,x_2) \in \mathcal{R} \cup \mathcal{R}'$.
We need to provide a $y_2$ such that $x_2 \xrightarrow{a} y_2$ in  $\NNN$ and $(y,y_2) \in \mathcal{R} \cup \mathcal{R}'$.
If  $x \xrightarrow{a} y$ in $\mathsf{glb}(\MMM, \MMM')$, then, from the definition of $\mathsf{merge}$, either $x \xrightarrow{a} y$ in $\MMM$ or $x \xrightarrow{a} y$ in $\MMM'$. If the former, and given that $\mathcal{R}$ is a simulation from $\MMM$ to $\NNN$, then there is a $y_2$ such that $(y,y_2) \in \mathcal{R}$ and $x_2 \xrightarrow{a} y_2$ in $\NNN$. But, if $(y,y_2) \in \mathcal{R}$, then also $(y,y_2) \in \mathcal{R} \cup \mathcal{R}'$.

Finally, we need to show that if $(x,y) \in \mathcal{R} \cup \mathcal{R}'$ then
\[
\lambda_{\mathsf{glb}(\MMM, \MMM')}(x) \supseteq \lambda_{\NNN}(y)
\]
If $(x,y) \in \mathcal{R} \cup \mathcal{R}'$ then either $(x,y) \in \mathcal{R}$ or $(x,y) \in \mathcal{R}'$.
Assume the former.
Given that $\mathcal{R}$ is a simulation from $\MMM$ to $\NNN$, we know that if $(x,y) \in \mathcal{R}$, then 
\[
\lambda_{\MMM}(x) \supseteq \lambda_{\NNN}(y)
\]
Let $\MMM = ((S,\rightarrow,\lambda),w)$.
If $x \neq w$ (i.e. $x$ is some state other than the start state), then, from the definition of $\mathsf{merge}$, $\lambda_{\mathsf{glb}(\MMM, \MMM')}(x) = \lambda_{\MMM}(x)$.
So, given $\lambda_{\MMM} \supseteq \lambda_{\NNN}(y)$, $\lambda_{\mathsf{glb}(\MMM, \MMM')}(x) \supseteq \lambda_{\NNN}(y)$.
If, on the other hand, $x = w$ (i.e. $x$ is the start state  of our cathoristic model $\MMM$), then, from the definition of $\mathsf{merge}$:
\[
\lambda_{\mathsf{glb}(\MMM, \MMM')}(w) = \lambda_{\MMM}(w) \cap \lambda_{\MMM'}(w')
\]
where $w'$ is the start state  of $\MMM'$.
In this case, given $\lambda_{\MMM}(w) \supseteq \lambda_{\NNN}(y)$ and $\lambda_{\MMM'}(w') \supseteq \lambda_{\NNN}(y)$, it follows that $\lambda_{\MMM}(w) \cap \lambda_{\MMM'}(w') \supseteq \lambda_{\NNN}(y)$ and hence 
\[
\lambda_{\mathsf{glb}(\MMM, \MMM')}(w) \supseteq \lambda_{\NNN}(y)
\]
\end{proof}

\NI Next, define the least upper bound ($\mathsf{lub}$) of two models as:
\begin{eqnarray*}
\mathsf{lub}(\MMM, \bot) & = & \MMM \\
\mathsf{lub}(\bot, \MMM) & = & \MMM \\
\mathsf{lub}((\CAL{L},w), (\CAL{L}',w')) & = & \mathsf{lub}_2(\CAL{L}, \CAL{L}', (\MMM_\top, z), \{(w, w', z)\})
\end{eqnarray*}

\NI where $\MMM_\top$ is the topmost model $(\mathcal{W}=\{z\},
\rightarrow=\emptyset, \lambda=\{z \mapsto \Sigma\})$ for some state $z$.
$\mathsf{lub}_2$ takes four parameters: the two cathoristic transition
systems $\CAL{L}$ and $\CAL{L}'$, an accumulator representing the
constructed result so far, and a list of state triples (each triple
contains one state from each of the two input models plus the state of
the accumulated result) to consider next.  It is defined as:
\begin{eqnarray*}
  \mathsf{lub}_2(\CAL{L}, \CAL{L}', \MMM, \emptyset) 
     &\quad =\quad & 
  \MMM \\
  \mathsf{lub}_2(\CAL{L}, \CAL{L}', ((\mathcal{W}, \rightarrow, \lambda), y), \{(w,w',x)\} \cup R) 
     & = & 
  \mathsf{lub}_2(\CAL{L}, \CAL{L}', ((\mathcal{W} \cup \mathcal{W}', \rightarrow \cup \rightarrow', \lambda'), y), R' \cup R\}
\end{eqnarray*}
where:
\begin{eqnarray*}
  \{(a_i, w_i, w'_i) \;|\; i = 1 ... n\} 
     &\quad =\quad & 
  \mathsf{sharedT}((\CAL{L},w), (\CAL{L}',w')) \\
  \mathcal{W}' 
     & = & 
  \{x_i \;|\; i = 1 ... n\} \\
  \rightarrow' 
     & = & 
  \{(x, a_i, x_i) \;|\; i = 1 ... n\} \\
  \lambda' 
     & = & 
  \lambda [x \mapsto \lambda(w) \cup \lambda(w)'] \\
  R' 
     & = & 
  \{(w_i, w'_i, x_i) \;|\; i = 1 ... n\}
\end{eqnarray*}

\NI Here $\lambda[x \mapsto S]$ is the state labelling function that is
exactly like $\lambda$, except that it maps $x$ to $S$.  Moreover,
$\mathsf{sharedT}$ returns the shared transitions between two models,
and is defined as:
\[
\mathsf{sharedT}(((\mathcal{W}, \rightarrow, \lambda),w) ((\mathcal{W}', \rightarrow', \lambda'),w')) =  \{(a, x, x') \;|\; w \xrightarrow{a} x \land w' \xrightarrow{a}' x'\}
\]
If $((S*,\rightarrow*,\lambda*),w*) = ((S,\rightarrow,\lambda),w) \sqcup ((S',\rightarrow',\lambda'),w')$ then define
the set $\mathsf{triples}_\mathsf{lub}$ as the set of triples $(x,x',x*) \; | \; x \in S, x' \in S', x* \in S*$ that were used during the construction of $\mathsf{lub}$ above. So $\mathsf{triples}_\mathsf{lub}$ stores the associations between states in $\MMM$, $\MMM'$ and $\MMM \sqcup \MMM'$. 

\begin{figure}[H]
\centering
\begin{tikzpicture}[node distance=1.3cm,>=stealth',bend angle=45,auto]
  \tikzstyle{place}=[circle,thick,draw=blue!75,fill=blue!20,minimum size=6mm]
  \tikzstyle{red place}=[place,draw=red!75,fill=red!20]
  \tikzstyle{transition}=[rectangle,thick,draw=black!75,
  			  fill=black!20,minimum size=4mm]
  \tikzstyle{every label}=[red]
  \begin{scope}
    \node [place] (w1) {$\Sigma$};
    \node [place] (e1) [below left of=w1] {$\Sigma$}
      edge [pre]  node[swap] {a}                 (w1);      
    \node [place] (e2) [below right of=w1] {$\Sigma$}
      edge [pre]  node[swap] {b}                 (w1);      
  \end{scope}
  \begin{scope}[xshift=4cm]
    \node [place] (w1) {$\Sigma$};
    \node [place] (e1) [below left of=w1] {$\Sigma$}
      edge [pre]  node[swap] {a}                 (w1);      
    \node [place] (e2) [below right of=w1] {$\Sigma$}
      edge [pre]  node[swap] {c}                 (w1);      
  \end{scope}
  \begin{scope}[xshift=8cm]
    \node [place] (w1) {$\Sigma$};
    \node [place] (e1) [below of=w1] {$\Sigma$}
      edge [pre]  node[swap] {a}                 (w1);      
  \end{scope}
  \draw (2,0) node {$\sqcup$};
  \draw (6,0) node {$=$};
\end{tikzpicture}
\caption{Example of $\sqcup$}
\end{figure}

\begin{figure}[H]
\centering
\begin{tikzpicture}[node distance=1.3cm,>=stealth',bend angle=45,auto]
  \tikzstyle{place}=[circle,thick,draw=blue!75,fill=blue!20,minimum size=6mm]
  \tikzstyle{red place}=[place,draw=red!75,fill=red!20]
  \tikzstyle{transition}=[rectangle,thick,draw=black!75,
  			  fill=black!20,minimum size=4mm]
  \tikzstyle{every label}=[red]
  \begin{scope}
    \node [place] (w1) {$\{a\}$};
    \node [place] (e1) [below of=w1] {$\Sigma$}
      edge [pre]  node[swap] {a}                 (w1);      
  \end{scope}
  \begin{scope}[xshift=4cm]
    \node [place] (w1) {$\{b\}$};
    \node [place] (e1) [below of=w1] {$\Sigma$}
      edge [pre]  node[swap] {b}                 (w1);      
  \end{scope}
  \begin{scope}[xshift=8cm]
    \node [place] (w1) {$\{a,b\}$};
  \end{scope}
  \draw (2,0) node {$\sqcup$};
  \draw (6,0) node {$=$};
\end{tikzpicture}
\caption{Example of $\sqcup$}
\end{figure}

\begin{figure}[H]
\centering
\begin{tikzpicture}[node distance=1.3cm,>=stealth',bend angle=45,auto]
  \tikzstyle{place}=[circle,thick,draw=blue!75,fill=blue!20,minimum size=6mm]
  \tikzstyle{red place}=[place,draw=red!75,fill=red!20]
  \tikzstyle{transition}=[rectangle,thick,draw=black!75,
  			  fill=black!20,minimum size=4mm]
  \tikzstyle{every label}=[red]
  \begin{scope}
    \node [place] (w1) {$\{a\}$};
    \node [place] (e1) [below of=w1] {$\Sigma$}
      edge [pre]  node[swap] {a}                 (w1);      
    \node [place] (e2) [below of=e1] {$\{c\}$}
      edge [pre]  node[swap] {b}                 (e1);      
  \end{scope}
  \begin{scope}[xshift=4cm]
    \node [place] (w1) {$\{a,b\}$};
    \node [place] (e1) [below of=w1] {$\{b,c\}$}
      edge [pre]  node[swap] {a}                 (w1);      
    \node [place] (e2) [below left of=e1] {$\{d\}$}
      edge [pre]  node[swap] {b}                 (e1);      
    \node [place] (e2) [below right of=e1] {$\Sigma$}
      edge [pre]  node[swap] {c}                 (e1);      
  \end{scope} 
  \begin{scope}[xshift=8cm]
    \node [place] (w1) {$\{a,b\}$};
    \node [place] (e1) [below of=w1] {$\Sigma$}
      edge [pre]  node[swap] {a}                 (w1);      
    \node [place] (e2) [below of=e1] {$\{c,d\}$}
      edge [pre]  node[swap] {b}                 (e1);      
  \end{scope}
  \draw (2,0) node {$\sqcup$};
  \draw (6,0) node {$=$};
\end{tikzpicture}
\caption{Example of $\sqcup$}
\end{figure}

\begin{lemma}
$\mathsf{lub}$ as defined is the least upper bound
\end{lemma}
We will show that:
\begin{itemize}
\item
$\MMM \MODELLEQ \mathsf{lub}(\MMM, \MMM')$ and $\MMM' \MODELLEQ \mathsf{lub}(\MMM, \MMM')$
\item
If $\MMM \MODELLEQ \NNN $ and $\MMM' \MODELLEQ \NNN $, then $\mathsf{lub}(\MMM, \MMM') \MODELLEQ \NNN$
\end{itemize}
If $\MMM$ or $ \MMM'$ are equal to $\bot$, then we just apply the rule that $\bot \MODELLEQ m$ for all models $m$. 
So let us assume that neither $\MMM$ not $\MMM'$ are $\bot$.

\begin{proof}
To see that $\MMM \MODELLEQ \mathsf{lub}(\MMM, \MMM')$, observe that, by construction of $\mathsf{lub}$ above, every transition in $\mathsf{lub}(\MMM, \MMM')$ has a matching transition in $\MMM$, and every state label in  $\mathsf{lub}(\MMM, \MMM')$ is a superset of the corresponding state label in $\MMM$.

To show that $\MMM \MODELLEQ \NNN $ and $\MMM' \MODELLEQ \NNN $ together imply $\mathsf{lub}(\MMM, \MMM') \MODELLEQ \NNN$, assume a simulation $\mathcal{R}$ from $\NNN$ to $\MMM$ and a simulation $\mathcal{R}'$ from $\NNN$ to $\MMM'$.
We need to produce a simulation relation $\mathcal{R}*$ from $\NNN$ to $\mathsf{lub}(\MMM, \MMM')$.
Define
\[
\mathcal{R}* =   \{(x, y*) \; | \; \exists y_1 . \exists y_2 . (x,y_1) \in \mathcal{R}, (x,y_2) \in \mathcal{R}', (y_1,y_2,y*) \in \mathsf{triples}_\mathsf{lub} \}
\]
In other words, $\mathsf{R}*$ contains the pairs corresponding to the pairs in both $\mathsf{R}$ and $\mathsf{R}'$.
We just need to show that $\mathsf{R}*$ as defined is a simulation from $\NNN$ to $\mathsf{lub}(\MMM, \MMM')$.
Assume $(x,x*) \in \mathsf{R}*$ and $x \xrightarrow{a} y$ in $\NNN$. 
We need to produce a $y*$ such that $(x*,y*) \in \mathsf{R}*$ and $x* \xrightarrow{a} y*$ in $\mathsf{lub}(\MMM, \MMM')$.
Given that $\mathcal{R}$ is a simulation from $\NNN$ to $\MMM$, and that  $\mathcal{R}'$ is a simulation from $\NNN$ to $\MMM'$, we know that there is a pair of states $x_1, y_1$ in $\MMM$ and a pair of states $x_2, y_2$ in $\MMM'$ such that $(x,x_1) \in \mathsf{R}$ and $(x,x_2) \in \mathsf{R}'$ and $x_1 \xrightarrow{a} y_1$ in $\MMM$ and $x_2 \xrightarrow{a} y_2$ in $\MMM'$.
Now, from the construction of $\mathsf{lub}$ above, there is a triple $(y_1, y_2, y*) \in \mathsf{triples}_\mathsf{lub}$.
Now, from the construction of $\mathsf{R}*$ above, $(x*,y*) \in \mathsf{R}*$.

Finally, we need to show that for all states $x$ and $y$, if $(x,y) \in \mathsf{R}*, \lambda_{\NNN}(x) \supseteq \lambda_{\mathsf{lub}(\MMM, \MMM')}(y)$.
Given that $\mathcal{R}$ is a simulation from $\NNN$ to $\MMM$, and that  $\mathcal{R}'$ is a simulation from $\NNN$ to $\MMM'$, we know that if $(x,y_1) \in \mathsf{R}$, then $\lambda_\NNN(x) \supseteq \lambda_\MMM(y_1)$.
Similarly, if  $(x,y_2) \in \mathsf{R}$, then $\lambda_\NNN(x) \supseteq \lambda_\MMM'(y_2)$.
Now, from the construction of $\mathsf{lub}$, $\lambda_{\mathsf{lub}(\MMM, \MMM')}(y*) = \lambda_{\MMM}(y_1) \cup \lambda_{\MMM}(y_2)$ for all triples $(y_1, y_2, y*) \in \mathsf{triples}_\mathsf{lub}$. 
So $\lambda_{\NNN}(x) \supseteq \lambda_{\mathsf{lub}(\MMM, \MMM')}(y)$, as required.
\end{proof}

\subsection{A decision procedure for \cathoristic{}}\label{decisionprocedure}

We use the semantic constructions above to provide a quadratic-time
decision procedure.  The complexity of the decision procedure is an
indication that \cathoristic{} is useful as a query language in knowledge
representation.

\Cathoristic{}'s lack of connectives for negation, disjunction or
implication is the key reason for the efficiency of the decision
procedure.  Although any satisfiable formula has an infinite number of
models, we have shown that the satisfying models form a bounded
lattice with a least upper bound.  The $\SIMPL$ function defined above
gives us the least upper bound of all models satisfying an expression.  Using this least
upper bound, we can calculate entailment by checking a \emph{single
  model}.  To decide whether $\phi \models \psi$, we use the following
algorithm.

\begin{enumerate}

\item Compute $\SIMPL{\phi}$.

\item Check if $\SIMPL{\phi} \models \psi$.

\end{enumerate}

\NI The correctness of this algorithm is given by the follow theorem.

\begin{theorem}\label{theorem:decision}
  The following are equivalent:
  \begin{enumerate}
    \item\label{theorem:decision:1} For all cathoristic models $\MMM$,
      $\MMM \models \phi$ implies $\MMM \models \psi$.
    \item\label{theorem:decision:2} $\SIMPL{\phi} \models \psi$.
  \end{enumerate}
\end{theorem}

\begin{proof}
The implication from  (\ref{theorem:decision:1}) to
(\ref{theorem:decision:2}) is trivial because $\SIMPL{\phi} \models \phi$ by construction.

For the reverse direction, we make use of the following lemma (proved in the Appendix):
\begin{lemma}
\label{lemmasimpl}
If $\MMM \models \phi$ then $\MMM \MODELLEQ \SIMPL{\phi}$.
\end{lemma}

With Lemma \ref{lemmasimpl} in hand, the proof of Theorem \ref{theorem:decision} is straightforward.
Assume $\MMM \models \phi$. We need to show
$\MMM \models \psi$.  Now if $\MMM \models \phi$ then $\MMM \MODELLEQ
\SIMPL{\phi}$ (by Lemma \ref{lemmasimpl}).  Further, if $\MMM' \models \xi $
and $\MMM \MODELLEQ \MMM'$ then $\MMM \models \xi $ by Theorem
\ref{theorem:completeLattice}. So, substituting $\psi$ for $\xi $ and
$\SIMPL{\phi}$ for $\MMM'$, it follows that $\MMM \models \psi$.
\end{proof}

Construction of $\SIMPL{\phi}$ is quadratic in the size of $\phi$,
and computing whether a model satisfies $\psi$ is of order $|\psi|
\times |\phi|$, so computing whether $\phi \models \psi$ is quadratic
time.

\subsection{Incompatibility semantics}\label{incompatibility}

\NI One of \cathoristic{}'s unusual features is that it satisfies Brandom's incompatibility semantics constraint, \emph{even though it has no negation operator}.
In this section, we formalise what
this means, and prove it.

Define the \emph{incompatibility set of $\phi$} as:
\[
\mathcal{I}(\phi) = \{
  \psi \; | \; \forall \MMM. \MMM \not \models \phi \AND \psi\}
\]

\NI The reason why Brandom introduces the incompatibility set
\footnote{Brandom \cite{brandom} defines incompatibility slightly
  differently: he defines the set of \emph{sets} of formulae which are
  incompatible with a \emph{set} of formulae.  But in \cathoristic{},
  if a set of formulae is incompatible, then there is an incompatible
  subset of that set with exactly two members.  So we can work with
  the simpler definition in the text above.} is that he
wants to use it define \emph{semantic content}:
\begin{quote}
  Here is a semantic suggestion: represent the propositional content
  expressed by a sentence with the set of sentences that express
  propositions incompatible with it\footnote{\cite{brandom} p.123.}.
\end{quote}
Now if the propositional content of a claim determines its logical consequences, and the propositional content is identified with the incompatibility set, then the incompatibility set must determine the logical consequences.
A logic satisfies Brandom's \emph{incompatibility semantics constraint} if
\[
\phi \models \psi \; \quad\mbox{ iff }\quad \; \mathcal{I}(\psi) \subseteq \mathcal{I}(\phi)
\]

\NI Not all logics satisfy this property.  Brandom has shown that
first-order logic and the modal logic S5 satisfy the incompatibility
semantics property.
Hennessy-Milner logic satisfies it, but Hennessy-Milner logic without
negation does not.  \Cathoristic{} is the simplest logic we have found
that satisfies the property.  To establish the incompatibility
semantics constraint for \cathoristic, we need to define a related
incompatibility function on models.  $\mathcal{J}(\MMM)$ is the set of
models that are incompatible with $\MMM$:
\[
\mathcal{J}(\MMM) = \{ \MMM_2 \; | \; \MMM \sqcap \MMM_2 = \bot \}
\]
We shall make use of two lemmas, proved in Appendix
\ref{app:completeness:proofs}:
\begin{lemma}
\label{inc1}
$\mbox{If }\phi \models \psi \mbox{ then } \SIMPL{\phi} \MODELLEQ \SIMPL{\psi}$
\end{lemma}
\begin{lemma}
\label{inc3}
$\mathcal{I}(\psi) \subseteq \mathcal{I}(\phi) \mbox{ implies } \mathcal{J}(\SIMPL{\psi}) \subseteq \mathcal{J}(\SIMPL{\phi})$
\end{lemma}

\begin{theorem}
\label{incompatibilitytheorem}
$\phi \models \psi \; \mbox{ iff } \; \mathcal{I}(\psi) \subseteq \mathcal{I}(\phi)$
\end{theorem}

\begin{proof}

Left to right: Assume $\phi \models \psi$ and $\xi \in \mathcal{I}(\psi)$.  
We need to show $\xi \in \mathcal{I}(\phi)$.
By the definition of $\mathcal{I}$, if $\xi \in \mathcal{I}(\psi)$ then
$\SIMPL{\xi} \sqcap \SIMPL{\psi} = \bot$.
If $\SIMPL{\xi} \sqcap \SIMPL{\psi} = \bot$, then either
\begin{itemize}
\item $\SIMPL{\xi} = \bot$
\item $\SIMPL{\psi} = \bot$
\item Neither $\SIMPL{\xi}$ nor $\SIMPL{\psi}$ are $\bot$, but $\SIMPL{\xi} \sqcap \SIMPL{\psi} = \bot$.
\end{itemize}
If $\SIMPL{\xi} = \bot$, then $\SIMPL{\xi} \sqcap \SIMPL{\phi} = \bot$ and we are done.
If $\SIMPL{\psi} = \bot$, then as $\phi \models \psi$, by Lemma \ref{inc1}, $\SIMPL{\phi} \MODELLEQ \SIMPL{\psi}$.
Now the only model that is $\MODELLEQ \bot$ is $\bot$ itself, so $\SIMPL{\phi} = \bot$. Hence $\SIMPL{\xi} \sqcap \SIMPL{\phi} = \bot$, and we are done.
The interesting case is when neither $\SIMPL{\xi}$ nor $\SIMPL{\psi}$ are $\bot$, but $\SIMPL{\xi} \sqcap \SIMPL{\psi} = \bot$.
Then (by the definition of $\mathsf{consistent}$ in Section \ref{simpl}), either $\mathsf{out}(\SIMPL{\xi}) \nsubseteq \lambda(\SIMPL{\psi})$ or $\mathsf{out}(\SIMPL{\psi}) \nsubseteq \lambda(\SIMPL{\xi})$.
In the first sub-case, if  $\mathsf{out}(\SIMPL{\xi}) \nsubseteq \lambda(\SIMPL{\psi})$, then there is some action $a$ such that $\xi \models \MAY{a} \top$ and $a \notin \lambda(\SIMPL{\psi})$.
If $a \notin \lambda(\SIMPL{\psi})$ then $\psi \models ! A$ where $a \notin A$.
Now $\phi \models \psi$, so $\phi \models ! A$.
In other words, $\phi$ also entails the $A$-restriction that rules out the $a$ transition.
So $\SIMPL{\xi} \sqcap \SIMPL{\phi} = \bot$ and $\xi \in \mathcal{I}(\phi)$.
In the second sub-case, $\mathsf{out}(\SIMPL{\psi}) \nsubseteq \lambda(\SIMPL{\xi})$.
Then there is some action $a$ such that $\psi \models \MAY{a} \top$ and $a \notin \lambda(\SIMPL{\xi})$.
If $a \notin \lambda(\SIMPL{\xi})$ then $\xi \models ! A$ where $a \notin A$.
But if $\psi \models \MAY{a} \top$ and $\phi \models \psi$, then $\phi \models \MAY{a} \top$ and $\phi$ is also incompatible with $\xi$'s $A$-restriction.
So $\SIMPL{\xi} \sqcap \SIMPL{\phi} = \bot$ and $\xi \in \mathcal{I}(\phi)$.

Right to left: assume, for reductio, that $\MMM \models \phi$ and $\MMM \nvDash
\psi$. we will show that $\mathcal{I}(\psi) \nsubseteq \mathcal{I}(\phi)$.
Assume $\MMM \models \phi \mbox{ and } \MMM \nvDash \psi$. We will construct
another model $\MMM_2$ such that $\MMM_2 \in \mathcal{J}(\SIMPL{\psi})$ but $\MMM_2
\notin \mathcal{J}(\SIMPL{\phi})$.  This will entail, via Lemma \ref{inc3}, that
$\mathcal{I}(\psi) \nsubseteq \mathcal{I}(\phi)$.

If $\MMM \nvDash \psi$, then there is a formula $\psi'$ that does not contain
$\AND$ such that $\psi \models \psi'$ and $\MMM \nvDash \psi'$. $\psi'$ must be
either of the form (i) $\langle a_1 \rangle ... \langle a_n \rangle
\top$ (for $n > 0$) or (ii) of the form $\langle a_1 \rangle
... \langle a_n \rangle \; !\{A\}$ where $A \subseteq \mathcal{S}
\mbox{ and } n >= 0$.

In case (i), there must be an $i$ between $0$ and $n$ such that $\MMM
\models \langle a_1 \rangle ... \langle a_i \rangle \top$ but $\MMM
\nvDash \langle a_1 \rangle ... \langle a_{i+1} \rangle \top$. We need
to construct another model $\MMM_2$ such that $\MMM_2 \sqcap \SIMPL{\psi} = \bot$,
but $\MMM_2 \sqcap \SIMPL{\phi} \neq \bot$. Letting $\MMM =
((\mathcal{W},\rightarrow,\lambda),w)$, then $\MMM \models \langle a_1
\rangle ... \langle a_i \rangle \top$ implies that there is at least
one sequence of states of the form $w, w_1, ..., w_i$ such that $w
\xrightarrow{a_1} w_1 \rightarrow ... \xrightarrow{a_i} w_i$.  Now let
$\MMM_2$ be just like $\MMM$ but with additional transition-restrictions on
each $w_i$ that it not include $a_{i+1}$.  In other words,
$\lambda_{\MMM_2}(w_i) = \lambda_\MMM(w_i) - \{a_{i+1}\}$ for all $w_i$ in
sequences of the form $w \xrightarrow{a_1} w_1 \rightarrow
... \xrightarrow{a_i} w_i$. Now $\MMM_2 \sqcap \SIMPL{\psi} = \bot$ because of
the additional transition restriction we added to $\MMM_2$, which rules out
$\langle a_1 \rangle ... \langle a_{i+1} \rangle \top$, and
a-fortiori $\psi$. But $\MMM_2 \sqcap \SIMPL{\phi} \neq \bot$, because $\MMM
\models \phi$ and $\MMM_2 \MODELLEQ \MMM$ together imply $\MMM_2 \models \phi$. So $\MMM_2$ is
indeed the model we were looking for, that is incompatible with
$\SIMPL{\psi}$ while being compatible with $\SIMPL{\phi}$.

In case (ii), $\MMM \models \langle a_1 \rangle ... \langle a_n \rangle
\top$ but $\MMM \nvDash \langle a_1 \rangle ... \langle a_n \rangle !A$
for some $A \subset \mathcal{S}$. We need to produce a model $\MMM_2$ that
is incompatible with $\SIMPL{\psi}$ but not with $\SIMPL{\phi}$. Given that
$\MMM \models \langle a_1 \rangle ... \langle a_n \rangle \top$, there is
a sequence of states $w, w_1, ..., w_n$ such that $w \xrightarrow{a_1}
w_1 \rightarrow ... \xrightarrow{a_i} w_n$. Let $\MMM_2$ be the model just
like $\MMM$ except it has an additional transition from each such $w_n$
with an action $a \notin A$. 
Clearly, $\MMM_2 \sqcap \SIMPL{\psi'} = \bot$
because of the additional $a$-transition, and given that $\psi \models
\psi'$, it follows that $\MMM_2 \sqcap \SIMPL{\psi} = \bot$. Also, $\MMM_2 \sqcap
\SIMPL{\phi} \neq \bot$, because $\MMM_2 \MODELLEQ \MMM$ and $\MMM \models \phi$.

\end{proof}

\section{Inference Rules}\label{elAndBangMore}

\begin{FIGURE}
\begin{RULES}

  \ZEROPREMISERULENAMEDRIGHT
  {
    \phi \judge \phi
  }{Id}
    \qquad
  \ZEROPREMISERULENAMEDRIGHT
  {
    \phi \judge \top
  }{$\top$-Right}
    \qquad
  \ZEROPREMISERULENAMEDRIGHT
  {
    \bot \judge \phi
  }{$\bot$-Left}
    \qquad
  \TWOPREMISERULENAMEDRIGHT
  {
    \phi \judge \psi
  }
  {
    \psi \judge \xi
  }
  {
    \phi \judge \xi
  }{Trans}
    \\\\
  \ONEPREMISERULENAMEDRIGHT
  {
    \phi \judge \psi
  }
  {
    \phi \AND \xi \judge \psi
  }{$\AND$-Left 1}
     \qquad
  \ONEPREMISERULENAMEDRIGHT
  {
    \phi \judge \psi
  }
  {
    \xi \AND \phi  \judge \psi
  }{$\AND$-Left 2}
     \qquad
  \TWOPREMISERULENAMEDRIGHT
  {
    \phi \judge \psi
  }
  {
    \phi \judge \xi
  }
  {
    \phi \judge \psi \AND \xi
  }{$\AND$-Right}
     \\\\
     \ONEPREMISERULENAMEDRIGHT
     {
       a \notin A
     }
     {
       !A \AND \MAY{a}{\phi} \judge \bot
     }{$\bot$-Right 1}
        \qquad
     \ZEROPREMISERULENAMEDRIGHT
     {
       \MAY{a}{\bot} \judge \bot
     }{$\bot$-Right 2}
     \qquad
     \TWOPREMISERULENAMEDRIGHT
     {
       \phi \judge !A
     }
     {
       A \subseteq A'
     }
     {
       \phi \judge!A'
     }{!-Right 1}
     \\\\
     \TWOPREMISERULENAMEDRIGHT
     {
       \phi \judge !A
     }
     {
       \phi \judge !B
     }
     {
       \phi \judge !(A \cap B)
     }{!-Right 2}
     \qquad
     \ONEPREMISERULENAMEDRIGHT
     {
       \phi \judge \psi
     }
     {
       \MAY{a}{\phi} \judge \MAY{a}{\psi}
     }{Normal}
     \qquad
     \ONEPREMISERULENAMEDRIGHT
     {
       \phi \judge \MAY{a}\psi \land \MAY{a}\xi
     }
     {
       \phi \judge \MAY{a}(\psi \land \xi)
     }{Det}
\end{RULES}
\caption{Proof rules.}\label{figure:elAndBangRules}
\end{FIGURE}

\NI We now present the inference rules for \cathoristic{}. There are no
axioms.

\begin{definition} Judgements are of the following form.
\[
  \phi \judge \psi.
\]
We also write $\judge \phi$ as a shorthand for $\top \judge
\phi$. Figure \ref{figure:elAndBangRules}
presents all proof rules.
\end{definition}

\NI Note that $\phi$ and $\psi$ are single formulae, not sequents.  By
using single formulae, we can avoid structural inference rules.  The
proof rules can be grouped in two parts: standard rules and rules
unique to \cathoristic{}.  Standard rules are [\RULENAME{Id}],
[\RULENAME{$\top$-Right}], [\RULENAME{$\bot$-Left}],
[\RULENAME{Trans}], [\RULENAME{$\AND$-Left 1}],
[\RULENAME{$\AND$-Left 2}] and [\RULENAME{$\AND$-Right}]. They hardly need
explanation as they are variants of familiar rules for propositional
logic, see e.g.~\cite{TroelstraAS:basprot,vanDalenD:logstr}.  We now
explain the rules that give \cathoristic{} its distinctive properties.

The rule [\RULENAME{$\bot$-Right 1}]  captures the core
exclusion property of the tantum !: for example if $A = \{male, female\}$
then $\MAY{orange}{\phi}$ is incompatible with $!A$. Thus $!A \AND
\MAY{orange}{\phi}$ must be false.

The rule [\RULENAME{$\bot$-Right 2}] expresses that falsity is `global'
  and cannot be suppressed by the modalities. For example
  $\MAY{orange}{\bot}$ is false, simply because $\bot$ is already
  false.

[\RULENAME{Normal}] enables us to prefix an inference with a
may-modality.  This rule can also be stated in the the following more
general form:
\[
   \ONEPREMISERULENAMEDRIGHT
   {
     \phi_1\AND ...\AND \phi_n \judge \psi
   }
   {
     \MAY{a}{\phi_1}\AND ...\AND \MAY{a}{\phi_n} \judge \MAY{a}{\psi}
   }{Normal-Multi}
\]

\NI But it is not necessary because [\RULENAME{Normal-Multi}] is
derivable from [\RULENAME{Normal}] as we show in the examples below.

\subsection{Example inferences}

\NI We prove that we can use $\phi \AND \psi \judge \xi$ to derive
$\MAY{a} \phi \land \MAY{a} \psi \judge \MAY{a} \xi$:

\begin{center}
  \AxiomC{$\MAY{a} \phi \land \MAY{a} \psi \judge \MAY{a} \phi \land \MAY{a} \psi$}
  \RightLabel{\RULENAME{\small Det}}
  \UnaryInfC{$\MAY{a} \phi \land \MAY{a} \psi \judge \MAY{a} (\phi \land \psi)$}
  \AxiomC{$\phi \AND \psi  \judge  \xi$}
  \RightLabel{\RULENAME{\small Normal}}
  \UnaryInfC{$\MAY{a} (\phi \AND \psi)  \judge  \MAY{a} \xi$}
  \RightLabel{\RULENAME{\small Trans}}
  \BinaryInfC{$\MAY{a} \phi \land \MAY{a} \psi \judge \MAY{a} \xi$}
  \DisplayProof
\end{center}

\NI Figure \ref{figure:elAndBangMore:bigDerivations} demonstrates how
to infer $\MAY{a}!\{b,c\} \land \MAY{a}!\{c,d\} \judge \MAY{a}!\{c\}$
and $\MAY{a} !\{b\} \land \MAY{a} \MAY{c} \top \judge \MAY{d} \top$.

\subsection{\RULENAME{!-Left} and \RULENAME{!-Right}}

The rules [\RULENAME{!-Right 1}, \RULENAME{!-Right 2}] jointly express
how the subset relation $\subseteq$ on sets of actions relates to
provability. Why  don't we need a corresponding rule \RULENAME{!-Left} for
strengthening $!$ on the left hand side?
\[
   \TWOPREMISERULENAMEDRIGHT
     {
       \phi \AND \, !A \judge \psi
     }
     {
       A' \subseteq A
     }
     {
       \phi \AND\, !A' \judge \psi
     }{!-Left}
\]
The reason is that [\RULENAME{!-Left}] can be derived as follows.
\begin{center}
  \AxiomC{$\phi \AND\, !A'  \judge  \phi \AND\, !A'$}
  \AxiomC{$A' \subseteq A$}
  \RightLabel{\small\RULENAME{!-Right 1}}
  \BinaryInfC{$\phi \AND\, !A'  \judge  \phi \AND\, !A$}
  \AxiomC{$\phi \AND\, !A  \judge  \psi$}
  \RightLabel{\small\RULENAME{Trans}}
  \BinaryInfC{$\phi \AND\, !A'  \judge  \psi$}
  \DisplayProof
\end{center}

\NI Readers familiar with object-oriented programming will recognise
[\RULENAME{!-Left}] as contravariant and [\RULENAME{!-Right 1}]
as covariant subtyping. Honda \cite{HondaK:thetypftpc} develops a full
theory of subtyping based on similar ideas.  All three rules embody
the intuition that whenever $A \subseteq A'$ then asserting that $!A'$
is as strong as, or a stronger statement than
$!A$. [\RULENAME{!-Left}] simply states that we can always strengthen
our premise, while [\RULENAME{!-right 1}] allows us to weaken the
conclusion.

\subsection{Characteristic formulae}

In order to prove completeness, below, we need the notion of a
\emph{characteristic formula} of a model.  The function
$\SIMPL{\cdot}$ takes a formula as argument and returns
the least upper bound of the satisfying models. Characteristic formulae go the other way: given a model
$\MMM$, $\CHAR{\MMM}$ is the logically weakest formula that describes that model. 
\clearpage
\begin{SIDEWAYSFIGURE}

\begin{center}
  \AxiomC{$!\{b,c\} \judge !\{b,c\}$}
  \RightLabel{\small \RULENAME{ $\land$ Left 1}}
  \UnaryInfC{$!\{b,c\} \land !\{c,d\} \judge !\{b,c\}$}
  \AxiomC{$!\{c,d\} \judge !\{c,d\}$}
  \RightLabel{\small \RULENAME{ $\land$ Left 2}}
  \UnaryInfC{$!\{b,c\} \land !\{c,d\} \judge !\{c,d\}$}
  \RightLabel{\small \RULENAME{ ! Right 2}}
  \BinaryInfC{$!\{b,c\} \land !\{c,d\} \judge !\{c\}$}
  \RightLabel{\small \RULENAME{ Normal}}
  \UnaryInfC{$\MAY{a} (!\{b,c\} \land !\{c,d\}) \judge \MAY{a} !\{c\}$}
  \AxiomC{$\MAY{a} !\{b,c\} \land \MAY{a} !\{c,d\} \judge \MAY{a} !\{b,c\} \land \MAY{a} !\{c,d\}$}
  \RightLabel{\small \RULENAME{ Det}}
  \UnaryInfC{$\MAY{a} !\{b,c\} \land \MAY{a} !\{c,d\} \judge \MAY{a} (!\{b,c\} \land !\{c,d\})$}
  \RightLabel{\small \RULENAME{ Trans}}
  \BinaryInfC{$\MAY{a} !\{b,c\} \land \MAY{a} !\{c,d\} \judge \MAY{a}!\{c\}$}
  \DisplayProof
\end{center}
\vspace{7mm}
\begin{center}
  \AxiomC{$!\{b\} \land \MAY{c}\top \judge \bot$}
  \RightLabel{\small \RULENAME{ Normal}}
  \UnaryInfC{$\MAY{a}(!\{b\} \land \MAY{c}\top) \judge \MAY{a} \bot$}
  \AxiomC{$\MAY{a}!\{b\} \land \MAY{a} \MAY{c}\top \judge \MAY{a}!\{b\} \land \MAY{a} \MAY{c}\top$}
  \RightLabel{\small \RULENAME{ Det}}
  \UnaryInfC{$\MAY{a}!\{b\} \land \MAY{a} \MAY{c}\top \judge \MAY{a}(!\{b\} \land \MAY{c}\top)$}
  \RightLabel{\small \RULENAME{ Trans}}
  \BinaryInfC{$\MAY{a}!\{b\} \land \MAY{a} \MAY{c}\top  \judge \MAY{a} \bot$}  
  \AxiomC{$\MAY{a}\bot \judge \bot$}
  \RightLabel{\small \RULENAME{ Trans}}
  \BinaryInfC{$\MAY{a}!\{b\} \land \MAY{a} \MAY{c}\top  \judge \bot$}  
  \AxiomC{$\bot \judge \MAY{d} \top$}
  \RightLabel{\small \RULENAME{ Trans}}
  \BinaryInfC{$\MAY{a}!\{b\} \land \MAY{a} \MAY{c}\top  \judge  \MAY{d} \top$}    
  \DisplayProof
\end{center}

\caption{Derivations of $\MAY{a}!\{b,c\} \land \MAY{a}!\{c,d\} \judge \MAY{a}!\{c\}$ (top) and
 $\MAY{a} !\{b\} \land \MAY{a} \MAY{c} \top \judge \MAY{d} \top$ (bottom).}\label{figure:elAndBangMore:bigDerivations}
\end{SIDEWAYSFIGURE}

\clearpage

\begin{definition}
Let $\MMM$ be a cathoristic model that is a tree.
\begin{eqnarray*}
  \CHAR{\bot} &\ =\ & \langle a \rangle \top \AND ! \emptyset  \mbox{ for some fixed action }a \in \Sigma  \\
  \CHAR{\MMM, w} & = & \mathsf{bang}(\MMM,w) \AND \bigwedge_{w \xrightarrow{a} w'} \langle a \rangle \CHAR{\MMM, w'}  
\end{eqnarray*}

\end{definition}

\NI Note that $\bot$ requires a particular action $a \in \Sigma$. This
is why we required, in Section \ref{elsyntax}, that $\Sigma$ is
non-empty.

The functions $\mathsf{bang}(\cdot)$ on models are given by the
following clauses.

\begin{eqnarray*}
  \mathsf{bang}((S,\rightarrow,\lambda),w) 
     & \ = \ & 
  \begin{cases}
    \top & \mbox{ if } \lambda(w) = \Sigma  \\
    ! \; \lambda(w) & \mbox{ otherwise }  
  \end{cases} \\
\end{eqnarray*}

\NI Note that $\CHAR{\MMM}$ is finite if $\MMM$ contains no cycles and
if $\lambda(x)$ is either $\Sigma$ or finite for all states $x$.  We
state without proof that $\SIMPL{\cdot}$ and $\CHAR{\cdot}$ are
inverses of each other (for tree models $\MMM$) in that:

\begin{itemize}

\item $\SIMPL{\CHAR{\MMM}} \ \MODELEQ \  \MMM$. 

\item $\models \CHAR{\SIMPL{\phi}}$ iff $\models\phi$.

\end{itemize}

\subsection{Soundness and completeness}

\begin{theorem}\label{theorem:elAndBang:soundComplete}
The rules in Figure \ref{figure:elAndBangRules} are sound and complete:
\begin{enumerate}

\item\label{theorem:elAndBang:sound} (Soundness) $\phi \judge \psi$ implies $\phi \models \psi$.

\item\label{theorem:elAndBang:complete} (Completeness) $\phi \models \psi$ implies $\phi \judge \psi$.

\end{enumerate}
\end{theorem}

\NI Soundness is immediate from the definitions. 
 To prove completeness  we will show that $\phi
\models \psi$ implies there is a derivation of $\phi \judge \psi$.  Our proof
will make use of two key facts (proved in Sections \ref{prooflemma4} and \ref{prooflemma5} below):







\begin{lemma}\label{lemma:completeness:4}
If $\MMM \models \phi$ then $\CHAR{\MMM} \judge \phi$.
\end{lemma}

\begin{lemma}\label{lemma:completeness:5}
For all formulae $\phi$, we can derive $\phi \judge \CHAR{\SIMPL{\phi}}$.
\end{lemma}

\NI Lemma \ref{lemma:completeness:4} states that, if $\phi$ is
satisfied by a model, then there is a proof that the characteristic
formula describing that model entails $\phi$.  In Lemma
\ref{lemma:completeness:5}, $\SIMPL{\phi}$ is the simplest
model satisfying $\phi$, and $\CHAR{\MMM}$ is the simplest formula
describing $m$, so $\CHAR{\SIMPL{\phi}}$ is a simplified form of
$\phi$. This lemma states that \cathoristic\ has the inferential capacity to
transform any proposition into its simplified form.

With these two lemmas in hand, the proof of completeness is
straightforward.  Assume $\phi \models \psi$.  Then all models which satisfy
$\phi$ also satisfy $\psi$.  In particular, $\SIMPL{\phi} \models \psi$.  Then
$\CHAR{\SIMPL{\phi}} \judge \psi$ by Lemma \ref{lemma:completeness:4}.  But we
also have, by Lemma \ref{lemma:completeness:5}, $\phi \judge
\CHAR{\SIMPL{\phi}} $.  So by transitivity, we have $\phi \judge \psi$.  

\subsection{Proofs of Lemmas \ref{lemma:completeness:4}, \ref{lemma:completeness:5} and \ref{final_completeness_lemma}}

\subsubsection{Proof of Lemma \ref{lemma:completeness:4}}
\label{prooflemma4}
If $\MMM\models \phi$ then $\CHAR{\MMM} \judge \phi$.

\NI We proceed by induction on $\phi$.

\CASE{$\phi$ is $\top$} Then we can prove $ \CHAR{\MMM} \judge \phi$
immediately using axiom [\RULENAME{$\top$ Right}.

\CASE{$\phi$ is $\psi \AND \psi'$} By the induction hypothesis, $
\CHAR{\MMM} \judge \psi$ and $ \CHAR{\MMM} \judge \psi'$.  The proof
of $ \CHAR{\MMM} \judge \psi \AND \psi'$ follows immediately using
[\RULENAME{$\AND$ Right}.

\CASE{$\phi$ is $\langle a \rangle \psi$}
If $\MMM \models \langle a \rangle \psi$, then either $\MMM = \bot$ or $\MMM$ is a  model of the form $(\CAL{L},w)$.

\SUBCASE{$\MMM = \bot$} In this case, $ \CHAR{\MMM} = \CHAR{\bot} =
\bot$. (Recall, that we are overloading $\bot$ to mean both the model
at the bottom of our lattice and a formula (such as $\langle a \rangle
\top \AND !\emptyset$) which is always false).  In this case, $ \CHAR{\bot}
\judge \langle a \rangle \psi$ using [\RULENAME{$\bot$ Left}.

\SUBCASE{$m$ is a model of the form $(\CAL{L},w)$} Given $\MMM \models
\langle a \rangle \psi$, and that $\MMM$ is a model of the form
$(\CAL{L},w)$, we know that:
\[
(\CAL{L},w) \models \langle a \rangle \psi
\]
From the satisfaction clause for $\langle a \rangle$, it follows that:
\[
\exists w' \mbox{ such that } w \xrightarrow{a} w' \mbox { and } (\CAL{L},w') \models \psi
\]
By the induction hypothesis:
\[
 \CHAR{(\CAL{L},w')} \judge \psi
\]
Now by [\RULENAME{Normal}]:
\[
\langle a \rangle  \CHAR{(\CAL{L},w')} \judge \langle a \rangle \psi
\]
Using repeated application of [\RULENAME{$\AND$ Left}], we can show:
\[
 \CHAR{(\CAL{L},w)} \judge \langle a \rangle  \CHAR{(\CAL{L},w')}
\]
Finally, using [\RULENAME{Trans}], we derive:
\[
 \CHAR{(\CAL{L},w)} \judge  \langle a \rangle \psi
\]

\CASE{$\phi$ is $\fBang \psi$} If $(\CAL{L},w) \models \fBang A$,
then $\lambda(w) \subseteq A$.  Then $ \CHAR{(\CAL{L},w)} = ! \;
\lambda(w) \AND \phi$.  Now we can prove $! \; \lambda(w) \AND \phi
\judge \fBang A$ using [\RULENAME{$!$ Right 1}] and repeated applications of
       [\RULENAME{$\AND$ Left}].

\subsubsection{Proof of Lemma \ref{lemma:completeness:5}}
\label{prooflemma5}

Now we prove Lemma \ref{lemma:completeness:5}: 
for all formulae $\phi$, we can derive $\phi \judge \CHAR{\SIMPL{\phi}}$.

\begin{proof}
Induction on $\phi$.

\CASE{$\phi$ is $\top$} Then we can prove $\top \judge \top$ using
either [\RULENAME{$\top$ Right}] or [\RULENAME{Id}].

\CASE{$\phi$ is $\psi \AND \psi'$} By the induction hypothesis,
$\psi \judge \CHAR{\SIMPL{\psi}}$ and $\psi' \judge
\CHAR{\SIMPL{\psi'}}$.  Using [\RULENAME{$\AND$ Left}] and [\RULENAME{$\AND$
  Right}], we can show:
\[
\psi \AND \psi' \judge  \CHAR{\SIMPL{\psi}} \AND  \CHAR{\SIMPL{\psi'}}
\]

In order to continue the proof, we need the following lemma, proven
in the next subsection.

\begin{lemma}
\label{final_completeness_lemma}
For all cathoristic  models $\MMM$ and $\MMM_2$ that are trees, $ \CHAR{\MMM} \AND
\CHAR{\MMM_2} \judge \CHAR{\MMM \sqcap \MMM_2}$.
\end{lemma}

\NI From Lemma \ref{final_completeness_lemma} (substituting $\SIMPL{\psi}$ for $\MMM$ and $\SIMPL{\psi'}$ for $\MMM_2$, and noting that $\SIMPL{}$ always produces acyclic models), it follows that:
\[
 \CHAR{\SIMPL{\psi}} \AND  \CHAR{\SIMPL{\psi'}} \judge  \CHAR{\SIMPL{\psi \AND \psi'}}
\]
Our desired result follows using [\RULENAME{Trans}].

\CASE{$\phi$ is $\langle a \rangle \psi$} By the induction
hypothesis, $\psi \judge \CHAR{\SIMPL{\psi}}$.  Now there are two
sub-cases to consider, depending on whether or not $
\CHAR{\SIMPL{\psi}} = \bot$.  

\SUBCASE{$ \CHAR{\SIMPL{\psi}} = \bot$} In this case, $
\CHAR{\SIMPL{\langle a \rangle \psi}}$ also equals $\bot$.  By the
induction hypothesis:
\[
\psi \judge \bot
\]
By [\RULENAME{Normal}]:
\[
\langle a \rangle \psi \judge \langle a \rangle \bot
\]
By [\RULENAME{$\bot$ Right 2}]:
\[
\langle a \rangle \bot \judge \bot
\]
The desired proof that:
\[
\langle a \rangle \psi \judge \bot
\]
follows by [\RULENAME{Trans}].

\SUBCASE{$ \CHAR{\SIMPL{\psi}} \neq \bot$}
By the induction hypothesis, $\psi \judge  \CHAR{\SIMPL{\psi}}$.
So, by [\RULENAME{Normal}]:
\[
\langle a \rangle \psi \judge \langle a \rangle  \CHAR{\SIMPL{\psi}}
\]
The desired conclusion follows from noting that:
\[
 \langle a \rangle  \CHAR{\SIMPL{\psi}} =  \CHAR{\SIMPL{\langle a \rangle \psi}}
 \]

 \CASE{$\phi$ is $\fBang A$} If $\phi$ is $\fBang A$, then $
 \CHAR{\SIMPL{\phi}}$ is $\fBang A \AND \top$.  We can prove $\fBang A
 \judge \fBang A \AND \top$ using [\RULENAME{$\AND$ Right}], [\RULENAME{$\top$
   Right}] and [\RULENAME{Id}].
\end{proof}

\subsubsection{Proof of Lemma \ref{final_completeness_lemma}}
\label{prooflemma6}

We can now finish the proof of Lemma \ref{lemma:completeness:5} by
giving the missing proof of Lemma \ref{final_completeness_lemma}.

\begin{proof}
There are two cases to consider, depending on whether or not $(\MMM
\sqcap \MMM_2) = \bot$.

\CASE{$(\MMM \sqcap \MMM_2) = \bot$}
If $(\MMM \sqcap \MMM_2) = \bot$, there are three possibilities:
\begin{itemize}
\item
$\MMM = \bot$
\item
$\MMM_2 = \bot$
\item
Neither $\MMM$ nor $\MMM_2$ are $\bot$, but together they are incompatible. 
\end{itemize}
If either $\MMM$ or $\MMM_2$ is $\bot$, then the proof is a simple application of [\RULENAME{Id}] followed by [\RULENAME{$\AND$ Left}].

Next, let us consider the case where neither $\MMM$ nor $\MMM_2$ are $\bot$, but together they are incompatible.
Let $\MMM = (\mathcal{L}, w_1)$ and $\MMM' = (\mathcal{L}', w'_1)$.
If $\MMM \sqcap \MMM_2 = \bot$, then there is a finite sequence of actions $a_1, ..., a_{n-1}$ such that both $\MMM$ and $\MMM'$ satisfy $\MAY{a_1} ... \MAY{a_{n-1}}\top$, but they disagree about the state-labelling on the final state of this chain. In other words, there is a $b$-transition from the final state in $\MMM$ which is ruled-out by the $\lambda'$ state-labelling in $\MMM'$. So there is a set of states $w_1, ..., w'_1, ...$ and a finite set $X$ of actions such that:
\begin{itemize}

\item $w_1 \xrightarrow{a_1} w_2 \xrightarrow{a_2} ... \xrightarrow{a_{n-1}} w_n$.

\item $w_1' \xrightarrow{a_1} w'_2 \xrightarrow{a_2} ... \xrightarrow{a_{n-1}} w'_n$.

\item $w_n \xrightarrow{b} w_{n+1}$.

\item $\lambda'(w'_n) = X \text{ with } b \notin X$.

\end{itemize}
Now it is easy to show, using [\RULENAME{$\AND$ Left}], that
\begin{eqnarray*}
\CHAR{\mathsf{\MMM}} &\judge& \MAY{a_1} ... \MAY{a_{n-1}} \MAY{b} \top \\
\CHAR{\mathsf{\MMM'}} &\judge& \MAY{a_1} ... \MAY{a_{n-1}} \fBang X
\end{eqnarray*}
Now using [\RULENAME{$\AND$ Left}] and [\RULENAME{$\AND$ Right}]:
\[
\CHAR{\mathsf{\MMM}} \land \CHAR{\mathsf{\MMM}'} \judge  \MAY{a_1} ... \MAY{a_{n-1}} \MAY{b} \top \land  \judge \MAY{a_1} ... \MAY{a_{n-1}} \fBang X
\]
Now using [\RULENAME{Det}]:
\[
\CHAR{\mathsf{\MMM}} \land \CHAR{\mathsf{\MMM}'} \judge  \MAY{a_1} ... \MAY{a_{n-1}} (\MAY{b} \top \land \fBang X)
\]
Now, using [\RULENAME{$\bot$ Right 1}]:
\[
\MAY{b} \top \land \fBang X \judge \bot
\]
Using $n-1$ applications of  [\RULENAME{$\bot$ Right 2}]:
\[
\MAY{a_1} ... \MAY{a_{n-1}} (\MAY{b} \top \land \fBang X) \judge \bot
\]
Finally, using [\RULENAME{Trans}], we derive:
\[
\CHAR{\mathsf{\MMM}} \land \CHAR{\mathsf{\MMM}'} \judge \bot
\]
\CASE{$(\MMM \sqcap \MMM_2) \neq \bot$} From the construction of
$\mathsf{merge}$, if $\MMM$ and $\MMM'$ are acyclic, then $\MMM \sqcap
\MMM'$ is also acyclic.  If $\MMM \sqcap \MMM'$ is acyclic, then
$\CHAR{\MMM \sqcap \MMM'}$ is equivalent to a set $\Gamma$ of
sentences of one of two forms:
\[
   \MAY{a_1} ... \MAY{a_n} \top 
      \qquad\qquad
   \MAY{a_1} ... \MAY{a_n} ! X
\]

\begin{figure}[H]
\centering
\begin{tikzpicture}[node distance=1.3cm,>=stealth',bend angle=45,auto]
  \tikzstyle{place}=[circle,thick,draw=blue!75,fill=blue!20,minimum size=6mm]
  \tikzstyle{red place}=[place,draw=red!75,fill=red!20]
  \tikzstyle{transition}=[rectangle,thick,draw=black!75,
  			  fill=black!20,minimum size=4mm]
  \tikzstyle{every label}=[red]
    
  \begin{scope}
    \node [place] (w1) {$\Sigma$};
    \node [place] (e1) [below left of=w1] {$\{c,d\}$}
      edge [pre]  node[swap] {a}                 (w1);      
    \node [place] (e2) [below right of=w1] {$\Sigma$}
      edge [pre]  node[swap] {b}                 (w1);      
    \node [place] (e3) [below of=e1] {$\Sigma$}
      edge [pre]  node[swap] {c}                 (e1);      
  \end{scope}
    
\end{tikzpicture}
\caption{Example of $\sqcap$}
\label{setofpaths}
\end{figure}
For example, if $\MMM \sqcap \MMM'$ is as in Figure \ref{setofpaths}, then 
\[
\CHAR{\MMM \sqcap \MMM'} = \MAY{a}(\fBang \{c,d\} \land \MAY{c} \top) \land \MAY{b} \top
\]
This is equivalent to the set $\Gamma$ of sentences:
\[
\MAY{a}\MAY{c} \top \qquad\qquad
\MAY{b} \top \qquad\qquad
\MAY{a}\fBang\{c,d\}
\]
Now using [\RULENAME{$\AND$ Right}] and [\RULENAME{Det}] we can show that
\[
\bigwedge_{\phi \in \Gamma} \phi \judge \CHAR{\MMM \sqcap \MMM'}
\]
We know that for all $\phi \in \Gamma$
\[
\MMM \sqcap \MMM' \models \phi
\]
We just need to show that:
\[
\CHAR{\MMM} \land \CHAR{\MMM'} \judge \phi
\]
Take any $\phi \in \Gamma$ of the form $\MAY{a_1} ... \MAY{a_n} ! X$ for some finite $X \subseteq \Sigma$. (The case where $\phi$ is of the form $\MAY{a_1} ... \MAY{a_n} \top$ is very similar, but simpler).
If $\MMM \sqcap \MMM' \models \MAY{a_1} ... \MAY{a_n} ! X$ then either:
\begin{enumerate}
\item
$\MMM \models \MAY{a_1} ... \MAY{a_n} ! X$ but $\MMM' \nvDash \MAY{a_1} ... \MAY{a_n} \top$
\item
$\MMM' \models \MAY{a_1} ... \MAY{a_n} ! X$ but $\MMM \nvDash \MAY{a_1} ... \MAY{a_n} \top$
\item
$\MMM \models \MAY{a_1} ... \MAY{a_n} ! X_1$ and $\MMM' \models \MAY{a_1} ... \MAY{a_n} ! X_2$ and $X_1 \cap X_2 \subseteq X$
\end{enumerate}
In the first two cases, showing $\CHAR{\MMM} \land \CHAR{\MMM'} \judge \phi$ is just a matter of repeated application of   [\RULENAME{$\AND$ Left}] and [\RULENAME{$\AND$ Right}].
In the third case, let $\MMM = (\mathcal{L}, w_1)$ and $\MMM' = (\mathcal{L}', w'_1)$.
If $\MMM \models \MAY{a_1} ... \MAY{a_n} ! X_1$ and $\MMM' \models \MAY{a_1} ... \MAY{a_n} ! X_2$ then there exists sequences $w_1, ..., w_{n+1}$ and $w'_1, ..., w'_{n+1}$ of states such that
\begin{itemize}

\item $w_1 \xrightarrow{a_1} ... \xrightarrow{a_n} w_{n+1}$.

\item $w'_1 \xrightarrow{a_1} ... \xrightarrow{a_n} w'_{n+1}$.

\item $\lambda(w_{n+1}) \subseteq X_1$.

\item $\lambda'(w'_{n+1}) \subseteq X_2$.

\end{itemize}

\NI Now from the definition of $\CHAR{}$:
\[
   \CHAR{(\mathcal{L}, w_{n_1})} \judge \fBang X_1 
      \qquad\qquad
   \CHAR{(\mathcal{L}', w'_{n_1})} \judge \fBang X_2
\]
Now using [\RULENAME{\fBang Right 2}]:
\[
\CHAR{(\mathcal{L}, w_{n_1})} \land \CHAR{(\mathcal{L}', w'_{n_1})} \judge \fBang (X_1 \cap X_2)
\]
Using [\RULENAME{\fBang Right 1}]:
\[
\CHAR{(\mathcal{L}, w_{n_1})} \land \CHAR{(\mathcal{L}', w'_{n_1})} \judge \fBang X
\]
Using $n$ applications of [\RULENAME{Normal}]:
\[
\MAY{a_1} ... \MAY{a_n} (\CHAR{(\mathcal{L}, w_{n_1})} \land \CHAR{(\mathcal{L}', w'_{n_1})}) \judge \MAY{a_1} ... \MAY{a_n} \fBang X
\]
Finally, using $n$ applications of [\RULENAME{Det}]:
\[
\CHAR{ (\mathcal{L}, w_1)} \land \CHAR{ (\mathcal{L}', w'_1)} \judge \MAY{a_1} ... \MAY{a_n} (\CHAR{(\mathcal{L}, w_{n_1})} \land \CHAR{(\mathcal{L}', w'_{n_1})})
\]
So, by [\RULENAME{Trans]}
\[
\CHAR{\MMM} \land \CHAR{\MMM'} \judge \MAY{a_1} ... \MAY{a_n} \fBang X
\]
\end{proof}

\section{Compactness and the standard translation to first-order logic }
\label{compactness}

This section studies two embeddings of \cathoristic{} into first-order
logic. The second embedding is used to prove that \cathoristic{} satisfies compactness.

\subsection{Translating from  cathoristic to
            first-order logic}\label{standardTranslation}

The study of how a logic embeds into other logics is interesting in
parts because it casts a new light on the logic that is the target of
the embedding.  A good example is the standard translation of modal
into first-order logic.  The translation produces various fragments:
the finite variable fragments, the fragment closed under bisimulation,
guarded fragments.  These fragments have been investigated deeply, and
found to have unusual properties not shared by the whole of \fol.
Translations also enable us to push techniques, constructions and
results between logics.  In this section, we translate \cathoristic{}
into first-order logic.

\begin{definition}
The first-order signature $\SSS$ has a nullary predicate $\top$, a
family of unary predicates $\RESTRICT{A}{\cdot}$, one for each finite
subset $A \subseteq \Sigma$, and a family of binary predicates
$\ARROW{a}{x}{y}$, one for each action $a \in \Sigma$. 

\end{definition}

\NI The intended interpretation is as follows.

\begin{itemize}

\item The universe is composed of states.

\item The predicate $\top$ is true everywhere.

\item For each finite $A \subseteq \Sigma$ and each state $s$,  $\RESTRICT{A}{s}$
is true if 
  $\lambda(x) \subseteq A$.

\item A set of two-place predicates $\ARROW{a}{x}{y}$, one for each $a
  \in \Sigma$, where $x$ and $y$ range over states. $\ARROW{a}{x}{y}$
  is true if $x \xrightarrow{a} y$.

\end{itemize}

\NI If $\Sigma$ is infinite, then $\RESTRICT{A}{\cdot}$ and
$\ARROW{a}{\cdot}{\cdot}$ are infinite families of relations.

\begin{definition}
 Choose two fixed variables $x, y$, let $a$ range over actions in
$\Sigma$, and $A$ over finite subsets of $\Sigma$. Then the restricted
fragment of first-order logic that is the target of our translation is given by the
following grammar, where $w, z$ range over $x, y$.

\begin{GRAMMAR}
  \phi 
     &\quad ::= \quad&
  \top \fOr \ARROW{a}{w}{z}\fOr \RESTRICT{A}{z} \fOr \phi \AND \psi \fOr \exists x. \phi 
\end{GRAMMAR}

\end{definition}

\NI This fragment has no negation, disjunction, implication, or
universal quantification.

\begin{definition}
The translations $\SEMB{\phi}_x$ and $\SEMB{\phi}_y$ of cathoristic formula 
$\phi$ are given relative to a state, denoted by either $x$ or $y$.

\[
\begin{array}{rclcrcl}
  \SEMB{\top}_x & \ = \ & \top  
     &\quad& 
  \SEMB{\top}_y & \ = \ & \top 
     \\
  \SEMB{\phi \AND \psi}_x & = & \SEMB{\phi}_x \AND \SEMB{\psi}_x  
     && 
  \SEMB{\phi \AND \psi}_y & = & \SEMB{\phi}_y \AND \SEMB{\psi}_y  
     \\
  \SEMB{\langle a \rangle \phi}_x & = & \exists y.(\ARROW{a}{x}{y} \AND \SEMB{\phi}_y)  
     &&
  \SEMB{\langle a \rangle \phi}_y & = & \exists x.(\ARROW{a}{y}{x} \AND \SEMB{\phi}_x)  
     \\
  \SEMB{\fBang A}_x & = & \RESTRICT{A}{x}
     &&
  \SEMB{\fBang A}_y & = & \RESTRICT{A}{y}
\end{array}
\]

\end{definition}

\NI The translations on the left and right are identical, except for
switching $x$ and $y$. Here is an example translation.
\[
   \SEMB{\langle a \rangle \top \AND \fBang \{a\}}_x 
      = 
   \exists y.(\ARROW{a}{x}{y} \AND \top ) \AND \RESTRICT{\{a\}}{x}
\]

\NI We now establish the correctness of the encoding. The key issue is
that not every first-order model of our first-order signature
corresponds to a cathoristic model because determinism, well-sizedness and
admissibility are not enforced by our signature alone. In other words,
models may contain `junk'.  We deal with this problem following ideas
from modal logic \cite{BlackburnP:modlog}: we add a translation
$\SEMB{\LLL}$ for cathoristic transition systems, and then prove the
following theorem.

\begin{theorem}[correspondence theorem]\label{correspondence:theorem:1}
Let $\phi$ be a \cathoristic{} formula and $\MMM = (\LLL, s)$ a cathoristic
model.
\[
   \MMM \models \phi \quad  \text{iff} \quad \SEMB{\LLL} \models_{x \mapsto s} \SEMB{\phi}_x.
\]
And likewise for $\SEMB{\phi}_y$.
\end{theorem}

\NI The definition of $\SEMB{\LLL}$ is simple.

\begin{definition}
Let $\LLL = (S, \rightarrow, \lambda)$ be a cathoristic transition
system. Clearly $\LLL$ gives rise to an $\SSS$-model $\SEMB{\LLL}$ as
follows.
\begin{itemize}

\item The universe is the set $S$ of states.

\item The relation symbols are interpreted as follows.

  \begin{itemize}

    \item $\top^{\SEMB{\LLL}}$ always holds.

    \item $\mathsf{Restrict}_{A}^{\SEMB{\LLL}} = \{ s \in S\ |\ \lambda(s) \subseteq A\}$.

    \item $\mathsf{Arrow^{\SEMB{\LLL}}}_{a} = \{(s, t) \in S \times S\ |\ s \TRANS{a} t\}$.

  \end{itemize}
\end{itemize}
\end{definition}

\NI We are now ready to prove Theorem \ref{correspondence:theorem:1}.
\begin{proof}
By induction on the structure of $\phi$. The cases $\top$ and $\phi_1
\AND \phi_2$ are straightforward.  The case $\MAY{a}\psi$ is handled
as follows.
\begin{eqnarray*}
  \lefteqn{
  \SEMB{\LLL} \models_{x \mapsto s} \SEMB{\MAY{a}\psi}_x}\hspace{5mm} 
     \\
     &\quad \text{iff}\quad &
  \SEMB{\LLL} \models_{x \mapsto s} \exists y.(\ARROW{a}{x}{y} \AND \SEMB{\psi}_y) 
     \\
     &\text{iff}&
  \text{exists}\ t \in S. \SEMB{\LLL} \models_{x \mapsto s, y \mapsto t} \ARROW{a}{x}{y} \AND \SEMB{\psi}_y
     \\
     &\text{iff}&
  \text{exists}\ t \in S. \SEMB{\LLL} \models_{x \mapsto s, y \mapsto t} \ARROW{a}{x}{y} \ \text{and}\ \SEMB{\LLL} \models_{x \mapsto s, y \mapsto t}  \SEMB{\psi}_y
     \\
     &\text{iff}&
  \text{exists}\ t \in S. s \TRANS{a} t \ \text{and}\ \SEMB{\LLL} \models_{x \mapsto s, y \mapsto t}  \SEMB{\psi}_y
     \\
     &\text{iff}&
  \text{exists}\ t \in S. s \TRANS{a} t \ \text{and}\ \SEMB{\LLL} \models_{y \mapsto t}  \SEMB{\psi}_y \qquad (\text{as $x$ is not free in $\psi$})
     \\
     &\text{iff}&
  \text{exists}\ t \in S. s \TRANS{a} t \ \text{and}\ \MMM \models \psi
     \\
     &\text{iff}&
  \MMM \models \MAY{a}\psi  
\end{eqnarray*}

\NI Finally, if $\phi$ is $!A$ the derivation comes straight from the
definitions.
\begin{eqnarray*}
  \SEMB{\LLL} \models_{x \mapsto s} \SEMB{!A}_x
    &\quad \text{iff}\quad &
  \SEMB{\LLL} \models_{x \mapsto s} \RESTRICT{A}{x}
     \\
     &\text{iff}&
  \lambda(s) \subseteq A
     \\
     &\text{iff}&
  \MMM \models\ !A.
\end{eqnarray*}
\end{proof}

\subsection{Compactness by translation}\label{compactnessProof}

\NI First-order logic satisfies \emph{compactness}: a set $S$ of sentences has a
model exactly when every finite subset of $S$ does. What about
\cathoristic{}?

We can prove compactness of modal logics using the standard
translation from modal to first-order logic \cite{BlackburnP:modlog}:
we start from a set of modal formula such that each finite subset has
a model. We translate the modal formulae and models to first-order
logic, getting a set of first-order formulae such that each finite
subset has a first-order model. By compactness of first-order logic, we
obtain a first-order model of the translated modal formulae. Then we
translate that first-order model back to modal logic, obtaining a
model for the original modal formulae, as required. The last step
proceeds without a hitch because the modal and the first-order notions
of model are identical, save for details of presentation.

Unfortunately we cannot do the same with the translation from
\cathoristic{} to first-order logic presented in the previous
section. The problem are the first-order models termed `junk' above.
The target language of the translation is not expressive enough to
have formulae that can guarantee such constraints.  As we have no
reason to believe that the first-order model whose existence is
guaranteed by compactness isn't `junk', we cannot prove compactness
with the translation.  We solve this problem with a second translation, this time
into a more expressive first-order fragment where we can constrain
first-order models easily using formulae. The fragment we use now
lives in two-sorted first-order logic (which can easily be reduced to
first-order logic \cite{EndertonHB:matinttl}).

\begin{definition}
The two-sorted first-order signature $\SSS'$ is given as follows.
\begin{itemize}

\item $\SSS'$ has two sorts, states and actions. 

\item The action constants are given by $\Sigma$. There
are no state constants. 

\item $\SSS'$ has a nullary predicate $\top$.

\item A binary predicate $\ALLOWED{}{\cdot}{\cdot}$. The intended
  meaning of $\ALLOWED{}{x}{a}$ is that at the state denoted by $x$ we
  are allowed to do the action $a$.

\item A ternary predicate $\ARROWTWO{}{\cdot}{\cdot}{\cdot}$ where
  $\ARROWTWO{}{x}{a}{y}$ means that there is a transition from the
  state denoted by $x$ to the state denoted by $y$, and that
  transition is labelled $a$.

\end{itemize}
\end{definition}

\begin{definition}
The encoding $\SEMBTWO{\phi}_x$ of \cathoristic{} formulae is given by the following clauses.
\begin{eqnarray*}
  \SEMBTWO{\top}_x & \ = \ & \top
     \\
  \SEMBTWO{\phi \AND \psi}_x & = & \SEMBTWO{\phi}_x \AND \SEMBTWO{\psi}_x
     \\
  \SEMBTWO{\langle a \rangle \phi}_x & = & \exists^{st} y.(\ARROWTWO{}{x}{a}{y} \AND \SEMBTWO{\phi}_y)
     \\
  \SEMBTWO{\fBang A}_x & = & \forall^{act} a.(\ALLOWED{}{x}{a} \IMPLIES a \in A) 
\end{eqnarray*}

\end{definition}

\NI Here we use $\exists^{st}$ to indicate that this existential
quantifier ranges over the sort of states, and $\forall^{act}$ for the
universal quantifier ranging over actions. The expression $a \in A$ is
a shorthand for the first-order formula
\[
   a = a_1 \OR a = a_2 \OR \cdots \OR a = a_n
\]
assuming that $A = \{a_1, ..., a_n\}$. Since by definition, $A$ is
always a finite set, this is well-defined.  The translation could be
restricted to a two-variable fragment. Moreover, the standard
reduction from many-sorted to one-sorted first-order logic does not
increase the number of variables used (although predicates are added,
one per sort). We will not consider this matter further here.  

We also translate cathoristic transition systems $\SEMBTWO{\LLL}$.

\begin{definition}
Let $\LLL = (S, \rightarrow, \lambda)$ be a cathoristic transition
system. $\LLL$ gives rise to an $\SSS'$-model $\SEMBTWO{\LLL}$
as follows.
\begin{itemize}

\item The sort of states is interpreted by the set $S$.

\item The sort of actions is interpreted by the set $\Sigma$.

\item For each constant $a \in \Sigma$, $a^{\SEMBTWO{\LLL}}$ is $a$ itself.

\item The relation symbols are interpreted as follows.

  \begin{itemize}

    \item $\top^{\SEMBTWO{\LLL}}$ always holds.

    \item $\ALLOWED{\SEMBTWO{\LLL}}{s}{a}$ holds whenever $a \in \lambda(s)$.

    \item $\ARROWTWO{\SEMBTWO{\LLL}}{s}{a}{t}$ holds whenever $s \TRANS{a} t$.

  \end{itemize}
\end{itemize}
\end{definition}

\begin{theorem}[correspondence theorem]\label{correspondence:theorem:2}
Let $\phi$ be a \cathoristic{} formula and $\MMM = (\LLL, s)$ a cathoristic
model.
\[
   \MMM \models \phi \quad  \text{iff} \quad \SEMBTWO{\LLL} \models_{x \mapsto s} \SEMBTWO{\phi}_x.
\]
\end{theorem}
\begin{proof}
The proof proceeds by induction on the structure of $\phi$ and is
similar to that of Theorem \ref{correspondence:theorem:2}.
The case for the may modality proceeds as follows.  

\begin{alignat*}{2}
  \MMM \models \MAY{a}\phi
     &\quad\text{iff}\quad 
  \text{exists state $t$ with }\ s \TRANS{a} t\ \text{and}\ (\LLL, t) \models \phi \\
     &\quad\text{iff}\quad
  \text{exists state $t$ with }\ s \TRANS{a} t\ \text{and}\ \SEMBTWO{\LLL} \models_{y \mapsto t} \SEMBTWO{\phi}_y &\qquad& \text{by (IH)}\\
     &\quad\text{iff}\quad
  \SEMBTWO{\LLL} \models_{x \mapsto s} \exists^{st} y.(\ARROWTWO{}{x}{a}{y} \AND \SEMBTWO{\phi}_y) \\
     &\quad\text{iff}\quad
   \SEMBTWO{\LLL} \models_{x \mapsto s} \SEMBTWO{\MAY{a}{\phi}}_x
\end{alignat*}

Finally $!A$.
\begin{alignat*}{2}
  \MMM \models !A
     &\quad\text{iff}\quad
  \lambda (s) \subseteq A \\
      &\quad\text{iff}\quad
  \text{for all }\ a \in \Sigma. a \in A \\
     &\quad\text{iff}\quad
  \SEMBTWO{\LLL} \models_{x \mapsto s} \forall^{act} a. (\ALLOWED{}{x}{a} \IMPLIES a \in A) \\
     &\quad\text{iff}\quad
  \SEMBTWO{\LLL} \models_{x \mapsto s} \SEMBTWO{!A}_x
\end{alignat*}
\end{proof}

\NI We use the following steps in our compactness proof.

\begin{enumerate}

\item Choose a set $\Gamma$ of \cathoristic{} formulae such that each finite
  subset $\Gamma'$ of $\Gamma$ has a cathoristic model $(\LLL, s)$.

\item The translation gives a set $\SEMBTWO{\Gamma} =
  \{\SEMBTWO{\phi}\ |\ \phi \in \Gamma\}$ of first-order formulae such that
  each finite subset has a first-order model $\SEMBTWO{\LLL}$.

\item By compactness of (two-sorted) first-order logic, we can find a
  first-order model $\CAL{M}$ of $\SEMBTWO{\Gamma}$.

\item\label{compactness:step:4} Convert $\CAL{M}$ into a cathoristic transition system
  $\CAL{M}^{\sharp}$ such that $(\CAL{M}^{\sharp}, s) \models \Gamma$.

\end{enumerate}

\NI The problematic step is (\ref{compactness:step:4}) - for how would
we know that the first-order model $\CAL{M}$ can be converted back to
a cathoristic transition system? What if it contains `junk' in the
sense described above?  
We solve this by adding formulae to 
$\SEMBTWO{\Gamma}$ that preserve finite satisfiability but force the
first-order models to be convertible to cathoristic models.
 To ensure admissibility we use this formula.
\begin{eqnarray*}
   \phi_{admis} 
      & \ =\ &
   \forall^{st} s.\forall^{act} a.\forall^{st} t.( \ARROWTWO{}{s}{a}{t} \IMPLIES \ALLOWED{}{s}{a}) 
\end{eqnarray*}

\NI The formula $\phi_{det}$ ensures model determinism.
\begin{eqnarray*}
   \phi_{det} 
      & \ =\ &
   \forall^{st} s.\forall^{act} a.\forall^{st} t.\forall^{st} t'.
   ((\ARROWTWO{}{s}{a}{t}  \AND \ARROWTWO{}{s}{a}{t'} ) \IMPLIES t = t' )   
\end{eqnarray*}

\begin{lemma}\label{compactness:lemma:23399}
If $\LLL$ is a cathoristic transition system then $\SEMBTWO{\LLL} \models
\phi_{admis} \AND \phi_{det}$.
\end{lemma}

\begin{proof}
Straightforward from the definitions.
\end{proof}

We can now add, without changing satisfiability, $\phi_{admis}
\AND \phi_{det}$ to any set of first-order formulae that has a model
that is the translation of a cathoristic model.

We also need to deal with well-sizedness in first-order models,
because nothing discussed so far prevents models whose state labels are
infinite sets without being $\Sigma$.  Moreover, a model may interpret
the set of actions with a proper superset of $\Sigma$.  This also
prevents conversion to cathoristic models. We solve these problems by
simply removing all actions that are not in $\Sigma$ and all
transitions involving such actions.  We map all
infinite state labels to $\Sigma$. It is easy to see that this does not
change satisfiability of (translations of) cathoristic formulae.




\begin{definition}
Let $\LLL = (S, \rightarrow, \lambda)$ be a cathoristic transition system
and $X$ a set, containing actions. The \emph{restriction of
  $\LLL$ to $X$}, written $\LLL \setminus X$ is the cathoristic model $(S,
\rightarrow', \lambda')$ where $\rightarrow' = \{(s, a, t) \in
\rightarrow \ |\ a \notin X\}$, and for all states $s$ we set:
\[
   \lambda'(s) 
        =
   \begin{cases}
       \lambda(s) \setminus  X & \text{whenever}\ \lambda(s) \neq \Sigma \\
       \Sigma & \text{otherwise}
   \end{cases}
\]

\end{definition}

\begin{lemma}\label{compactness:lemma:1717}
Let $\phi$ be a \cathoristic{} formula and $X$ be a set such that no action
occurring in $\phi$ is in $X$. Then:
\[
   (\LLL, s) \models \phi
      \quad\text{iff}\quad
   (\LLL \setminus X, s) \models \phi.
\]
\end{lemma}
\begin{proof}
By straightforward induction on the structure of $\phi$, using the
fact that by assumption $X$ only contains actions not occurring in
$\phi$.  
\end{proof}

\begin{definition}
Let $\CAL{M}$ be a first-order model for the signature $\SSS'$.
We construct a cathoristic transition system
$\CAL{M}^{\sharp} = (S, \rightarrow, \lambda)$.
\begin{itemize}

\item The actions $\Sigma$ are given by the $\CAL{M}$ interpretation of actions.

\item The states $S$ are given by the $\CAL{M}$ interpretation of states.

\item The reduction relation $s \TRANS{a} t$ holds exactly when
  $\ARROWTWO{\CAL{M}}{s}{a}{t}$.

\item The function $\lambda$ is given by the following clause:
  \[
     \lambda(s) 
        =
     \begin{cases} 
       X & \text{whenever}\ X = \{a \ |\ \ALLOWED{\CAL{M}}{s}{a} \}\ \text{ is finite} \\
       \Sigma & \text{otherwise}
     \end{cases}
  \]

\end{itemize}

\end{definition}

\begin{lemma}
If $\CAL{M}$ be a first-order model for $\SSS'$ such that $\CAL{M}
\models \phi_{admis} \AND \phi_{det}$.  Then $\CAL{M}^{\sharp}$ is an
cathoristic transition system with actions $\Sigma$.
\end{lemma}
\begin{proof}
Immediate from the definitions.
\end{proof}

\begin{theorem}[correspondence theorem]\label{correspondence:theorem:223}
Let $\CAL{M}$ be a first-order model for the signature $\SSS'$ such
that $\CAL{M} \models \phi_{admis} \AND \phi_{det}$.  Then we have for
all \cathoristic{} formulae $\phi$ with actions from $\Sigma$:
\[
   \CAL{M} \models_{x \mapsto s} \SEMBTWO{\phi}_x 
        \quad  \text{iff} \quad 
   (\CAL{M}^{\sharp} \setminus X, s) \models \phi.
\]
\end{theorem}
Here $X$ is the set of all elements in the universe of $\CAL{M}$ interpreting
actions that are not in $\Sigma$.
\begin{proof}
The proof proceeds by induction on the structure of $\phi$. 
\end{proof}

\begin{definition}
Let $\Gamma$ be a set of cathoristic formulae, and $\MMM$ a cathoristic model.  We
write $\MMM \models T$ provided $\MMM \models \phi$ for all $\phi \in
T$.  We say $\Gamma$ is \emph{satisfiable} provided $\MMM \models T$.
\end{definition}

\begin{theorem}[Compactness of \cathoristic{}]
A set $\Gamma$ of \cathoristic{} formulae is satisfiable iff each finite subset of
$\Gamma$ is satisfiable.
\end{theorem}
\begin{proof}
For the non-trivial direction, let $\Gamma$ be a set of \cathoristic{} formulae
such that any finite subset has a cathoristic model. Define 
\[
  \SEMBTWO{\Gamma} 
     \ =\ 
  \{\SEMBTWO{\phi}\ |\ \phi \in \Gamma\} 
     \qquad\qquad
  \Gamma^*
     \ =\ 
  \SEMBTWO{\Gamma} \cup \{\phi_{admis} \AND  \phi_{det}\}
\]
which both are sets of first-order formulae. Clearly each finite subset $\Gamma'$ of 
$\Gamma^*$ has a first-order model. Why? First consider the subset $\Gamma'_{CL}$ of $\Gamma'$
which is given as follows.
\[
   \Gamma'_{CL} \ =\ \{ \phi \in \Gamma\ |\ \SEMBTWO{\phi} \in \Gamma' \}
\]
Since $\Gamma'_{CL}$ is finite, by assumption there is a cathoristic model 
\[
   (\LLL, s) \models \Gamma'_{CL}
\]
which means we can apply Theorem \ref{correspondence:theorem:223} to get
\[
   \SEMBTWO{\LLL} \models_{x \mapsto s} \SEMBTWO{\Gamma'_{CL}},
\]
By construction $\Gamma' \setminus \SEMBTWO{\Gamma'_{CL}} \subseteq
\{\phi_{admis} \AND \phi_{det}\}$, so all we have to
show for $\Gamma'$ to have a model is that
\[
    \SEMBTWO{\LLL} \models_{x \mapsto s} \{\phi_{admis}\} \cup \{ \phi_a\ |\ a \in \Sigma\},
\]
but that is a direct consequence of Lemma
\ref{compactness:lemma:23399}.  That
means each finite subset of $\Gamma^*$ has a model and by appealing to
compactness of first-order many-sorted logic (which is an immediate
consequence of compactness of one-sorted first-order logic
\cite{EndertonHB:matinttl}), we know there must be a first-order model
$\CAL{M}$ of $\Gamma^*$, i.e.
\[
   \CAL{M} \models \Gamma^*.
\]
Since $\CAL{M} \models \phi_{admis} \AND \phi_{det}$ we can apply
Theorem \ref{correspondence:theorem:223} that also
\[
   (\CAL{M}^{\sharp} \setminus X, s) \models \Gamma
\]
where $X$ is the set of all actions in $\CAL{M}^{\sharp}$ that are not
in $\Sigma$. Hence $\Gamma$ is satisfiable. 
\end{proof}

\section{\Cathoristic{} and negation}\label{ELAndNegation}

\NI We have presented \cathoristic{} as a language that can express
incompatible claims without negation.  In this section, we briefly
consider \cathoristic{} enriched with negation.

\subsection{Syntax and semantics}

\begin{definition}
Given a set $\Sigma$ of actions, the \emph{formulae of \cathoristic{}
  with negation} are given by the following grammar.
\begin{GRAMMAR}
  \phi 
     &\quad ::= \quad & 
   ... \fOr \neg \phi
\end{GRAMMAR}

\NI We can now define disjunction $\phi \lor \psi$ and implication
$\phi \IMPLIES \psi$ by de Morgan duality: $\phi \OR \psi$ is short
for $\neg (\neg \phi \land \neg \psi )$, and $\phi \IMPLIES \psi$  abbreviates
$\neg\phi \OR \psi$.
\end{definition}

The semantics of \cathoristic\ with negation is just that of plain
\cathoristic\, except for the obvious clause for negation.
\begin{eqnarray*}
\MMM \models \neg \phi &\quad\mbox{ iff }\quad& \MMM \nvDash \phi  
\end{eqnarray*}

\NI Negation is a core operation of classical logic, and its absence makes
\cathoristic{} unusual. In order to understand \cathoristic{} better, we
now investigate how negation can be seen as a definable abbreviation
in \cathoristic{} with disjunction. The key idea is to use the fact that
\[
   \neg \MAY{a}{\phi}
\]
can be false in two ways: either there is no $a$-labelled action at
the current state - or there is, but $\phi$ is false. Both arms of
this disjunction can be expressed in \cathoristic{}, the former as
$!\Sigma \setminus \{a\}$, the latter as $\MAY{a}{\neg \phi}$.
 Hence, we can see $\neg \MAY{a}{\phi}$ as a shorthand for 
\[
   !(\Sigma \setminus \{a\}) \OR \MAY{a}{\neg \phi}
\]

\NI Negation still occurs in this term, but prefixing a formula of
lower complexity.

This leaves the question of negating the tantum. That's easy: when
$\neg !A$, then clearly the current state can do an action $a \notin
A$. In other words
\[
   \BIGOR_{a \in \Sigma}\MAY{a}{\TRUE}
\]

\NI When $\Sigma$ is infinite, then so is the disjunction.

Note that both the negation of the modality and the negation of
the tantum involve the set $\Sigma$ of actions. 
So far, we have defined negation with respect to
the whole (possibly infinite) set $\Sigma$. For technical reasons, we
generalise negation and define it with respect to a \emph{finite}
subset $S \subseteq \Sigma$. We use this finitely-restricted version of
negation in the decision procedure below.

\begin{definition}
The function $\neg_{S}(\phi)$ removes negation from $\phi$
relative to a finite subset $S \subseteq \Sigma$:

\begin{align*}
  \neg_{S}(\top) &\ =\  \bot  &
  \neg_S(\bot) &\ =\  \top  \\
  \neg_S(\phi \land \psi) &\ =\  \neg_S(\phi) \lor \neg_S(\psi)  &
  \neg_S(\phi \lor \psi) &\ =\  \neg_S(\phi) \land \neg_S(\psi)  \\
  \neg_S(\langle a \rangle \phi) &\ =\  \fBang(S-\{a\}) \lor \langle a \rangle \neg_S(\phi)  &
  \neg_S(\fBang A) &\ =\  \bigvee_{a \in S - A} \langle a \rangle \top
\end{align*}

\end{definition}

\subsection{Decision procedure}

\NI We can use the fact that \cathoristic\ has a quadratic-time
decision procedure to build a super-polynomial time decision procedure for
\cathoristic\ with negation.  Given $\phi \models \psi$, let $S =
\ACTIONS{\phi} \cup \ACTIONS{\psi} \cup \{a\}$,   where $a$ is a
fresh action.  The function $\ACTIONS{\cdot}$ returns all actions
occurring in a formula, e.g. $\ACTIONS{\MAY{a}{\phi}} = \{a\} \cup
\ACTIONS{\phi}$ and $\ACTIONS{!A} = A$. The decision procedure executes the following steps.

\begin{enumerate}

\item Inductively translate away all negations in $\phi$ using
  $\neg_S(\phi)$ as defined above.  Let the result be $\phi'$.

\item Reduce $\phi'$ to disjunctive normal form by repeated
  application of the rewrite rules:
  \[
    \phi \land (\psi \lor \xi)  \ \leadsto \ (\phi \land \psi) \lor (\phi \land \xi)  
       \qquad
    (\phi \lor \psi) \land \xi  \ \leadsto \ (\phi \land \xi) \lor (\psi \land \xi). 
  \]

\item Let the resulting disjuncts be $\phi_1, ..., \phi_n$. 
Note that
\[
\phi \models \psi \quad\text{ iff }\quad \phi_i \models \psi \text{ for all } i = 1, ..., n.
\]
For each disjunct $\phi_i$ do the following.
\begin{itemize}

\item Notice that $\phi_i \models \psi$ if and only if all
  $S$-extensions (defined below) of $\SIMPL{\phi_i}$ satisfy $\psi$.
  So, to check whether $\phi_i \models \psi$, we enumerate the
  $S$-extensions of $\SIMPL{\phi_i}$ (there are a finite number of
  such extensions - the exact number is exponential in the size of
  $\SIMPL{\phi_i}$) and check for each such $S$-extension $\MMM$
  whether $\MMM\models \psi$, using the algorithm of Section
  \ref{decisionprocedure}.

\end{itemize}
\end{enumerate}

\NI  Here is the definition of $S$-extension.

\begin{definition}
Given an cathoristic transition system $\LLL =
(\mathcal{W},\rightarrow,\lambda)$, and a set $S$ of actions, then
$(\mathcal{W'},\rightarrow',\lambda')$ is a \emph{$S$-extension} of
$\LLL$ if it is a valid cathoristic transition system (recall
Definition \ref{cathoristicTS}) and for all $(x,a,y) \in
\rightarrow'$, either:
\begin{itemize} 

\item $(x, a, y) \in\ \rightarrow$,  or;

\item $x \in \mathcal{W}$, $a \in S$, $a \in \lambda(x)$, and $y$ is a new state not
  appearing elsewhere in $\mathcal{W}$ or $\mathcal{W'}$.

\end{itemize}
\end{definition}
The state-labelling $\lambda'$ is:
\begin{eqnarray*}
\lambda'(x) & = & \lambda(x) \text{ if } x \in \mathcal{W} \\
\lambda'(x) & = & \Sigma \text{ if } x \notin \mathcal{W} \\
\end{eqnarray*}

\NI In other words, $\MMM'$ is an extension of an annotated model
$\MMM$, if all its transitions are either from $\MMM$ or involve
states of $\MMM$ transitioning via elements of $S$ to new states not
appearing in $\MMM$ or $\MMM'$.  The number of extensions grows
quickly.  If the model $\MMM$ has $n$ states, then the number of
possible extensions is:
\[
   ({2^{|S|}})^n
\] 

\NI But recall that we are computing these extensions in order to
verify $\psi$. So we can make a significant optimisation by
restricting the height of each tree to $|\psi|$. We state, without proof, that this optimisation preserves correctness.
A Haskell implementation of the decision procedure is available
\cite{HaskellImplementation}.

\section{Quantified \cathoristic{}}\label{quantifiedEL}

\NI So far, we have presented \cathoristic\ as a propositional modal
logic.  This section sketches quantified \cathoristic, primarily to
demonstrate that this extension works smoothly. 

\begin{definition} 
Let $\Sigma$ be a non-empty set of actions, ranged over by $a, a',
...$ as before.  Given a set $\VVV$ of \emph{variables}, with
$x, x', y, y', ...$ ranging over $\VVV$, the \emph{terms},
ranged over by $t, t', ...$ and formulae of quantified \cathoristic{}
are given by the following grammar:

\begin{GRAMMAR}
  t
     &\quad ::= \quad & 
  x
     \VERTICAL 
  a
  \\[1mm]
  \phi 
     &\quad ::= \quad & 
  \TRUE 
     \VERTICAL 
  \phi \AND \psi
     \VERTICAL 
  \MAY{t}{\phi}
     \VERTICAL 
  \fBang A 
     \VERTICAL 
  \exists x . {\phi}
     \VERTICAL 
  \forall x . {\phi}
\end{GRAMMAR}

\NI Now $A$ ranges over finite subsets of terms. The \emph{free
  variables} of a $\phi$, denoted $\FV{\phi}$ is given as expected,
e.g.~$\FV{\MAY{t}{\phi}} = \FV{t} \cup \FV{\phi}$ and $\FV{!A} =
\bigcup_{t \in A}\FV{t}$ where $\FV{a} = \emptyset$ and $\FV{x} =
\{x\}$.
\end{definition}

\begin{definition}
The semantics of quantified \cathoristic{} is constructed along
conventional lines. An \emph{environment} is a map $\sigma : \VVV
\rightarrow \Sigma$ with finite domain.  We write $\sigma, x : a$ for
the environment that is just like $\sigma$, except it also maps $x$ to
$a$, implicitly assuming that $x$ is not in $\sigma$'s domain.  The
\emph{denotation} $\SEMB{t}_{\sigma}$ of a term $t$ under an
environment $\sigma$ is given as follows:
\[
   \SEMB{a}_{\sigma} = a
      \qquad\qquad
   \SEMB{x}_{\sigma} = \sigma(x)
\]
where we assume that $\FV{t}$ is a subset of the domain of $\sigma$.

The \emph{satisfaction
  relation} $\MMM \models_{\sigma} \phi$ is defined whenever
$\FV{\phi}$ is a subset of $\sigma$'s domain. It is given by the
following clauses, where we assume that $\MMM = (\LLL, s)$ and $\LLL =
(S, \rightarrow, \lambda)$.

\[
\begin{array}{lclcl}
  \MMM & \models_{\sigma} & \top   \\
  \MMM & \models_{\sigma} & \phi \AND \psi &\ \mbox{ iff } \ & \MMM  \models_{\sigma} \phi \mbox { and } \MMM \models_{\sigma} \psi  \\
  \MMM & \models_{\sigma} & \langle t \rangle \phi & \mbox{ iff } & \text{there is transition } s \TRANS{\SEMB{t}_{\sigma}} s' \mbox { such that } (\LLL, s') \models_{\sigma} \phi  \\
  \MMM & \models_{\sigma} & \fBang A &\mbox{ iff } & \lambda(s) \subseteq \{\SEMB{t}\ |\ t \in A\} \\
  \MMM & \models_{\sigma} & \forall x.\phi &\mbox{ iff } & \text{for all} \ a \in \Sigma\ \text{we have}\ \MMM \models_{\sigma, x : a} \phi \\
  \MMM & \models_{\sigma} & \exists x.\phi &\mbox{ iff } & \text{there exists} \ a \in \Sigma \ \text{such that}\  \MMM \models_{\sigma, x : a} \phi
\end{array}
\]

\end{definition}

\NI In quantified \cathoristic{}, we can say that there is exactly one
king of France, and he is bald, as:
\[
   \exists x . (\MAY{king} \MAY{france} ! \{x\} \land \MAY{x} \MAY{bald})
\]

\NI Expressing this in \fol{} is more cumbersome:
\[
   \exists x. ( king(france, x) \land bald(x) \land \forall y. ( king(france, y) \rightarrow y = x ))
\]

\NI The \fol{} version uses an extra universal quantifier, and also
requires the identity relation with concomitant axioms.

To say that every person has exactly one sex, which is either male or
female, we can write in quantified \cathoristic{}:
\[
   \forall x . 
      ( \MAY{x} \MAY{person} \rightarrow \MAY{x} \MAY{sex} !\{male, female\} 
      \land 
      \exists y . \MAY{x} \MAY{sex} (\MAY{y} \land
   \fBang \{y\}) )
\]

\NI This is more elegant than the equivalent in \fol{}:
\[
   \forall x. ( person(x) \rightarrow \exists y .
   \left(
      \begin{array}{l}
        sex(x,y) \\
        \quad\land\\
        (y = male \; \lor \; y = female)\\ 
        \quad\land\\
        \forall z . sex(x,z) \rightarrow    y = z 
      \end{array}
   \right))
\]

\NI To say that every traffic light is coloured either green, amber or
red, we can write in quantified \cathoristic{}:
\[
   \forall x. (\MAY{x} \MAY{light} \rightarrow \MAY{x} \MAY{colour}
   !\{green, amber, red\} \land \exists y . \MAY{x} \MAY{colour}
   (\MAY{y} \land !\{y\}))
\]

\NI Again, this is less verbose than the equivalent in
\fol{}:
\[
   \forall x. ( light(x) \rightarrow \exists y .
   \left(
      \begin{array}{l}
        colour(x,y) \\
        \quad\land\\
        (y = green \; \lor \; y = amber \; \lor \; y = red) \\
        \quad\land \\
        \forall z . colour(x,z) \rightarrow y = z
      \end{array}
   \right)  )
\]

\section{Related work}\label{relatedWork}


This section surveys  \cathoristic's intellectual background, and related
approaches.


\subsection{Brandom's incompatibility semantics}
\NI In \cite{brandom}, Chapter 5, Appendix I, Brandom developed a new
type of semantics, incompatibility semantics, that takes material
incompatibility - rather than truth-assignment - as the semantically
primitive notion.

Incompatibility semantics applies to any language, $\mathcal{L}$,
given as a set of sentences.  Given a predicate $\mathsf{Inc}(X)$
which is true of sets $X \subseteq \mathcal{L}$ that are incompatible,
he defines an incompatibility function $\mathcal{I}$ from subsets of
$\mathcal{L}$ to sets of subsets of $\mathcal{L}$:
\[
X \in \mathcal{I}(Y) \quad\text{ iff }\quad \mathsf{Inc}(X \cup Y).
\]
We assume that $\mathcal{I}$ satisfies the
monotonicity requirement (Brandom calls it ``Persistence''):
\[
   \text{If } X \in \mathcal{I}(Y) \text{ and } X \subseteq X' \text{ then } X' \in \mathcal{I}(Y).
\]

\NI Now Brandom defines entailment in terms of the incompatibility
function. Given a set $X \subseteq \mathcal{L}$ and an individual
sentence $\phi \in \mathcal{L}$:

\[
   X \models \phi\quad \text{ iff }\quad \mathcal{I}(\{\phi\}) \subseteq \mathcal{I}(X).
\]

\NI Now, given material incompatibility (as captured by the
$\mathcal{I}$ function) and entailment, he introduces logical negation
as a \emph{derived} concept via the rule:

\[
   \{\neg \phi\} \in \mathcal{I}(X)\quad \text{ iff }\quad X \models \phi.
\]

\NI Brandom goes on to show that the $\neg$ operator, as defined, satisfies
the laws of classical negation.  He also introduces a modal operator,
again defined in terms of material incompatibility, and shows that
this operator satisfies the laws of $S5$.

\Cathoristic{} was inspired by Brandom's vision that material
incompatibility is conceptually prior to logical negation: in other
words, it is possible for a community of language users to make incompatible claims, even if that
language has no explicit logical operators such as negation.  The
language users of this simple language may go on to introduce logical
operators, in order to make certain inferential properties explicit -
but this is an optional further development.  The language before that
addition was already in order as it is.

The approach taken in this paper takes Brandom's original insight in a
different direction.  While Brandom defines an unusual (non
truth-conditional) semantics that applies to any language, we have
defined an unusual logic with a standard (truth-conditional) semantics, and then shown that this logic satisfies the Brandomian connection between incompatibility and entailment.

\subsection{Peregrin on defining a negation operator}\label{peregrin}

Peregrin \cite{PeregrinJ:logbasoi} investigates the  structural
rules that any logic must satisfy if it is to connect incompatibility
($\mathsf{Inc}$) and entailment ($\models$) via the Brandomian
incompatibility semantics constraint:
\[
X \models \phi \quad\text{ iff }\quad \mathcal{I}(\{\phi\}) \subseteq \mathcal{I}(X).
\]

\NI The general structural rules are:
\begin{eqnarray*}
  (\bot) & & \text{If } \mathsf{Inc}(X) \text{ and } X \subseteq Y \text{ then } \mathsf{Inc}(Y). \\
  (\models 1) & & \phi, X \models \phi. \\
  (\models 2) & & \text{If }X, \phi \models \psi \text{ and } Y \models \phi \text{ then } X, Y \models \psi. \\
  (\bot \models 2) & & \text{If } X \models \phi \text{ for all } \phi, \text{ then } \mathsf{Inc}(X). \\
  (\models \bot 2) & & \text{If } \mathsf{Inc}(Y \cup \{\phi\}) \text{ implies } \mathsf{Inc}(Y \cup X) \text{ for all } Y, \text{ then } X \models \phi.
\end{eqnarray*}

\NI Peregrin shows that if a logic satisfied the above laws, then
incompatibility and entailment are mutually interdefinable, and the
logic satisfies the Brandomian incompatibility semantics constraint.

Next, Peregrin gives a pair of laws for defining negation in terms
of $\mathsf{Inc}$ and $\models$\footnote{The converse of $(\neg 2)$
  follows from $(\neg 1)$ and the general structural laws above.}:
\begin{eqnarray*}
  (\neg 1) & & \mathsf{Inc}(\{\phi, \neg \phi\}). \\
  (\neg 2) & & \text{If } \mathsf{Inc}(X, \phi) \text{ then } X \models \neg \phi.
\end{eqnarray*}

\NI These laws characterise intuitionistic negation as the
\emph{minimal incompatible}\footnote{$\psi$ is the minimal
  incompatible of $\phi$ iff for all $\xi$, if $\mathsf{Inc}(\{\phi\}
  \cup \{\xi\})$ then $\xi \models \psi$.}.  
  Now, in \cite{brandom},
Brandom defines negation slightly differently. He uses the rule:
\begin{eqnarray*}
  (\neg B) & &\mathsf{Inc}(X, \neg \phi) \text{ iff } X \models \phi.
\end{eqnarray*}
Using this stronger rule, we can infer the classical law of
double-negation: $\neg \neg \phi \models \phi$.  Peregrin establishes
that Brandom's rule for negation entail $(\neg 1)$ and $(\neg 2)$
above, but not conversely: Brandom's rule is stronger than Peregrin's
minimal laws $(\neg 1)$ and $(\neg 2)$.

Peregrin concludes that the Brandomian constraint between
incompatibility and entailment is satisfied by many different logics.
Brandom  happened to choose a particular rule for negation
that led to classical logic, but the general connection between
incompatibility and entailment is satisfied by many different logics,
including intuitionistic logic.  This paper supports Peregrin's
conclusion: we have shown that \cathoristic{} also satisfies the
Brandomian constraint.

\subsection{Peregrin and Turbanti on defining a necessity operator}\label{peregrinTurbanti}

In \cite{brandom}, Brandom gives a rule for defining necessity in terms of incompatibility and entailment:
\[
X \in \mathcal{I}(\{\Box \phi\}) \quad\text{ iff }\quad \mathsf{Inc}(X) \lor \exists Y . \; Y \notin \mathcal{I}(X) \land Y \nvDash \phi.
\]
In other words, $X$ is incompatible with $\Box \phi$ if $X$ is compatible with something that does not entail $\phi$.

The trouble is, as Peregrin and Turbanti point out, if $\phi$ is not tautological, then \emph{every set} $X \subseteq \mathcal{L}$ is incompatible with $\Box \phi$.
To show this, take any set $X \subseteq \mathcal{L}$. 
If $\mathsf{Inc}(X)$, then $X \in \mathcal{I}(\Box \phi)$ by definition.
If, on the other hand, $\neg \mathsf{Inc}(X)$, then let $Y = \emptyset$.
Now $\neg \mathsf{Inc}(X \cup Y)$ as $Y = \emptyset$, and $Y \nvDash \phi$ as $\phi$ is not tautological.
Hence $X \in \mathcal{I}(\Box \phi)$ for all $X \subseteq \mathcal{L}$. 
Brandom's rule, then, is only capable of specifying a very specific form of necessity: logical necessity.

In \cite{PeregrinJ:logbasoi} and \cite{turbanti}, Peregrin and Turbanti describe alternative ways of defining necessity.
These alternative rule sets can be used to characterise modal logics other than S5.
For example, Turbanti defines the accessibility relation between worlds in terms of a \emph{compossibility relation}, and then argues that the S4 axiom of transitivity fails because compossibility is not transitive.

We draw two conclusions from this work.
The first is, once again, that a commitment to connecting incompatibility and entailment via the Brandomian constraint:
\[
X \models \phi\quad \text{ iff }\quad \mathcal{I}(\{\phi\}) \subseteq \mathcal{I}(X)
\]
does not commit us to any particular logical system. 
There are a variety of logics that can satisfy this constraint.
Second, questions about the structure of the accessibility relation in Kripke semantics - questions that can seem hopelessly abstract and difficult to answer - can be re-cast in terms of concrete questions about the incompatibility relation.
Incompatibility semantics can shed light on possible-world semantics \cite{turbanti}. 

\subsection{Linear logic}

Linear logic \cite{GirardJY:linlog} is a refinement of
first-order logic and was introduced by J.-Y.~Girard and
brings the symmetries of classical logic to constructive
logic. 

Linear logic splits conjunction into additive and multiplicative
parts. The former, additive conjunction $A \& B$, is especially
interesting in the context of \cathoristic{}. In the terminology of
process calculus it can be interpreted as an external choice operation
\cite{AbramskyS:comintoll}. (`External', because the choice is offered to
the environment).  This interpretation has been influential in the
study of types for process calculus,
e.g.~\cite{HondaK:lanpriatdfscbp,HondaK:unitypsfsifLONG,TakeuchiK:intbaslaits}.
Implicitly, additive conjunction gives an explicit upper bound on how
many different options the environment can choose from. For example 
$A \& B \& C$ has  three options (assuming that none of $A, B, C$
can be decomposed into further additive conjunctions).  With this in
mind, and simplifying a great deal, a key difference between $!A$ and
additive conjunction $A \& B$ is that the individual actions in $!A$
have no continuation, while they do with $A \& B$: the tantum $!\{l, r\}$ says
that the only permitted  actions are $l$ and $r$. What
happens at later states is not constrained by $!A$.  In contrast, $A \&
B$ says not only that at this point the only permissible options are $A$
and $B$, but also that if we choose $A$, then $A$ holds `for ever',
and likewise for choosing $B$. To be sure, the alternatives in $A \&
B$ may themselves contain further additive conjunctions, and in this
way express how exclusion changes 'over time'.

In summary, \cathoristic{} and linear logic offer  operators that restrict
the permissible options. How are they related? Linear logic has an
explicit linear negation $(\cdot)^{\bot}$ which, unlike classical
negation, is constructive. In contrast, \cathoristic{} defines a restricted
form of negation using $!A$. Can these two perspectives be fruitfully
reconciled?

\subsection{Process calculus}

Process calculi are models of concurrent computation.  They are based
on the idea of message passing between actors running in parallel.
Labelled transition systems are often used as models for process
calculi, and many concepts used in the development of \cathoristic{} -
for example, bisimulations and Hennessy-Milner logic - originated in
process theory (although some, such as bisimulation, evolved
independently in other contexts).

Process calculi typically feature a construct called sum, that is an
explicit description of mutually exclusive option:
\[
     \sum_{i \in I} P_i
\]
That is a process that can internally choose, or be chosen externally
by the environment to evolve into the process $P_i$ for each $i$. Once
the choice is made, all other options disappear.  Sums also relate
closely to linear logic's additive conjunction. Is this conceptual
proximity a coincidence or indicative of deeper common structure?

\subsection{Linguistics}

Linguists have also investigated how mutually exclusive alternatives
are expressed, often in the context of antonymy
\cite{AllanK:conencos,AronoffM:hanlin,OKeeffeA:rouhanocl}, but, to the
best of our knowledge have not proposed formal theories of linguistic
exclusion.

\section{Open problems}\label{conclusion}

In this paper, we have introduced \cathoristic{} and established key
meta-logical properties. However,
many questions are left open. 

\subsection{Excluded middle}

One area we would like to investigate further is what happens to the law of excluded middle in \cathoristic{}.
The logical law of excluded middle states that either a proposition or
its negation must be true. In \cathoristic{}
\[
\models \phi \lor \neg_S(\phi)
\]

\NI does not hold in general. (The negation operator $\neg_{S}(\cdot)$
was defined in Section \ref{ELAndNegation}.) For example, let $\phi$ be
$\langle a \rangle \top$ and $S = \Sigma = \{a, b\}$.  Then
\[
   \phi \lor \neg_{S} \phi 
       \quad=\quad 
   \langle a \rangle \top \; \lor \; ! \{b\} \; \lor \; \langle a \rangle \bot
\]

\NI Now this will not in general be valid - it will be false for
example in the model $((\{x\}, \emptyset, \{(x, \Sigma)\}), x)$, the
model having just the start state (labelled $\Sigma$) and no transitions.
Restricting $S$ to be a proper subset of $\Sigma = \{a, b\}$ is also not
enough. For example with $S = \{a\}$ we have
\[
   \MAY{a}{\TRUE} \lor \neg_{S}(\MAY{a}{\TRUE})
      \quad=\quad
   \MAY{a}{\TRUE}\ \lor\ !\emptyset \lor \MAY{a}{\FALSE}
\]
This formula cannot hold in any cathoristic model which contains a
$b$-labelled transition, but no $a$-transition from the start state.

Is it possible to identify classes of models that nevertheless verify
excluded middle? The answer to this question appears to depend 
on the chosen notion of semantic model.

\subsection{Understanding the expressive strength of \cathoristic{}}

\subsubsection{Comparing \cathoristic{} and Hennessy-Milner logic}

Section \ref{standardTranslation} investigated the relationship
between \cathoristic{} and first-order logic. Now we compare
\cathoristic{} with a logic that is much closer in spirit:
Hennessy-Milner logic \cite{HennessyM:alglawfndac}, a multi-modal logic
designed to reason about process calculi. Indeed, the present shape of
\cathoristic{} owes much to Hennessy-Milner logic. We contrast both by
translation from the former into the latter.  This will reveal, more
clearly than the translation into first-order logic, the novelty of
\cathoristic{}.

\begin{definition}
Assume a set $\Sigma$ of symbols, with $s$ ranging over
$\Sigma$, the \emph{formulae} of Hennessy-Milner logic are given
by the following grammar:
\begin{GRAMMAR}
  \phi 
     &\quad ::= \quad & 
  \top \fOr \BIGAND_{i \in I} \phi_i  \fOr \langle s \rangle \phi \fOr \neg \phi 
\end{GRAMMAR}
\end{definition}

\NI The index set $I$ in the conjunction can be infinite, and needs to
be so for applications in process theory.

\begin{definition}
 \emph{Models} of Hennessy-Milner logic are simply pairs $(\LLL, s)$
 where $\LLL = (S, \rightarrow)$ is a labelled transition system over
 $\Sigma$, and $s \in S$.  The \emph{satisfaction relation} $(\LLL, s)
 \models \phi$ is given by the following inductive clauses.

\[
\begin{array}{lclcl}
  (\LLL, s) 
     & \models & 
  \top  \\
  (\LLL, s) 
     & \models & 
  \BIGAND_{i \in I} \phi_i  &\  \mbox{ iff }\  & \mbox { for all $i \in I$ }: (\LLL, s) \models \phi_i  \\
  (\LLL, s) 
     & \models & 
  \langle a \rangle \phi & \mbox{ iff } & \mbox{ there is a } s \xrightarrow{a} s' \mbox { such that } (l,s') \models \phi  \\
  (\LLL, s) 
     & \models & 
  \neg \phi & \mbox{ iff } & (\LLL, s)  \nvDash \phi 
\end{array}
\]
\end{definition}

\NI There are two differences between \cathoristic{} and
Hennessy-Milner logic - one syntactic, the other semantic.

\begin{itemize}

\item Syntactically, \cathoristic{} has the tantum operator ($!$) instead of
  logical negation ($\neg$).

\item Semantically,  cathoristic models are deterministic,
  while (typically) models of Hennessy-Milner logic are
  non-deterministic (although the semantics makes perfect sense for
  deterministic transition systems, too). 
  Moreover, models of Hennessy-Milner logic lack state labels.

\end{itemize}

\begin{definition}
 We translate formulae of \cathoristic{} into Hennessy-Milner logic
 using the function $\SEMB{\cdot}$:

\begin{eqnarray*}
  \SEMB{\top} & \ = \ & \top  \\
  \SEMB{\phi_1 \AND \phi_2} & \ = \ & \SEMB{\phi_1} \AND \SEMB{\phi_2}  \\
  \SEMB{\langle a \rangle \phi} & \ = \ & \langle a \rangle \SEMB{\phi}  \\
  \SEMB{! A} & \ = \ & \bigwedge_{a \in \Sigma \setminus A}\!\!\!\! \neg \langle a \rangle \top 
\end{eqnarray*}

\end{definition}

\NI If $\Sigma$ is an infinite set, then the translation of a $!$-formula 
will be an infinitary conjunction.  If $\Sigma$ is finite, then
the size of the Hennessy-Milner logic formula will be of the order of $n \cdot | \Sigma |$
larger than the original cathoristic formula, where $n$ is the number of
tantum operators occurring in the cathoristic formula). In both logics we
use the number of logical operators as a measure of size.

We can also translate cathoristic models by forgetting state-labelling:
\[
   \SEMB{((S, \rightarrow, \lambda), s)} 
      =
   ((S, \rightarrow), s)
\]

\NI We continue with an obvious consequence of the translation.

\begin{theorem}
Let $\MMM$ be a (deterministic or non-deterministic) cathoristic
  model. Then $\MMM \models \phi$ implies $\SEMB{\MMM} \models
  \SEMB{\phi}$.
\end{theorem}
\begin{proof}
Straightforward by induction on $\phi$.
\end{proof}

\NI However, note that the following natural extension is \emph{not} true
under the translation above:
\[
  \text{If } \phi \models \psi \text{ then } \SEMB{\phi} \models \SEMB{\psi}
\]
To see this, consider an entailment which relies on determinism, such as
\[
\MAY{a} \MAY{b} \land \MAY{a} \MAY{c} \models \MAY{a} (\MAY{b} \land \MAY{c})
\]

\NI The first entailment is valid in \cathoristic\ because of the
restriction to deterministic models, but not in Hennessy-Milner
logic, where it is invalidated by any model with two outgoing $a$
transitions, one of which satisfies $\MAY{b}$ and one of which
satisfies $\MAY{c}$.

We can restore the desired connection between cathoristic implication and
Hennessy-Milner logic implication in two ways. First we can restrict
our attention to deterministic models of Hennessy-Milner logic.  The
second solution is to add a determinism constraint to our
translation. Given a set $\Gamma$ of cathoristic formulae, closed under
sub formulae, that contains actions from the set $A \subseteq \Sigma$,
let the determinism constraint for $\Gamma$ be:
\[
\bigwedge_{a \in A, \phi \in \Gamma, \psi \in \Gamma} \neg \; (\MAY{a} \phi \land \MAY{a} \psi \land \neg \MAY{a} (\phi \land \psi) )
\]
If we add this sentence as part of our translation $\SEMB{\cdot}$, we
do get the desired result that
\[
\text{If } \phi \models \psi \text{ then } \SEMB{\phi} \models \SEMB{\psi}
\]

\subsubsection{Comparing \cathoristic{} with Hennessy-Milner logic and propositional logic}

Consider the following six languages:

\begin{FIGURE}
\begin{tikzpicture}[node distance=3.3cm,>=stealth',bend angle=45,auto]
  \tikzstyle{place}=[circle,thick,draw=blue!75,fill=blue!20,minimum size=6mm]
  \tikzstyle{red place}=[place,draw=red!75,fill=red!20]
  \tikzstyle{transition}=[rectangle,thick,draw=black!75,
  			  fill=black!20,minimum size=4mm]
  \tikzstyle{every label}=[red]
  \begin{scope}  
    \node [place] (w1) {PL[$\land$]};
    \node [place] (e1) [below of=w1] {PL [$\land, \neg$] };
  \end{scope}
  \begin{scope}[xshift=4cm]
    \node [place] (w1) {HML[$\land$]};
    \node [place] (e1) [below of=w1] {HML [$\land, \neg$] };
  \end{scope} 
  \begin{scope}[xshift=8cm]
    \node [place] (w1) {CL[$\land, !$]};
    \node [place] (e1) [below of=w1] {CL [$\land, !, \neg$] };
  \end{scope}
  \draw (2,0) node {$\subseteq $};
  \draw (6,0) node {$\subseteq$};
  \draw (2,-3) node {$\subseteq $};
  \draw (6,-3) node {$\subseteq$};
  \draw (0,-1.5) node {$\subseteq $};
  \draw (4,-1.5) node {$\subseteq $};
  \draw (8,-1.5) node {$\subseteq $};
\end{tikzpicture}
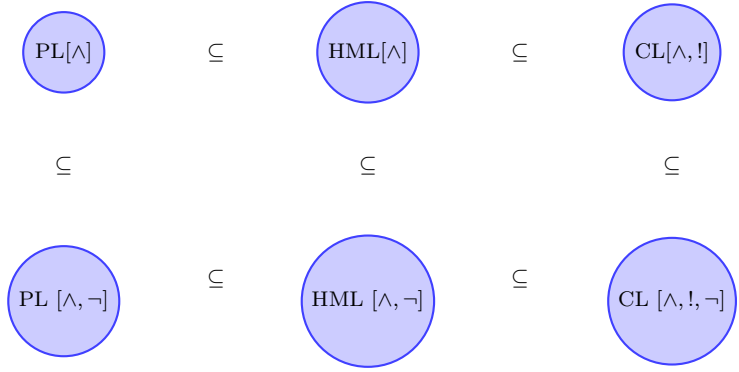
\caption{Conjectured relationships of expressivity between
  logics. Here $L_1 \subseteq L_2$ means that the logic $L_2$ is more
  expressive than $L_1$. We leave the precise meaning of logical
  expressivity open.}\label{figure:relationships}
\end{FIGURE}

\begin{center}
\begin{tabular}{ l | r }
Language & Description \\
\hline
PL[$\land$] & Propositional logic without negation \\
Hennessy-Milner logic[$\land$] & Hennessy-Milner logic without negation \\
CL[$\land, !$] & \Cathoristic{} \\
PL [$\land, \neg$] & Full propositional logic \\
HML [$\land, \neg$] & Full Hennessy-Milner logic \\
CL [$\land, !, \neg$] & \Cathoristic{} with negation\\
\end{tabular}
\end{center}

\NI The top three languages are simple. In each case: there is no
facility for expressing disjunction, every formula that is satisfiable
has a simplest satisfying model, and there is a simple quadratic-time
decision procedure But there are two ways in which CL[$\land, !$] is
more expressive.  Firstly, CL[$\land, !$], unlike HML[$\land$], is expressive enough to be able to distinguish
between any two models that are not bisimilar, cf.~Theorem
\ref{hennessymilnertheorem}.  The second way in which
CL[$\land, !$] is significantly more expressive than both PL[$\land$]
and HML[$\land$] is in its ability to express incompatibility.  No two
formulae of PL[$\land$] or HML[$\land$] are incompatible\footnote{The notion of incompatibility applies to all logics: two formulae are incompatible if there is no model which satisfies both.} with each
other.  But many
pairs of formulae of CL[$\land, !$] are incompatible.  (For example:
$\langle a \rangle \top$ and $! \emptyset$).  Because CL[$\land, !$] is
expressive enough to be able to make incompatible claims, it satisfies
Brandom's incompatibility semantics constraint.
CL[$\land, !$] is the only logic (we are aware of) with a
quadratic-time decision procedure that is expressive enough to respect
this constraint. 

The bottom three language can all be decided in super-polynomial time.  We
claim that Hennessy-Milner logic is more expressive than PL, and CL[$\land, !, \neg$] is more expressive than full Hennessy-Milner logic.
To see that full Hennessy-Milner logic is more expressive than full
propositional logic, fix a propositional logic with the nullary
operator $\top$ plus an infinite number of propositional atoms
$P_{(i,j)}$, indexed by $i$ and $j$.  Now translate each formula of
Hennessy-Milner logic via the rules:
\begin{align*}
  \SEMB{\top}  & =  \top  &
  \SEMB{\phi \land \psi} & =  \SEMB{\phi} \land \SEMB{\psi}  \\
  \SEMB{\neg \phi} & =  \neg \SEMB{\phi}   &
  \SEMB{\langle a_i \rangle \phi_j} & =  P_{(i,j)} 
\end{align*}

\NI We claim Hennessy-Milner logic is more expressive because there
are formulae $\phi$ and $\psi$ of Hennessy-Milner logic such that
\[
\phi \models_{\text{{HML}}} \psi \mbox{ but } \SEMB{\phi} \nvDash_{\text{PL}} \SEMB{\psi}
\]
For example, let $\phi = \langle a \rangle \langle b \rangle \top$ and
$\psi = \langle a \rangle \top$.  Clearly, $\phi \models_{\text{HML}}
\psi$. But $\SEMB{\phi} = P_{(i,j)}$ and $\SEMB{\psi} = P_{(i',j')}$
for some $i,j,i',j'$, and there are no entailments in propositional
logic between arbitrary propositional atoms.

We close by stating that CL[$\land, !, \neg$] is more expressive than
full Hennessy-Milner logic. As mentioned above, the formula $\fBang A$ of
\cathoristic\ can be translated into Hennessy-Milner logic as:
\[
\bigwedge_{a \in \Sigma - A} \neg \langle a \rangle \top
\]
But if $\Sigma$ is infinite, then this is an infinitary disjunction.
\Cathoristic{} can express the same proposition in a finite sentence.

\subsection{Acknowledgements}


We thank Tom Smith and Giacomo Turbanti for their thoughtful comments.

\bibliography{bib} 

\appendix
\section{Alternative semantics for \cathoristic{}}\label{pureModels}

\NI We use state-labelled transition systems as models for
\cathoristic{}. The purpose of the labels on states is to express
constraints, if any, on outgoing actions. This concern is reflected in
the semantics of $!A$.
\[
\begin{array}{lclcl}
  ((S, \rightarrow, \lambda), s) & \models & !A  &\mbox{\quad iff\quad } & \lambda(s) \subseteq A
\end{array}
\]

\NI There is an alternative, and in some sense even simpler approach
to giving semantics to $!A$ which does not require state-labelling: we
simply check if all actions of all outgoing transitions at the current
state are in $A$.  As the semantics of other formula requires state-labelling in its satisfaction condition, this means we can use plain
labelled transition systems (together with a current state) as
models. This gives rise to a subtly different theory that we now
explore, albeit not in depth.

\subsection{Pure cathoristic models}

\begin{definition}\label{pureModelsDef}
By a \emph{pure cathoristic model}, ranged over by $\PPP, \PPP', ...$,
we mean a pair $(\LLL, s)$ where $\LLL = (S, \rightarrow)$ is a
deterministic labelled transition system and $s \in S$ a state.
\end{definition}

\NI Adapting the satisfaction relation to pure cathoristic models is
straightforward.

\begin{definition}
Using pure cathoristic models, the  \emph{satisfaction relation} is defined 
inductively by the following clauses, where we assume that $\MMM =
(\LLL, s)$ and $\LLL = (S, \rightarrow)$.

\[
\begin{array}{lclcl}
  \MMM & \models & \top   \\
  \MMM & \models & \phi \AND \psi &\ \mbox{ iff } \ & \MMM  \models \phi \mbox { and } \MMM \models \psi  \\
  \MMM & \models & \langle a \rangle \phi & \mbox{ iff } & \text{there is a } s \xrightarrow{a} t \mbox { such that } (\LLL, t) \models \phi  \\
  \MMM & \models & A &\mbox{ iff } & \{a\ |\ \exists t.s \TRANS{a} t \} \subseteq A
\end{array}
\]
\end{definition}

\NI Note that all but the last clause are unchanged from Definition
\ref{ELsatisfaction}.

In this interpretation, $!A$ restricts the out-degree of the current
state $s$, i.e.~it constraints the 'width' of the graph.  It is easy
to see that all rules in Figure \ref{figure:elAndBangRules} are sound
with respect to the new semantics.  The key advantage pure cathoristic
models have is their simplicity: they are unadorned labelled
transition systems, the key model of concurrency theory
\cite{SassoneV:modcontac}. The connection with concurrency theory is
even stronger than that, because, as we show below (Theorem
\ref{hennessymilnertheorem}), the elementary equivalence on (finitely
branching) pure cathoristic models is bisimilarity, one of the more
widely used notions of process equivalence. This characterisation even
holds if we remove the determinacy restriction in Definition \ref{pureModelsDef}.

\subsection{Relationship between pure and cathoristic models}

The obvious way of converting an cathoristic model into a pure cathoristic model
is by forgetting about the state-labelling:
\[
   ((S, \rightarrow, \lambda), s ) \qquad\mapsto\qquad ((S, \rightarrow), s ) 
\]
Let this function be $\FORGET{\cdot}$. For going the other way, we
have two obvious choices:

\begin{itemize}

\item $((S, \rightarrow), s ) \mapsto ((S, \rightarrow, \lambda), s )$
  where $\lambda(t) = \Sigma$ for all states $t$. Call this map $\MAX{\cdot}$.

\item $((S, \rightarrow), s ) \mapsto ((S, \rightarrow, \lambda), s )$
  where $\lambda(t) = \{a \ |\ \exists t'. t \TRANS{a} t'\}$ for all
  states $t$. Call this map $\MIN{\cdot}$.

\end{itemize}

\begin{lemma}\label{modelRelationships}
Let $\MMM$ be an cathoristic model, and $\PPP$ a pure cathoristic model.
\begin{enumerate}

\item\label{modelRelationships:1}  $\MMM \models \phi$ implies
  $\FORGET{\MMM} \models \phi$. The reverse implication does not hold.

\item\label{modelRelationships:2}  $ \MAX{\PPP} \models \phi$ implies
  $\PPP \models \phi$. The reverse implication does not hold.

\item\label{modelRelationships:3} $\MIN{\PPP} \models \phi$ if and only if
  $\PPP \models \phi$. 

\end{enumerate}
\end{lemma}

\begin{proof}
The implication in (\ref{modelRelationships:1}) is immediate by
induction on $\phi$. A counterexample for the reverse implication is
given by the formula $\phi = !\{a\}$ and the cathoristic model $\MMM = ( \{s,
t\}, s \TRANS{a} t, \lambda), s)$ where $\lambda (s) = \{a, b, c\}$:
clearly $\FORGET{\MMM} \models \phi$, but $\MMM \not\models
\phi$.

The implication in (\ref{modelRelationships:2}) is immediate by
induction on $\phi$. To construct a counterexample for the reverse
implication, assume that $\Sigma$ is a strict superset of $\{a\}$
$a$. The formula $\phi = !\{a\}$ and the pure cathoristic model $\PPP = (
\{s, t\}, s \TRANS{a} t ), s)$ satisfy $\PPP \models \phi$, but clearly
$\MAX{\PPP} \not\models \phi$.

Finally, (\ref{modelRelationships:3}) is also straightforward by
induction on $\phi$.

\end{proof}

\subsection{Non-determinism and cathoristic models}

Both, cathoristic models and pure cathoristic models must be deterministic. That
is important for the incompatibility semantics. However, formally, the
definition of satisfaction makes sense for non-deterministic models as
well, pure or otherwise. Such models are important in the theory of
concurrent processes. Many of the theorems of the precious section
either hold directly, or with small modifications for
non-deterministic models. The rules of inference in Figure
\ref{figure:elAndBangRules} are sound except for
    [\RULENAME{Determinism}] which cannot hold in properly
    non-deterministic models. With this omission, they are also
    complete.  Elementary equivalence on non-deterministic cathoristic
    models also coincides with mutual simulation, while elementary
    equivalence on non-deterministic pure cathoristic models is
    bisimilarity. The proofs of both facts follow those of Theorems
    \ref{theorem:completeLattice} and \ref{hennessymilnertheorem},
    respectively. Compactness by translation can be shown following
    the proof in Section \ref{compactness}, except that the constraint
    $\phi_{det}$ is unnecessary.

We have experimented with a version of \cathoristic{} in which the models are
\emph{non-deterministic} labelled-transition systems.  Although
non-determinism makes some of the constructions simpler,
non-deterministic \cathoristic{} is unable to express incompatibility properly.
Consider, for example, the claim that Jack is married\footnote{We assume, in this discussion, that $married$ is a many-to-one predicate. We assume that polygamy is one person \emph{attempting} to marry two people (but failing to marry the second).} to Jill
In standard deterministic \cathoristic{} this would be rendered as:
\begin{eqnarray*}
  \MAY{jack} \MAY{married} (\MAY{jill} \land \fBang \{jill\})
\end{eqnarray*}
There are three levels at which this claim can be denied.
First, we can claim that Jack is married to someone else - Joan, say:
\begin{eqnarray*}
   \MAY{jack} \MAY{married} (\MAY{joan} \land \fBang \{joan\})
\end{eqnarray*}
Second, we can claim that Jack is unmarried (specifically, that being unmarried is Jack's only property):
\begin{eqnarray*}
  \MAY{jack} !\{unmarried\}
\end{eqnarray*}
Third, we can claim that Jack does not exist at all. Bob and Jill, for example, are the only people in our domain:
\begin{eqnarray*}
  \fBang \{bob, jill\}
\end{eqnarray*}

Now we can assert the same sentences in non-deterministic \cathoristic{}, but they
are \emph{no longer incompatible with our original sentence}.  In
non-deterministic \cathoristic{}, the following sentences are compatible (as long
as there are two separate transitions labelled with $married$, or two
separate transitions labelled with $jack$):
\begin{eqnarray*}
  \MAY{jack} \MAY{married} (\MAY{jill} \land \fBang \{jill\}) 
      \qquad
  \MAY{jack} \MAY{married} (\MAY{joan} \land \fBang \{joan\})
\end{eqnarray*}

\NI Similarly, the following sentences are fully compatible as long as
there are two separate transitions labelled with $jack$:
\begin{eqnarray*}
  \MAY{jack} \MAY{married}
     \qquad
  \MAY{jack} \fBang \{unmarried\}
\end{eqnarray*}

\NI Relatedly, non-deterministic \cathoristic{} does not satisfy Brandom's
incompatibility semantics property:
\[
   \phi \models \psi \; \mbox{ iff } \; \mathcal{I}(\psi) \subseteq \mathcal{I}(\phi)
\]

\NI To take a simple counter-example, $\MAY{a}\MAY{b}$ implies $\MAY{a}$,
but not conversely.  But in non-deterministic \cathoristic{}, the set of sentences
incompatible with $\MAY{a}\MAY{b}$ is identical with the set of
sentences incompatible with $\MAY{a}$.

\subsection{Semantic characterisation of elementary equivalence}

In Section \ref{elementaryEquivalence} we presented a semantic
analysis of elementary equivalence, culminating in Theorem
\ref{theorem:completeLattice} which showed that elementary equivalence
coincides with $\MODELEQ$, the relation of mutual simulation of
models. We shall now carry out a similar analysis for pure cathoristic
models, and show that elementary equivalence coincides with
bisimilarity, an important concept in process theory and modal logics
\cite{SangiorgiD:intbisac}. Bisimilarity is strictly finer on
non-deterministic transition systems than $\MODELEQ$, and more
sensitive to branching structure.  In the rest of this section, we
allow non-deterministic pure models, because the characterisation is
more interesting that in the deterministic case.

\begin{definition}
A pure cathoristic model $(\LLL, s)$ is finitely branching if its
underlying transition system $\LLL$ is finitely branching.
\end{definition}

\begin{definition}
A binary relation $\RRR$ is a \emph{bisimulation} between pure cathoristic
models $\PPP_i = (\LLL_i), s_i)$ for $i = 1, 2$ provided (1) $\RRR$ is
a bisimulation between $\LLL_1$ and $\LLL_2$, and (2) $(s_1, s_2) \in
\RRR$. We say $\PPP_1$ and $\PPP_2$ are \emph{bisimilar}, written
$\PPP_1 \BISIM \PPP_2$ if there is a bisimulation between $\PPP_1$ and
$\PPP_2$.
\end{definition}

\begin{definition}
The \emph{theory} of $\PPP$, written $\THEORY{\PPP}$, is the set
$\{\phi\ |\ \PPP \models \phi\}$.
\end{definition}

\begin{theorem}
\label{hennessymilnertheorem}
Let $\PPP$ and $\PPP'$ be two finitely branching pure cathoristic
models. Then: $\PPP \BISIM \PPP'$ if and only if $\THEORY{\PPP} =
\THEORY{\PPP'}$. 
\end{theorem}

\begin{proof}
\NI Let $\PPP = (\LLL, w)$ and $\PPP' = (\LLL', w')$ be finitely
branching, where $\LLL = (W, \rightarrow)$ and $(W', \rightarrow')$.
We first show the left to right direction, so assume that $\PPP \BISIM
\PPP'$.

The proof is by induction on formulae.  The only case which differs
from the standard Hennessy-Milner theorem is the case for $!A$, so
this is the only case we shall consider.  Assume $w \BISIM w'$ and $w
\models !A$. We need to show $w' \models !A$.
From the semantic clause for $!$, $w \models !A$ implies $\lambda(w)
\subseteq A$.  If $w \BISIM w'$, then $\lambda(w) = \lambda'(w')$.
Therefore $\lambda'(w') \subseteq A$, and hence $w' \models !A$.

The proof for the other direction is more involved.
For states $x \in W$ and $x' \in W$, we write 
\[
   x \equiv x'
      \qquad\text{iff}\qquad
   \THEORY{(\LLL, x)} = \THEORY{(\LLL', x')}.
\]

We define the bisimilarity relation:
\[
   Z = \{(x,x') \in \mathcal{W} \times \mathcal{W}' \fOr x \equiv x' \}
\]
To prove $w \BISIM w'$, we need to show:
\begin{itemize}

\item $(w,w') \in Z$. This is immediate from the definition of Z.

\item The relation $Z$ respects the transition-restrictions: if
  $(x,x') \in Z$ then $\lambda(x) = \lambda'(x')$

\item The forth condition: if $(x,x') \in Z$ and $x \xrightarrow{a}
  y$, then there exists a $y'$ such that $x' \xrightarrow{a} y'$ and $(y, y') \in Z$.

\item The back condition: if $(x,x') \in Z$ and $x' \xrightarrow{a}
  y'$, then there exists a $y$ such that $x \xrightarrow{a} y$ and $(y, y') \in Z$.

\end{itemize}
To show that $(x,x') \in Z$ implies $\lambda(x) = \lambda'(x')$, we
will argue by contraposition.  Assume $\lambda(x) \neq \lambda'(x')$.
Then either $\lambda'(x') \nsubseteq \lambda(x)$ or $\lambda(x)
\nsubseteq \lambda'(x')$.  If $\lambda'(x') \nsubseteq \lambda(x)$,
then $x' \nvDash \fBang \lambda(x)$.  But $x \models \fBang
\lambda(x)$, so $x$ and $x'$ satisfy different sets of propositions
and are not equivalent.  Similarly, if $\lambda(x) \nsubseteq
\lambda'(x')$ then $x \nvDash \fBang \lambda'(x')$.  But $x' \models
\fBang \lambda'(x')$, so again $x$ and $x'$ satisfy different sets of
propositions and are not equivalent.

We will show the forth condition in detail. The back condition is very
similar.  To show the forth condition, assume that $x \xrightarrow{a}
y$ and that $(x,x') \in Z$ (i.e. $x \equiv x'$).  We need to show that
$\exists y'$ such that $x' \xrightarrow{a} y'$ and $(y,y') \in Z$
(i.e. $y \equiv y'$).

Consider the set of $y'_i$ such that $x' \xrightarrow{a} y'_i$. Since
$x \xrightarrow{a} y$, $x \models \langle a \rangle \top$, and as $x
\equiv x'$, $x' \models \langle a \rangle \top$, so we know this set
is non-empty.  Further, since $(\mathcal{W}', \rightarrow')$ is
finitely-branching, there is only a finite set of such $y'_i$, so we
can list them $y'_1, ..., y'_n$, where $n >= 1$.

Now, in the Hennessy-Milner theorem for Hennessy-Milner logic, the proof proceeds as
follows: assume, for reductio, that of the $y'_1, ..., y'_n$, there is
no $y'_i$ such that $y \equiv y'_i$.  Then, by the definition of
$\equiv$, there must be formulae $\phi_1, ..., \phi_n$ such that for
all $i$ in $1$ to $n$:
\[
y'_i \models \phi_i \mbox{ and } y \nvDash \phi_i
\]
Now consider the formula:
\[
[a] (\phi_1 \lor ... \lor \phi_n)
\]
As each $y'_i \models \phi_i$, $x' \models [a] (\phi_1 \lor ... \lor \phi_n)$, but $x$ does not satisfy this formula, as each $\phi_i$ is not satisfied at $y$.
Since there is a formula which $x$ and $x'$ do not agree on, $x$ and $x'$ are not equivalent, contradicting our initial assumption.

But this proof cannot be used in \cathoristic{} because it relies on a formula $[a] (\phi_1 \lor ... \lor \phi_n)$ which cannot be expressed in \cathoristic{}: 
\Cathoristic{} does not include the box operator or disjunction, so this formula is ruled out on two accounts.
But we can massage it into a form which is more amenable to \cathoristic{}'s expressive resources:
\begin{eqnarray*}
[a] (\phi_1 \lor ... \lor \phi_n) & = & \neg \langle a \rangle \neg (\phi_1 \lor ... \lor \phi_n)  \\
	& = & \neg \langle a \rangle (\neg \phi_1\AND ... \AND \neg \phi_n) 
\end{eqnarray*}
Further, if the original formula $[a] (\phi_1 \lor ... \lor \phi_n)$ is true in $x'$ but not in $x$, then its negation will be true in $x$ but not in $x'$. 
So we have the following formula, true in $x$ but not in $x'$:
\[
 \langle a \rangle (\neg \phi_1\AND ... \AND \neg \phi_n)
 \]
The reason for massaging the formula in this way is so we can express it in \cathoristic{} (which does not have the box operator or disjunction).
At this moment, the revised formula is \emph{still} outside \cathoristic{} because it uses negation. 
But we are almost there: the remaining negation is in innermost scope, and innermost scope negation can be simulated in \cathoristic{} by the $!$ operator. 

We are assuming, for reductio, that of the $y'_1, ..., y'_n$, there is no $y'_i$ such that $y \equiv y'_i$.
But in \cathoristic{} without negation, we cannot assume that each $y'_i$ has a formula $\phi_i$ which is satisfied by $y'_i$ but not by $y$ - it might instead be the other way round: $\phi_i$ may be satisfied by $y$ but not by $y'_i$. So, without loss of generality, assume that $y'_1, ..., y'_m$ fail to satisfy formulae $\phi_1, ..., \phi_m$ which $y$ does satisfy, and that $y'_{m+1}, ..., y'_n$ satisfy formulae $\phi_{m+1}, ..., \phi_n$ which $y$ does not:
\begin{eqnarray*}
y \models \phi_i \mbox{ and } y'_i \nvDash \phi_i & & i = 1 \mbox{ to } m  \\
y \nvDash \phi_j \mbox{ and } y'_j \models \phi_j & & j = m+1 \mbox{ to } n 
\end{eqnarray*}
The formula we will use to distinguish between $x$ and $x'$ is:
\[
 \langle a \rangle ( \bigwedge_{i=1}^m \phi_i \; \AND \; \bigwedge_{j=m+1}^n \mathsf{neg}(y, \phi_j))
 \]
 Here, $\mathsf{neg}$ is a meta-language function that, given a state y and a formula $\phi_j$, returns a formula that is true in $y$ but incompatible with $\phi_j$. We will show that, since $y \nvDash \phi_j$, it is always possible to construct $ \mathsf{neg}(y, \phi_j)$ using the $!$ operator.

Consider the possible forms of $\phi_j$:
\begin{itemize}
\item
$\top$: this case cannot occur since all models satisfy $\top$.
\item
$\phi_1 \AND \phi_2$: we know $y'_j \models \phi_1 \AND \phi_2$ and $y \nvDash \phi_1 \AND \phi_2$. There are three possibilities:
\begin{enumerate}
\item
$y \nvDash \phi_1$ and $y \models \phi_2$. In this case, $\mathsf{neg}(y, \phi_1 \AND \phi_2) = \mathsf{neg}(y, \phi_1) \AND \phi_2$.
\item
$y \models \phi_1$ and $y \nvDash \phi_2$. In this case, $\mathsf{neg}(y, \phi_1 \AND \phi_2) = \phi_1 \AND \mathsf{neg}(y, \phi_2)$.
\item
$y \nvDash \phi_1$ and $y \nvDash \phi_2$. In this case, $\mathsf{neg}(y, \phi_1 \AND \phi_2) =  \mathsf{neg}(y, \phi_1) \AND \mathsf{neg}(y, \phi_2)$.
\end{enumerate}
\item
$!A$: if $y \nvDash !A \mbox{ and } y'_j \models !A$, then there is an action $a \in \Sigma-A$ such that $y \xrightarrow{a} z$ for some $z$ but there is no such $z$ such that $y'_j \xrightarrow{a} z$. In this case, let $\mathsf{neg}(y, \phi_j) = \langle a \rangle \top$.
\item
$\langle a \rangle \phi$. There are two possibilities:
\begin{enumerate}
\item
$y \models \langle a \rangle \top$. In this case, $\mathsf{neg}(y, \langle a \rangle \phi) =  \bigwedge\limits_{y \xrightarrow{a} z}  \langle a \rangle \mathsf{neg}(z, \phi)$.
\item
$y \nvDash \langle a \rangle \top$. In this case, $\mathsf{neg}(y, \langle a \rangle \phi) = \fBang \{ b \fOr \exists z. y \xrightarrow{b} z\}$. This set of $b$s is finite since we are assuming the transition system  is finitely-branching.
\end{enumerate}
\end{itemize}

\end{proof}

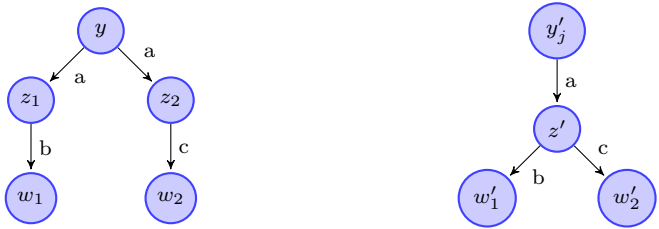
\begin{figure}[h]
\centering
\begin{tikzpicture}[node distance=1.3cm,>=stealth',bend angle=45,auto]
  \tikzstyle{place}=[circle,thick,draw=blue!75,fill=blue!20,minimum size=6mm]
  \tikzstyle{red place}=[place,draw=red!75,fill=red!20]
  \tikzstyle{transition}=[rectangle,thick,draw=black!75,
  			  fill=black!20,minimum size=4mm]
  \tikzstyle{every label}=[red]

  \begin{scope}[xshift=0cm]
    \node [place] (t) {$y$};
    \node [place] (a) [below left of=t] {$z_1$}
      edge [pre]  node[swap] {a}                 (t);
    \node [place] (a2) [below right of=t] {$z_2$}
      edge [pre]  node[swap] {a}                 (t);
    \node [place] (b) [below of=a] {$w_1$}
      edge [pre]  node[swap] {b}                 (a);
    \node [place] (c) [below of=a2] {$w_2$}
      edge [pre]  node[swap] {c}                 (a2);
  \end{scope}  
  \begin{scope}[xshift=6cm]
    \node [place] (t) {$y'_j$};
    \node [place] (a) [below of=t] {$z'$}
      edge [pre]  node[swap] {a}                 (t);
    \node [place] (b) [below left of=a] {$w'_1$}
      edge [pre]  node[swap] {b}                 (a);
    \node [place] (c) [below right of=a] {$w'_2$}
      edge [pre]  node[swap] {c}                 (a);
  \end{scope}  
\end{tikzpicture}
\caption{{\small Worked example of $\mathsf{neg}$. Note that the transition
  system on the left is non-deterministic.}}\label{figure:example:neg}
\end{figure}

\NI We continue with a worked example of $\mathsf{neg}$.  Consider
  $y$ and $y'_j$ as in Figure \ref{figure:example:neg}.  One formula
  that is true in $y'_j$ but not in $y$ is

\[
   \langle a \rangle (\langle b \rangle \top \AND \langle c \rangle \top)
\]

\NI Now:

\begin{eqnarray*}
\lefteqn{\mathsf{neg}(y, \langle a \rangle (\langle b \rangle \top \AND \langle c \rangle \top))}\qquad \qquad \qquad  \\
& = & \bigwedge\limits_{y \xrightarrow{a} z} \langle a \rangle \mathsf{neg}(z, \langle b \rangle \top \AND \langle c \rangle \top)  \\
& = & \langle a \rangle \mathsf{neg}(z_1, \langle b \rangle \top \AND \langle c \rangle \top) \AND \langle a \rangle\mathsf{neg}(z_2, \langle b \rangle \top \AND \langle c \rangle \top)  \\
& = & \langle a \rangle (\langle b \rangle \top \AND \mathsf{neg}(z_1, \langle c \rangle \top)) \AND \langle a \rangle\mathsf{neg}(z_2, \langle b \rangle \top \AND \langle c \rangle \top)  \\
& = & \langle a \rangle (\langle b \rangle \top \AND \mathsf{neg}(z_1, \langle c \rangle \top)) \AND \langle a \rangle(\mathsf{neg}(z_2, \langle b \rangle \top) \AND \langle c \rangle \top)  \\
& = & \langle a \rangle (\langle b \rangle \top \AND \fBang \{b\}) \AND \langle a \rangle(\mathsf{neg}(z_2, \langle b \rangle \top) \AND \langle c \rangle \top)  \\
& = & \langle a \rangle (\langle b \rangle \top \AND \fBang \{b\}) \AND \langle a \rangle(\fBang \{c\} \AND \langle c \rangle \top) 
\end{eqnarray*}

\NI The resulting formula is true in $y$ but not in $y'_j$.

\section{Omitted proofs}\label{app:completeness:proofs}\label{inc-appendix}

\subsection{Proof of Lemma \ref{lemmasimpl}}\label{app:decision:proofs}

If $\MMM \models \phi$ then $\MMM \MODELLEQ \SIMPL{\phi}$.

\begin{proof}
We shall show $\THEORY{\SIMPL{\phi}} \subseteq \THEORY{\MMM}$.
The desired result will then follow by applying Theorem \ref{theorem:completeLattice}.
We shall show that
\[
\text{If } \MMM \models \phi \text{ then } \THEORY{\SIMPL{\phi}} \subseteq \THEORY{\MMM}
\]
by induction on $\phi$.
In all the cases below, let $\SIMPL{\phi} = (\mathcal{L}, w)$ and let $\MMM = (\mathcal{L}', w')$.
The case where $\phi = \top$ is trivial.
Next, assume $\phi = \MAY{a} \psi$.
We know $\MMM \models \MAY{a} \psi$ and need to show that $\THEORY{\SIMPL{\MAY{a} \psi}} \subseteq \THEORY{\MMM}$.
Since $(\mathcal{L}', w') \models \MAY{a} \psi$, there is an $x'$ such that $w' \xrightarrow{a} x'$ and $(\mathcal{L}', x') \models \psi$.
Now from the definition of $\SIMPL{}$, $\SIMPL{\MAY{a} \psi}$ is a model combining $\SIMPL{\psi}$ with a new state $w$ not appearing in $\SIMPL{\psi}$ with an arrow $w \xrightarrow{a} x$ (where $x$ is the start state in $\SIMPL{\psi}$), and $\lambda(w) = \Sigma$. 
Consider any sentence $\xi$ such that $\SIMPL{\MAY{a} \psi} \models \xi$. Given the construction of $\SIMPL{\MAY{a}\psi}$, $\xi$ must be a conjunction of $\top$ and formulae of the form $\MAY{a} \tau$. In the first case, $(\mathcal{L}', x')$ satisfies $\top$; in the second case, $(\mathcal{L}', x') \models \tau$ by the induction hypothesis and hence $(\mathcal{L}', w') \models \MAY{a} \tau$.

Next, consider the case where $\phi = !A$, for some finite set $A \subset \Sigma$.
From the definition of $\SIMPL{}$, $\SIMPL{!A}$ is a model with one state $s$, no transitions, with $\lambda(s) = A$.
Now the only formulae that are true in $\SIMPL{!A}$ are conjunctions of $\top$ and $!B$, for supersets $B \supseteq A$.
If $\MMM \models !A$ then by the semantic clause for $!$, $\lambda'(w') \subseteq A$, hence $\MMM$ models all the formulae that are true in $\SIMPL{!A}$.

Finally, consider the case where $\phi = \psi_1 \land \psi_2$.
Assume $\MMM \models \psi_1$ and $\MMM \models \psi_2$.
We assume, by the induction hypothesis that $\THEORY{\SIMPL{\psi_1}} \subseteq \THEORY{\MMM}$ and $\THEORY{\SIMPL{\psi_2}} \subseteq \THEORY{\MMM}$.
We need to show that $\THEORY{\SIMPL{\psi_1\land \psi_2}} \subseteq \THEORY{\MMM}$.
By the definition of $\SIMPL{}$, $\SIMPL{\psi_1 \land \psi_2} = \SIMPL{\psi_1} \sqcap \SIMPL{\psi_2}$.
If $\SIMPL{\psi_1}$ and $\SIMPL{\psi_2}$ are $\mathsf{inconsistent}$ (see the definition of $\mathsf{inconsistent}$ in Section \ref{simpl}) then $\MMM = \bot$. In this case, $\THEORY{\SIMPL{\psi_1} \land \SIMPL{\psi_2}} \subseteq \THEORY{\bot}$.
If, on the other hand, $\SIMPL{\psi_1}$ and $\SIMPL{\psi_2}$ are not $\mathsf{inconsistent}$, we shall show that $\THEORY{\SIMPL{\psi_1 \land \psi_2}} \subseteq \THEORY{\MMM}$ by reductio.
Assume a formula $\xi$ such that $\SIMPL{\psi_1 \land \psi_2} \models \xi$ but $\MMM \nvDash \xi$.
Now $\xi \neq \top$ because all models satisfy $\top$.
$\xi$ cannot be of the form $\MAY{a} \tau$ because, by the construction of $\mathsf{merge}$ (see Section \ref{simpl}), all transitions in $\SIMPL{\psi_1 \land \psi_2}$ are transitions from $\SIMPL{\psi_1}$ or $\SIMPL{\psi_2}$ and we know from the inductive hypothesis that $\THEORY{\SIMPL{\psi_1}} \subseteq \THEORY{\MMM}$ and $\THEORY{\SIMPL{\psi_2}} \subseteq \THEORY{\MMM}$.
$\xi$ cannot be $!A$ for some $A \subset \Sigma$, because, from the construction of $\mathsf{merge}$, all state-labellings in $\SIMPL{\psi_1 \land \psi_2}$ are no more specific than the corresponding state-labellings in $\SIMPL{\psi_1}$ and $\SIMPL{\psi_2}$, and we know from the inductive hypothesis that $\THEORY{\SIMPL{\psi_1}} \subseteq \THEORY{\MMM}$ and $\THEORY{\SIMPL{\psi_2}} \subseteq \THEORY{\MMM}$.
Finally, $\xi$ cannot be $\xi_1 \land xi_2$ because the same argument applies to $xi_1$ and $xi_2$ individually.
We have exhausted the possible forms of $\xi$, so conclude that there is no formula $\xi$ such that $\SIMPL{\psi_1 \land \psi_2} \models \xi$ but $\MMM \nvDash \xi$.
Hence $\THEORY{\SIMPL{\psi_1\land \psi_2}} \subseteq \THEORY{\MMM}$.
\end{proof}

\subsection{Proof of Lemma \ref{inc1}}
If $\phi \models \psi \mbox{ then } \SIMPL{\phi} \MODELLEQ \SIMPL{\psi}$

\begin{proof}
By Theorem \ref{theorem:completeLattice}, $ \SIMPL{\phi} \MODELLEQ \SIMPL{\psi}$ iff $\THEORY{\SIMPL{\psi}} \subseteq  \THEORY{\SIMPL{\phi}}$.
Assume $\phi \models \psi$, and assume $\xi \in \THEORY{\SIMPL{\psi}} $. We must show $\xi \in \THEORY{\SIMPL{\phi}} $.
Now $\SIMPL$ is constructed so that:
\[
\SIMPL{\psi} = \bigsqcup \{ \MMM \; | \; \MMM \models \psi \}
\]
So  $\xi \in \THEORY{\SIMPL{\psi}} $ iff for all models $\MMM$, $\MMM \models \psi$ implies $\MMM \models \xi$.
We must show that $\MMM \models \phi$ implies $\MMM \models \xi$ for all models $\MMM$.
Assume $\MMM \models \phi$. Then since $\phi \models \psi$,  $\MMM \models \psi$. 
But since $\xi \in \THEORY{\SIMPL{\psi}} $, $\MMM \models \xi$ also.

\end{proof}

\subsection{Proof of Lemma \ref{inc3}}
If $\mathcal{I}(\psi) \subseteq \mathcal{I}(\phi) \mbox{ then } \mathcal{J}(\SIMPL{\psi}) \subseteq \mathcal{J}(\SIMPL{\phi})$
\begin{proof}
Assume $\mathcal{I}(\psi) \subseteq \mathcal{I}(\phi)$ and $\MMM \sqcap \SIMPL{\psi} = \bot$.
We need to show $\MMM \sqcap \SIMPL{\phi} = \bot$.
If $\mathcal{I}(\psi) \subseteq \mathcal{I}(\phi)$ then for all formulae $\xi$, if $\SIMPL{\xi} \sqcap \SIMPL{\psi} = \bot$ then $\SIMPL{\xi} \sqcap \SIMPL{\phi} = \bot$.
Let $\xi$ be $\CHAR{\MMM}$.
Given that $\MMM \sqcap \SIMPL{\psi} = \bot$ and $\SIMPL{\CHAR{\MMM}} \MODELLEQ \MMM$, $\SIMPL{\CHAR{\MMM}} \sqcap \SIMPL{\psi} = \bot$.
Then as $\mathcal{I}(\psi) \subseteq \mathcal{I}(\phi)$, $\SIMPL{\CHAR{\MMM}} \sqcap \SIMPL{\phi} = \bot$.
Now as $\MMM  \MODELLEQ \SIMPL{\CHAR{\MMM}}$, $\MMM \sqcap \SIMPL{\phi} = \bot$.

\end{proof}

\end{document}